\documentclass[a4paper,11pt]{article}

\pdfoutput=1

\usepackage{amsmath,amssymb,mathtools,amsthm,mathrsfs}
\usepackage{array}
\usepackage{bm}
\usepackage{braket}
\usepackage{amsmath, amssymb, amsthm}
\usepackage[table]{xcolor}

\usepackage{enumitem}

\usepackage{physics}

\usepackage{mathtools}
\usepackage{subcaption}

\newtheorem{theorem}{Theorem}

\newtheorem{lemma}{Lemma}

\newtheorem{corollary}{Corollary}
\newtheorem{definition}{Definition}

\usepackage{booktabs} 
\usepackage{array}    
\usepackage{amsmath}  
\usepackage{xcolor}   
\usepackage{colortbl} 
\usepackage{color}
\usepackage{graphicx}
\usepackage{cite}
\usepackage[colorlinks=true,linkcolor=blue, citecolor=red, urlcolor=blue, bookmarks]{hyperref}
\usepackage[figuresright]{rotating} 
\usepackage{textcomp}
\usepackage{wasysym}
\usepackage{diagbox}
\usepackage{ulem}

\usepackage[utf8]{inputenc}
\usepackage[T1]{fontenc}

\usepackage[text={17.1cm,24.6cm},centering]{geometry} 

\newcommand{\pf}{\operatorname{pf}}
\newcommand{\haf}{\operatorname{haf}}

\newcommand{\R}{\mathbb{R}}

\numberwithin{equation}{section}

\def \be {\begin{equation}}
\def \ee {\end{equation}}
\def \ba {\begin{array}}
\def \ea {\end{array}}
\def \bea{\begin{eqnarray}}
\def \eea{\end{eqnarray}}

\def \and {{~\textrm{and}~}}

\begin{document}
\sloppy

\newpage
\setcounter{page}{1}

\begin{center}{\Large \textbf{A Generalized Grassmann–Pfaffian Framework for Monomer--Dimer and Spanning Trees\\
}}\end{center}

\begin{center}
E. A. Ramirez Trino\textsuperscript{1}
M. A. Seifi MirJafarlou\textsuperscript{1} and
M. A. Rajabpour\textsuperscript{1},

\end{center}

\begin{center}
{1} Universidade Federal Fluminense, Niterói - RJ, Brazil
\end{center}


\begin{abstract}
We develop a unified framework for Berezin integrals over Grassmann variables that establishes master identities for exponential quadratic fermionic forms and linear fermionic forms coupled to both bosonic and fermionic sources. The construction is rigorous for both real and complex fermions in arbitrary dimensions and remains well-defined even when the underlying matrices are singular. Our main mathematical results appear in two master theorems. Theorem \ref{exponential with linear fermionic and bosonic sources} provides a comprehensive identity for Berezin integrals over Grassmann variables for real fermions with mixed bosonic–fermionic sources, applicable to any antisymmetric matrix. Its complex analogue, Theorem \ref{theorem 2.2}, yields corresponding determinant-based representations. Together, they serve as generating functionals for a wide range of combinatorial and physical models. Key applications include the dimer, monomer–dimer, matching, and almost-matching problems. We revisit the Kasteleyn theorem for planar dimers using Berezin integrals. We construct monomer–dimer systems through the \textit{Monobisyzexant (Mbsz)} function, which generalizes the Hafnian to incorporate monomer contributions and admits a Pfaffian-sum representation for planar graphs (Theorem \ref{thm:main2}); and practical techniques for handling singular matrices via unitary block decomposition (Theorem \ref{HuaTheorem}) and spectral analysis. We further present explicit mappings between Hafnians and Pfaffians and their submatrix generalizations (Hafnianinhos and Pfaffianinhos); an alternative source-ordered Berezin integral representation for spanning trees and forests using complex bosonic sources that regularizes the Laplacian zero mode (Theorem~\ref{TheoST}). Overall, this work offers a flexible toolkit for the theoretical analysis and computational implementation of graph-based models and lattice field theories using Berezin integrals over Grassmann variables .
\end{abstract}

\vspace{10pt}
\noindent\rule{\textwidth}{1pt}
\tableofcontents
\noindent\rule{\textwidth}{1pt}
\vspace{10pt}

\section{Introduction}
A graph $G$ consists of a set of points, called \textit{vertices}, and a set of lines, called \textit{edges}, connecting specified pairs of vertices. A \textit{dimer}, or \textit{matching}, is an object that occupies a single edge along with its endpoints. A \textit{dimer covering} of $G$, also known as a \textit{perfect matching}, is a collection of dimers such that each vertex is covered by exactly one dimer~\cite{LovaszPlummer1986}. Enumerating the perfect matchings of $G$ is equivalent to computing the partition function of the dimer model on the graph. This model provides a unifying framework that connects combinatorics, statistical mechanics, and fermionic field theory via Berezin integrals over Grassmann variables~\cite{Kasteleyn1961,hayn1997}.

From a combinatorial perspective, all the information of a graph $G$ is encoded in its adjacency matrix $\mathbf{A}$ (also called connectivity matrix). The natural algebraic object that enumerates perfect matchings is the \textit{Hafnian} of the adjacency matrix~\cite{Barvinok1999}. More generally, for an even dimensional symmetric matrix, the Hafnian (Definition~\ref{Def:Hafnian}) counts all possible pairwise partitions of indices, each weighted by the corresponding matrix elements. The term “Hafnian” was introduced by Caianiello~\cite{Caianiello1973}, inspired by his stay in Copenhagen (Hafnia in Latin). Although fundamental for dimer counting, evaluating the Hafnian is \#P-complete, making exact computation intractable for general graphs~\cite{Barvinok2016a,Rudelson2016,Chien2004,Sankowski2003}. Nevertheless, several algorithmic advances have been developed; notably, a parallelizable algorithm for Hafnians of complex weighted graphs was proposed in~\cite{Quesada2019}, with computational complexity $\mathcal{O}(L^{3} 2^{L/2})$, which remains the fastest exact method currently known.

The enumeration of perfect matchings exhibits a fundamental computational dichotomy between planar and non-planar graphs. For planar graphs, Kasteleyn's celebrated theorem~\cite{Kasteleyn1961,Kasteleyn1963} provides a profound simplification: by constructing a special edge orientation where each face has an odd number of clockwise-oriented edges, the dimer partition function becomes exactly computable as the Pfaffian of an antisymmetric Kasteleyn matrix $\mathbf{K}$, with Temperley and Fisher independently deriving analogous results for square lattices~\cite{TemperleyFisher1961,Fisher1961}. This methodology extends substantially beyond planar graphs: for bipartite graphs, Little's theorem completely characterizes Pfaffian orientability through forbidden $K_{3,3}$ subdivisions~\cite{Little1973,Little1975}, with polynomial-time recognition algorithms~\cite{Robertson1999,McCuaig2004}. For arbitrary graphs embedded on orientable surfaces of genus $g$, Kasteleyn's original framework was systematically generalized by~\cite{CimasoniReshetikhin2007,CimasoniReshetikhin2008}, which demonstrated that the partition function becomes a carefully signed linear combination of $4^g$ Pfaffians corresponding to different spin structures, naturally reducing to the single Pfaffian case when $g=0$. This rich mathematical edifice admits an elegant Grassmann integral formulation where the dimer model maps to a free fermion theory, with exact solvability in polynomial time for fixed genus analogous to Onsager's seminal solution of the 2D Ising model~\cite{Onsager1944}. In this work, we focus exclusively on the planar case, providing a self-contained treatment of Kasteleyn's method and its Grassmann formulation while acknowledging these broader mathematical contexts.

Beyond combinatorics, dimers are paradigmatic physical models. They describe diatomic molecules adsorbed on crystal surfaces \cite{Kasteleyn1961,Fisher1961}, spin-singlet valence-bond states in quantum antiferromagnets \cite{Anderson1973}, RVB states relevant to high-$T_c$ superconductors \cite{Anderson1987}, and frustrated magnetic systems or spin liquids \cite{Moessner2001,Misguich2002}. The Pfaffian structure highlights their fermionic nature, naturally connecting to quadratic fermionic actions and motivating Berezin integrals over Grassmann variables representations for dimer configurations \cite{Samuel1980,ZinnJustin2002,Kasteleyn1963}.

The \textit{monomer--dimer model} arises as a natural generalization of the pure dimer model, allowing for configurations that include unpaired vertices (\textit{monomers}). A \textit{monomer--dimer covering} of a graph \(G\) is a collection of monomers and dimers such that each vertex is occupied by exactly one of these objects, either a monomer or a dimer. This formulation was first introduced in the early 1960s by Fisher and Stephenson~\cite{FisherStephenson1963} to study defect correlations and the effects of vacancies in close-packed dimer systems.

In the pure dimer case, Temperley and Fisher~\cite{TemperleyFisher1961} and Kasteleyn~\cite{Kasteleyn1963} showed that the partition function on planar graphs can be expressed as a Pfaffian, enabling efficient exact computation. However, including monomers dramatically increases the combinatorial complexity, making the enumeration of monomer--dimer coverings \#P-complete~\cite{Jerrum1987}. Consequently, a general Pfaffian representation for the monomer--dimer problem is impossible. 

Nevertheless, several important exceptions arise when monomers are constrained to specific locations. For instance, \cite{TW03} and~\cite{Wu2006} derived a Pfaffian formula for the partition function of a system with a single monomer located on the boundary of a finite square lattice, using the Temperley bijection~\cite{Te74}. This result was extended by ~\cite{Wu_2011} to cylinders of odd width, which are non-bipartite lattices. \cite{Priezzhev2008} further considered the half-plane square lattice with monomers fixed on the boundary, deriving a Pfaffian formula to compute multipoint boundary monomer correlations. \cite{AllegraFortin2014} showed that when monomers are fixed at arbitrary positions in a square lattice, the partition function can be expressed as a product of two Pfaffians. Building on these ideas, \cite{GiulianiJauslinLieb2016} demonstrated that a Pfaffian formula is achievable for arbitrary planar graphs whenever all monomers are constrained to the boundary.

Berezin integral over Grassmann variables formulations have also proven valuable, providing a compact framework to study monomer--monomer correlations, which decay exponentially at finite densities~\cite{Krauth2003,Papanikolaou2007,AllegraFortin2014}. These approaches unify combinatorial insights with field-theoretic methods, allowing analytical and numerical exploration of otherwise intractable monomer--dimer systems.

From a combinatorial perspective, for graphs $G$ that may include loops, one can define a function encoding the full monomer–dimer partition function through the Laplacian of $G$, which we denote as the \textit{Monobisyzexant}\footnote{The term \textit{Monobisyzexant} derives from the Greek \textit{mono} (single), \textit{bi} (pair), and \textit{syzexis} (union or pairing).} (Mbsz). This construction generalizes the standard monomer–dimer model by extending the Hafnian to symmetric matrices with diagonal entries that account for loops. Unlike the Hafnian, the Mbsz naturally incorporates these loop contributions (single pair or monomer) and is well-defined in arbitrary dimensions. When the diagonal entries vanish, it reduces to the usual Hafnian, while in even dimensions it coincides with the \textit{loop-Hafnian}~\cite{loophafnian2019faster}, which is itself defined in terms of the graph Laplacian. A precise definition and detailed construction are provided in Section~\ref{Sec:2}.

Since the Hafnian admits a fermionic representation via Berezin integrals~\cite{Allegra2015tesis,AllegraFortin2014,Papanikolaou2007}, the Mbsz naturally inherits this structure, providing a unified Grassmann-Berezin framework that accommodates both paired and unpaired sites~\cite{Papanikolaou2007,AllegraFortin2014}. Early Grassmann formulations of the dimer and monomer–dimer models were introduced by Samuel~\cite{Samuel1980}, and subsequently developed in~\cite{hayn1997,Barbaro1997,Plechko1985}. In this approach, each vertex is associated with a pair of Grassmann variables to enforce single occupancy, resulting in a fermionic integral over a quartic action that encodes the standard dimer model. By including an additional quadratic term, one obtains the Grassmann–Berezin representation of the monomer–dimer partition function, which is amenable to diagrammatic expansions and analytical manipulations.

For odd $L$, no perfect matching exists because the Pfaffian of an odd-dimensional skew-symmetric matrix vanishes, leading necessarily to at least one monomer~\cite{Tutte1947,Lovasz1972,Gallai1963}. Specific graph families—including star graphs, unbalanced bipartite graphs, or graphs containing many odd components—can feature multiple monomers, with a number that scales linearly with system size~\cite{HeilmannLieb1972}. To handle both even- and odd-dimensional cases in a unified way, we define a generalized partition function (Theorem~\ref{exponential with linear fermionic and bosonic sources}) that does not require the Kasteleyn matrix be non singular. This construction introduces a fermionic term coupled to bosonic and fermionic source fields,
\begin{equation*}
    \overset{\text{Pure dimer (even $L$)}}{
        \int \mathbf{D}\boldsymbol{\chi}\,
        e^{\tfrac{1}{2}\boldsymbol{\chi}^\top \mathbf{K} \boldsymbol{\chi}}
    }
    \;\;\;\longrightarrow\;\;\;
    \overset{\text{Dimer with monomers (even and odd $L$)}}{
        \int \mathbf{D}\boldsymbol{\chi}\,
        e^{\tfrac{1}{2}\boldsymbol{\chi}^\top \mathbf{K} \boldsymbol{\chi}
        + \mathbf{u}^\top \boldsymbol{\chi}
        + \boldsymbol{\psi}^\top \boldsymbol{\chi}}
    }.
\end{equation*}
In this framework, the bosonic source $\mathbf{u}$ and fermionic source $\boldsymbol{\psi}$ play complementary roles in encoding monomer configurations. The bosonic source $\mathbf{u}$ introduces single monomer insertions, while the fermionic source $\boldsymbol{\psi}$ systematically encodes multiple monomer configurations. The generalized Berezin integral over Grassmann variables (Theorem~\ref{exponential with linear fermionic and bosonic sources}) thus provides a unified description of the dimer partition function with fixed monomers and monomer correlations.

Setting the fermionic source to zero yields two distinct physical regimes: for even vertex count, the partition function reduces to the pure dimer model; for odd vertex count, it describes systems with single monomers, enumerating almost-perfect matchings where the bosonic source components $u_i$ assign weights to individual monomers. Conversely, setting the bosonic source to zero generates partition functions with fixed multiple monomers, which can be integrated over the fermionic source to recover the dimer model with fixed monomer statistics, following the approach established in~\cite{AllegraFortin2014}.

We present two fundamental theorems that formalize this framework: Theorem~\ref{thm:main} establishes the Kasteleyn correspondence for planar dimers, while Theorem~\ref{thm:main2} generalizes this result to monomer-dimer systems. The latter theorem demonstrates the equivalence between the exponential form with mixed sources in Theorem~\ref{exponential with linear fermionic and bosonic sources} and the combinatorial dimer and monomer-dimer models. 

These results are established through combinatorial and algebraic methods employing Levi-Civita symbols and Grassmann variables, providing relationships between Hafnians and Pfaffians—and their submatrix generalizations—in planar graphs. For the pure dimer model, we demonstrate the transformation from Hafnians to Pfaffians, converting computationally intractable enumeration into tractable evaluation. For monomer-dimer systems, we systematically derive the decomposition of the Monobisyzexant function into sums of Pfaffians and Pfaffian submatrix functions (Pfaffianhos), revealing that while the inherent computational complexity persists even for planar graphs, our algebraic framework enables exact combinatorial enumeration through structured decomposition.

This framework enables the computation of both even- and odd-point correlation functions in arbitrary dimensions, extending beyond conventional Grassmann–Berezin treatments of dimer models.

A related approach was introduced in~\cite{SeifiMirjafarlou2024}, where linear bosonic sources are incorporated using complex fermions $c$ and $c^\dagger$. Like our construction, their results do not depend on the matrix dimension; however, they assume that the matrix is invertible and include only bosonic sources, whereas our framework employs abstract real fermions (Grassmann variables) with fermionic sources and does not require invertible matrices. Following~\cite{negele1988quantum}, such generalized Berezin integrals are also relevant in quantum many-body systems, where linear source terms naturally encode particle insertions and correlation functions.

The core mathematical results of this paper are two master Berezin integrals over Grassmann variables identities:
\begin{itemize}
    \item \textbf{Theorem \ref{exponential with linear fermionic and bosonic sources}:} A Berezin integral over Grassmann variables for real fermions with mixed bosonic--fermionic sources, valid for any antisymmetric matrix $\mathbf{A}$ regardless of dimension or invertibility.
    \item \textbf{Theorem \ref{theorem 2.2}:} The complex analogue, yielding Determinant-based representations.
\end{itemize}

We derive unified formulations of dimer, dimer with fixed monomers, and dimer--fugacity models (Section~\ref{Sec:2}), establishing a systematic framework for organizing Pfaffian hierarchies associated with different monomer sectors.

Having generalized the dimer model to arbitrary source insertions (bosonic and fermionic), it is natural to turn to other connectivity structures where similar Grassmann--Berezin techniques apply—most notably, spanning trees and forests. 

Spanning trees and forests are among the most elementary yet fundamental combinatorial structures in graph theory~\cite{Kirchhoff1847, Cayley1889, Tutte1954, ChaikenKleitman1978}. In a connected graph $g$, a \textit{spanning tree} is a subgraph that connects all vertices without forming loops. A natural generalization is provided by \textit{rooted spanning forests}, obtained by selecting a subset of vertices (the roots) and constructing trees rooted at each of them; hence, the number of roots coincides with the number of trees in the forest. These objects admit a remarkable algebraic interpretation: the number of spanning trees of $g$ is given by the determinant of any principal minor of its Laplacian matrix (obtained by deleting one row and one column)~\cite{Kirchhoff1847}. Likewise, the number of rooted spanning forests is obtained by deleting the rows and columns corresponding to the chosen roots~\cite{Tutte1958, Tutte1954, ChaikenKleitman1978}.  

The study of these structures originates in Kirchhoff’s pioneering work on electrical networks~\cite{Kirchhoff1847}, which led to the celebrated \textit{Matrix–Tree Theorem}. Since then, spanning trees and forests have played central roles across mathematics and physics, appearing in contexts such as effective resistances in electrical networks~\cite{Kirchhoff1847}, percolation theory~\cite{Grimmett1999}, the Abelian sandpile model and self-organized criticality~\cite{MajumdarDhar1992}, and renormalization theory through the rooted forests of the Connes–Kreimer Hopf algebra~\cite{ConnesKreimer1998}.

Here we introduce a \textit{source-ordered} Grassmann–Berezin representation of spanning trees (Corollary~\ref{Corollary2.6}), which extends the usual quadratic-term Berezin-integral formulations~\cite{berezin1965, Samuel1980} and organizes root insertions via an explicit ordering of source fields. The construction couples Grassmann variables to bosonic source (complex vectors) \(\mathbf{u},\bar{\mathbf{u}}\) and uses an ordered product of exponentials of linear terms that cannot, in general, be combined into a single exponential because of Grassmann anticommutation:
\begin{equation*}
\overset{\text{Spanning tree}}{\int \mathbf{D}(\boldsymbol{\chi},\bar{\boldsymbol{\chi}})\,
    \bar{\chi}_{a}\chi_{a}\,
    e^{\bar{\boldsymbol{\chi}}^\top\mathbf{L}\boldsymbol{\chi}}}
    \quad\longrightarrow\quad\overset{\text{Spanning tree}}{\left(\sum_{i=1}^L\sum_{j=1}^L\bar u_j u_i\right)^{-1}
    \int\mathbf{D}(\boldsymbol{\chi},\bar{\boldsymbol{\chi}})\,
    e^{\bar{\boldsymbol{\chi}}^\top\mathbf{L}\boldsymbol{\chi}}\,
    e^{\bar{\boldsymbol{\chi}}^\top\mathbf{u}}\,
    e^{\bar{\mathbf{u}}^\top\boldsymbol{\chi}}},
\end{equation*}

where \(\mathbf{L}\) denotes the graph Laplacian (a formal definition is given in Appendix~\ref{Laplacianmatrix}).

The bosonic sources \((\mathbf{u},\bar{\mathbf{u}})\) implement root insertions while preserving the full operator structure of \(\mathbf{L}\); in practice these sources regularize the Laplacian's zero mode without resorting to explicit cofactors or minors, providing a compact device for analytic continuation and perturbative expansions. Importantly, the spanning-tree measure extracted from the source-augmented Berezin integral is independent of the particular choice of bosonic source vectors, so the insertion plays only an auxiliary (gauge-like) role in the construction.

Finally, for the case of singular Kasteleyn matrices, we present a practical method for handling zero modes in Berezin integrals (Theorem~\ref{HuaTheorem}), based on the unitary block decomposition of an antisymmetric matrix $\mathbf{K}$. This decomposition separates the zero and non-zero sectors, preserves the Pfaffian in the non-zero block, and explicitly identifies the additional Grassmann monomials associated with the kernel, thereby allowing the direct evaluation of Berezin integrals without the need to invoke formal inverses. As a result, Berezin integrals can be evaluated directly, without resorting to formal inverses—a generalization of \cite{AllegraFortin2014}, which assumes non-singular matrices. Similarly, an analogous technique using spectral decomposition enables the enumeration of spanning forests in disconnected graphs: by separating the connected components in the block Laplacian matrix, one can count the trees in each component independently.

Our approach provides a systematic treatment of almost-perfect matchings and spanning trees, delivering explicit formulas for partition functions and correlation functions with mixed sources. The formalism bridges algebraic combinatorics and functional integration, offering a flexible toolkit for both theoretical analysis and computational implementations in graph-based and field-theoretic models.

This work is structured to systematically transition from established combinatorial models to our generalized mathematical framework. Section~\ref{Sec:2} reviews the combinatorial formulation of the dimer model, establishing the foundation for its Grassmann-Berezin integral representations. We begin with the general dimer model on arbitrary graphs before specializing to the planar case, where Pfaffian techniques enable significant computational simplifications. Section~\ref{Sec:3} extends this analysis to the monomer-dimer model for general graphs, demonstrating that even with Pfaffian methods the planar monomer-dimer case remains computationally challenging. We then introduce two tractable special cases of dimer models with fixed monomers: one employing fermionic sources and the other utilizing bosonic sources. Section~\ref{Sec:4} examines spanning trees and forests, developing an alternative Berezin integral representation with complex bosonic sources. These applications culminate in the mathematical framework presented in Section~\ref{Sec:5}, which provides rigorous definitions of the Hafnian and Monobisyzexant functions alongside their Grassmann-Berezin representations. The core analytical contributions appear in Section~\ref{Sec:6}, where we establish Theorems~\ref{exponential with linear fermionic and bosonic sources} and~\ref{theorem 2.2}—comprehensive Berezin integral identities for real and complex fermions with mixed bosonic-fermionic sources.

\section*{Notation}

We define the index set \([N] = \{1,\dots,N\}\). 
Subsets of \([N]\) with cardinality \(r\) are denoted by 
\(I_{[r]}, J_{[r]} \subseteq [N]\). 
Specifically,
\[
I_{[r]} = \{i_1,\dots,i_r\}, 
\qquad 
J_{[r]} = \{j_1,\dots,j_r\},
\]
where the indices are listed in increasing order. 
The complement of \(I_{[r]}\) in \([N]\) is denoted \(I^c_{[r]}\), 
and analogously \(J^c_{[r]}\). 
We also use the notation 
\[
I_{[0,r]} = \{i_0,i_1,\dots,i_r\}.
\]

Let \(\mathbf{A}\) be an \(N \times M\) matrix. 
For index sets \(I_{[r]} \subseteq [N]\) and \(J_{[s]} \subseteq [M]\), 
we denote by \(\mathbf{A}_{\left[I_{[r]},J_{[s]}\right]}\) the submatrix obtained by restricting 
\(\mathbf{A}\) to the rows indexed by \(I_{[r]}\) and the columns indexed by \(J_{[s]}\), 
preserving their original order. 
If all rows are retained, we write 
\(\mathbf{A}_{\left[\star,J_{[s]}\right]} = \mathbf{A}_{\left[[N],J_{[s]}\right]}\); 
if all columns are retained, we write 
\(\mathbf{A}_{\left[I_{[r]},\star\right]} = \mathbf{A}_{\left[I_{[r]},[M]\right]}\). 
For the special case \(I_{[r]}=J_{[r]}\), we use the shorthand 
\(\mathbf{A}_{\left[I_{[r]},I_{[r]}\right]} \equiv \mathbf{A}_{\left[I_{[r]}\right]}\).  

For multiple summations, we adopt the convention
\begin{equation*}
    \begin{split}
    \sum_{i_1,\dots,i_{m}=1}^{L}&\equiv\; \sum_{i_1=1}^{L}\cdots\sum_{i_{m}=1}^{L},\\
    \sum_{i_1<\dots<i_{m}=1}^{L}&\equiv\sum_{i_1=1}^{L-m+1}\sum_{i_2=i_1+1}^{L-m+2}\cdots\sum_{i_{m-1}=i_{m-2}+1}^{L-1}\sum_{i_{m}=i_{m-1}+1}^{L}.
    \end{split}
\end{equation*}

Finally, we define the sign functions  
\[
\epsilon(I_{[r]})=(-1)^{\frac{r(r-1)}{2}} \, (-1)^{\sum\limits_{i \in I_{[r]}} i},
\]
which corresponds to the signature of the permutation mapping the sequence \(1\dots N\) into \(I_{[r]}I^c_{[r]}\), with both subsequences ordered increasingly. Analogously,
\[
\epsilon\left(I_{[r]},J_{[r]}\right)=\epsilon\left(I_{[r]}\right)\epsilon\left(J_{[r]}\right) = (-1)^{\sum\limits_{i \in I_{[r]}} i+\sum\limits_{j\in J_{[r]}} j},
\]
representing the sign of the permutation that maps \(I_{[r]}I^c_{[r]}\) into \(J_{[r]}J^c_{[r]}\).

\section{The dimer model, Planar dimers, and Pfaffian Techniques I}\label{Sec:2}

This section develops a systematic analytical framework for dimer models through a synthesis of combinatorial and algebraic methodologies. We commence with the classical dimer model on arbitrary graphs, formulating the partition function in terms of the Hafnian of the adjacency matrix while introducing correlation functions through the Hafnianinho formalism. For planar graph configurations, we demonstrate how Kasteleyn's seminal method facilitates an equivalent representation via Pfaffians of oriented adjacency matrices, with correlation functions naturally expressed through corresponding Pfaffianinho constructions. The fundamental equivalence between Hafnian and Pfaffian formulations is rigorously established through Berezin integration over Grassmann variables. While the core combinatorial results presented in this section are well-established in the literature, we provide a comprehensive treatment with particular emphasis on their translation into the Grassmann-Berezin formalism, ensuring both completeness and pedagogical clarity for subsequent developments.

\subsection{The Hafnian and the dimer model}
\begin{definition}[Graph]
A (finite) graph is $G=(V,E)$, optionally with a \textit{weight} $w:E\to\R_{\ge 0}$. Unless stated, graphs are simple (no loops; no parallel edges). Write $n=|V|$ and $m=|E|$.
\end{definition}

\begin{definition}[Weighted adjacency matrix]
The adjacency matrix associated with the graph $G$ is $\mathbf{A}$, where each entry corresponds to the ordered pair of vertices $(v_i, v_j)$ and is given by
\[
A_{ij}=\begin{cases}
w_{ij}, & \text{if } v_i \text{ and } v_j \text{ are connected by an edge},\\[4pt]
0, & \text{otherwise}.
\end{cases}
\]
(When $w\equiv 1$ it is just the adjacency matrix.)
\end{definition}

\begin{definition}[Perfect Matching]
A perfect matching of \( G \) is a set \( (PM) \subseteq E \) of edges such that:
\begin{enumerate}
    \item Every vertex \( v \in V \) is incident to exactly one edge in \( (PM) \), i.e.,
    \[
    \forall v \in V, \; \exists! \, e \in (PM) \text{ such that } v \in e.
    \]
    \item No two edges in \( (PM) \) share a common vertex (edges are pairwise non-adjacent).
\end{enumerate}
Equivalently, \( (PM) \) is a perfect matching if it covers all vertices of \( G \) and forms a set of disjoint edges whose union is \( V \).
\end{definition}

\begin{definition}[Dimer (or Perfect Matching) Partition Function]
The partition function of the dimer model on a graph $G = (V,E)$ with edge weights $\{w_e\}_{e \in E}$ is defined as:
\[
Z_D(G) = \sum_{M \in \mathcal{M}(G)} \prod_{e \in M} w_e,
\]
where $\mathcal{M}(G)$ denotes the set of all perfect matchings of $G$. For unweighted graphs ($w_e = 1$ for all $e \in E$), this reduces to the number of perfect matchings:
\[
Z_D(G) = |\mathcal{M}(G)|.
\]
\end{definition}

\begin{definition}[Dimer (or Edge) Correlation Function]
For a set of distinct edges $\{e_1, e_2, \dots, e_p\}$ with $e_r = (i_r, j_r)$, the $p$-point correlation function is defined as:
\begin{equation}
\langle e_1 e_2 \cdots e_p\rangle_D:= \frac{Z_D(G \setminus \{e_1, \dots, e_p\})}{Z_D(G)}
\end{equation}
where $Z(G \setminus \{e_1, \dots, e_p\})$ denotes the partition function on the graph with edges $e_1, \dots, e_p$ removed.
\end{definition}

\begin{definition}[Monomer (or vertices) Correlation Function]
For a set of vertices $\{v_1, v_2, \dots, v_{2l}\}$, the monomer correlation function gives the probability that these vertices remain unoccupied (i.e., are not covered by any dimer):
\begin{equation}
\langle v_1 v_2 \cdots v_{2l} \rangle_D:= \frac{Z_D(G \setminus \{v_1, \dots, v_{2l}\})}{Z_D(G)}
\end{equation}
where $Z(G \setminus \{v_1, \dots, v_{2l}\})$ is the partition function on the graph with vertices $v_1, \dots, v_{2l}$ removed.
\end{definition}

\begin{definition}[Mixed Monomer--dimer Correlation Function]
For a set of vertices $\{v_1, v_2, \dots, v_{2l}\}$ and a set of distinct edges $\{e_1, e_2, \dots, e_p\}$ with $e_r= (i_r, j_r)$, the mixed correlation function gives the joint probability that all specified vertices are unoccupied (monomers) and all specified edges are occupied by dimers:
\begin{equation}
\langle v_1\cdots v_{2l};e_1\cdots e_p\rangle_D:=\frac{Z_D(G \setminus \{v_1,\dots,v_{2l},e_1,\dots,e_p\})}{Z_D(G)}
\end{equation}
where $Z_D(G \setminus \{v_1,\dots,v_{2l},e_1,\dots,e_p\})$ denotes the partition function on the graph with both the specified vertices and edges removed. The removal of edges $e_1,\dots,e_p$ and vertices $v_1,\dots,v_{2l}$ is performed simultaneously, and the resulting graph may have isolated vertices or disconnected components.
\end{definition}

\begin{theorem}[Berezin integral over Grassmann variables representation of the dimer model~\cite{Samuel1980, Samuel1980b}]

The partition function of the dimer model counts the total number of such coverings, that is, the number of perfect matchings of the graph. This partition function is given by
\begin{equation}
    Z_D(\mathbf{A})=\int\mathbf{D}(\boldsymbol{\chi},\bar{\boldsymbol{\chi}})
    e^{\frac{1}{2} (\Bar{\boldsymbol{\chi}}\boldsymbol{\chi})^\top\mathbf{A} (\Bar{\boldsymbol{\chi}}\boldsymbol{\chi})}=\haf(\mathbf{A}),
\end{equation}
where we define $(\Bar{\boldsymbol{\chi}}\boldsymbol{\chi})^\top:=\begin{pmatrix}
        \bar{\chi}_1\chi_1&\cdots&\bar{\chi}_{2L}\chi_{2L}
    \end{pmatrix}$, $\mathbf{D}(\boldsymbol{\chi},\bar{\boldsymbol{\chi}}):=d\chi_1d\bar{\chi}_1\cdots d\chi_{2L}d\bar{\chi}_{2L}$, and $\haf(\mathbf{A})$ is the Hafnian (Definition~\ref{Def:Hafnian}) of the adjacency matrix $\mathbf{A}$.

The mixed $2r$-correlation function is
\begin{equation}
    \langle v_1 \cdots v_{2l}; e_1 \cdots e_p \rangle_D=\left\langle\left(\prod_{\alpha=1}^{2l}\bar{\chi}_{i_\alpha}\chi_{i_\alpha}\right) \left(\prod_{\beta=1}^{p}\bar{\chi}_{j_\beta}\chi_{j_\beta}\bar{\chi}_{k_\beta}\chi_{k_\beta}\right)\right\rangle_{\mathrm{D}},\quad(2l+p=2r\leq2L).
\end{equation}
When $l=0$, we have the dimer correlation function, and if $p=0$, we have the monomer correlation. The left side is defined as
\begin{equation}
\begin{split}
&\left\langle\left(\prod_{\alpha=1}^{2l}\bar{\chi}_{i_\alpha}\chi_{i_\alpha}\right)\left(\prod_{\beta=1}^{p}\bar{\chi}_{j_\beta}\chi_{j_\beta}\bar{\chi}_{k_\beta}\chi_{k_\beta}\right)\right\rangle_{\mathrm{D}}\\
&\qquad\qquad=\left\langle \prod_{\alpha=1}^{2r}\bar{\chi}_{i_\alpha}\chi_{i_\alpha} \right\rangle_{\mathrm{D}}
   := \frac{1}{Z_{D}(\mathbf{A})}
   \int \mathbf{D}(\boldsymbol{\chi},\bar{\boldsymbol{\chi}})
   \left(\prod_{\alpha=1}^{2r}\bar{\chi}_{i_\alpha}\chi_{i_\alpha}\right)
   e^{\frac{1}{2}(\bar{\boldsymbol{\chi}}\boldsymbol{\chi})^\top \mathbf{A}(\bar{\boldsymbol{\chi}}\boldsymbol{\chi})}
   =\frac{\haf\left(\mathbf{A}_{\left[I^c_{[2r]}\right]}\right)}{\haf(\mathbf{A})},
\end{split}
\end{equation}
where $I_{[2r]}= \{i_1, \dots, i_{2r}\} \subseteq [2L]$, and $I^c_{[2r]}$ denotes the complement of $I_{[2r]}$. Here, $\haf\!\left(\mathbf{A}_{\left[I^c_{[2r]}\right]} \right)$ is the Hafnian of a submatrix of $\mathbf{A}$, referred to as the \textit{hafnianinho} (Definition~\ref{def:hafnianinho}). 

The proof follows directly from expanding the exponential and applying Berezin integration rules, which systematically extract the perfect matching contributions through the Hafnian structure. The correlation functions emerge naturally as normalized expectation values of Grassmann operator products, reducing to ratios of Hafnians through the fundamental properties of Grassmann Gaussian integrals. Complete fermionic Hafnian identities are further developed in Corollary~\ref{Coro:Fermionichafnian}.
\end{theorem}

As established, computing the Hafnian is \#P-complete~\cite{valiant_computer_science, Barvinok2016a}, making exact enumeration of perfect matchings computationally intractable for arbitrary graphs. Nevertheless, significant advances in approximation methods have emerged through algebraic transformations of the underlying combinatorial structure. In particular, the Berezin integral over Grassmann variables representation of the dimer model, when combined with bosonization techniques~\cite{HSTrans1,HSTrans}, transforms the combinatorial enumeration into an ordinary integral amenable to numerical approximation. Recent work demonstrates that fermionic dual transformations applied to the original Berezin integral over Grassmann variables, followed by bosonization, yield substantially improved approximations. This approach has proven successful in computing sums of powers of principal minors (SPPM)~\cite{najafi2024,Sourlas2019} and extends beyond dimer systems to any Berezin integral with quartic interactions in the exponential. Complementary to these methods, the Bethe approximation provides another powerful technique for evaluating Berezin integrals over Grassmann variables with quartic terms, as demonstrated in the context of SPPM~\cite{Ramezanpour2024}. These algebraic approaches collectively offer practical computational pathways for problems that are formally intractable through exact methods.

\subsection{From Hafnian to Kasteleyn Pfaffian and the Planar dimer model}
\begin{definition}[Planar drawing and planar graph]
A \textit{planar drawing} of $G$ is a mapping of vertices to distinct points in the plane and edges to simple curves connecting their endpoints such that curves only intersect at common endpoints. A graph is \textit{planar} if it admits a planar drawing.
\end{definition}

\begin{definition}[Skew adjacency / Kasteleyn matrix for an orientation]\label{def:skew}
Fix an orientation $\vec G$ of $G$. Its \textit{skew adjacency matrix} $K(\vec G,w)\in\mathbb R^{n\times n}$ is
\[
K_{ij} \;=\;
\begin{cases}
+w_{ij} & \text{if } i\!\to\! j \text{ in }\vec G,\\
-w_{ij} & \text{if } j\!\to\! i \text{ in }\vec G,\\
0 & \text{if } \{i,j\}\notin E,
\end{cases}
\qquad \mathbf{K}=-\mathbf{K}^\top .
\]
(When $w\equiv 1$ it is just the signed skew-adjacency.)
\end{definition}

\begin{definition}[Kasteleyn orientation]\label{Def:KasteleynOrientation}
An orientation of the edges is \textit{Kasteleyn} if for every (bounded) face, the number of edges oriented clockwise along the facial boundary is odd.
\end{definition}

As an illustrative example, consider the planar graph shown in Figure~\ref{Fig:Korientation}. The corresponding Kasteleyn matrix is given by
\begin{equation*}
    \mathbf{K} = 
    \begin{pmatrix}
        0 & -1 &  1 &  0 &  0 &  1 \\
        1 &  0 &  1 &  0 &  0 &  0 \\
       -1 & -1 &  0 &  1 &  0 &  0 \\
        0 &  0 & -1 &  0 &  1 &  1 \\
        0 &  0 &  0 & -1 &  0 & -1 \\
       -1 &  0 &  0 & -1 &  1 &  0
    \end{pmatrix},
\end{equation*}
where the matrix entries reflect the Kasteleyn orientation: $K_{12} = -1$ indicates the edge is oriented from vertex 2 to 1, $K_{34} = 1$ corresponds to orientation from vertex 3 to 4, and $K_{26} = 0$ confirms the absence of an edge between vertices 2 and 6.

\begin{figure}[]
    \centering
    \includegraphics[width=0.8\linewidth]{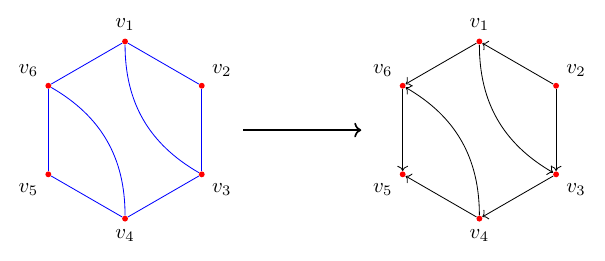}
    \caption{(Left) Original planar graph. (Right) Kasteleyn-oriented version with edge directions assigned to satisfy the odd clockwise orientation condition for each face.}
    \label{Fig:Korientation}
\end{figure} 

Using a fixed planar embedding and Kasteleyn orientation, we define the following:
\begin{definition}
The antisymmetric \textit{sign matrix} $\mathbf{S}=(S_{ij})$ by
\[
S_{ij}=
\begin{cases}
+1 &\text{if } i\to j,\\
-1 &\text{if } j\to i,\\
0 &\text{if } \{i,j\}\notin E,
\end{cases}
\quad\text{with } S_{ij}=-S_{ji}.
\]    
\end{definition}

We now provide a detailed explanation of the mathematical structures underlying the Hafnian and Pfaffian, highlighting both their similarities and fundamental differences.

The Hafnian, defined for $2L \times 2L$ symmetric matrices as (Definition~\ref{Def:Hafnian}):
\begin{equation*}
    \haf(\mathbf{A}):=\frac{1}{2^L L!}\sum_{i_1,\dots,i_{2L}=1}^{2L}\epsilon^{i_1\dots i_{2L}}\epsilon^{i_1\dots i_{2L}}A_{i_1i_2}\cdots A_{i_{2L-1}i_{2L}},
\end{equation*}
The product $A_{i_1i_2}\cdots A_{i_{2L-1}i_{2L}}$ represents the statistical weight of a potential dimer configuration, where each matrix element $A_{ij}$ encodes both the presence and the weight of the edge connecting vertices $i$ and $j$. 

The Levi-Civita symbol $\epsilon^{i_1\dots i_{2L}}$ plays two crucial roles in this construction:
\begin{itemize}
    \item It ensures the indices $(i_1,\dots,i_{2L})$ form a permutation of $\{1,\dots,2L\}$, thereby enforcing the condition that each vertex appears exactly once---the mathematical expression of a \textit{perfect matching}.
    
    \item In the Hafnian formulation, the square of the Levi-Civita symbol yields a constant positive sign for all contributing configurations, distinguishing it from the determinant where permutation signs are preserved.
\end{itemize}

When we introduce an arbitrary orientation to the graph, we assign signs to the edges, transforming the weights from $A_{ij} = w_{ij}$ to $S_{ij}A_{ij} = \pm w_{ij}$. Defining a skew-symmetric matrix through this orientation is always possible, but the critical insight lies in the behavior of the combinatorial sign:
\[
\epsilon^{i_1\dots i_{2L}}S_{i_1 i_2}\cdots S_{i_{2L-1}i_{2L}}.
\]
For a general orientation, this product depends on the specific perfect matching $M$ and can be $\pm 1$. The remarkable property of the \textit{Kasteleyn orientation} is that it ensures this product equals $+1$ for \textit{every} perfect matching $M$.

This fundamental observation enables the transformation from the Hafnian to the Pfaffian. The Pfaffian of a $2L \times 2L$ antisymmetric matrix $\mathbf{K}$ is defined as:
\begin{equation}
    \pf(\mathbf{K}):= \frac{1}{2^L L!}\sum_{i_1,\dots,i_{2L}=1}^{2L}\epsilon^{i_1\dots i_{2L}}K_{i_1i_2}\cdots K_{i_{2L-1}i_{2L}},
\end{equation}
and admits efficient computation via the identity $\pf(\mathbf{K})^2 = \det(\mathbf{K})$.

Kasteleyn's seminal insight~\cite{Kasteleyn1961,Kasteleyn1963} was to demonstrate that for planar graphs, the symmetric adjacency matrix $\mathbf{A}$ can be transformed into an antisymmetric \textit{Kasteleyn matrix} $\mathbf{K}$ through a carefully chosen sign matrix $\mathbf{S}$ encoding the Kasteleyn orientation, thereby establishing the following:
\begin{theorem}[Hafnian-Pfaffian correspondence for Planar Dimers (Partition Function)]\label{thm:main}
Let $G = (V,E)$ be a planar graph with $|V|$ even, equipped with a symmetric weight matrix $\mathbf{A}$ satisfying $A_{ij} \geq 0$ and $A_{ii} = 0$. Let $\mathbf{K}$ be a Kasteleyn matrix for $G$ constructed via $K_{ij}=\epsilon^{ij}A_{ij}$ and following a Kasteleyn orientaion, where $\mathbf{S}=(\epsilon^{ij})$ is a Kasteleyn sign matrix. Then the following are equivalent:

\begin{enumerate}[label=(\roman*)]
\item \textbf{Combinatorial form}~\cite{Kasteleyn1961}: The dimer partition function equals the Pfaffian of the Kasteleyn matrix:
\begin{equation}
Z_D=\haf(\mathbf{A})=\pf(\mathbf{K}).
\end{equation}

\item \textbf{Grassmann-Berezin form}: The quartic and quadratic Berezin integrals over Grassmann variables coincide:
\begin{equation}
Z_D=\int \mathbf{D}(\boldsymbol{\chi},\bar{\boldsymbol{\chi}})
e^{\frac{1}{2}(\bar{\boldsymbol{\chi}}\boldsymbol{\chi})^\top\mathbf{A}(\bar{\boldsymbol{\chi}}\boldsymbol{\chi})}
= \int \mathbf{D}\boldsymbol{\chi} \, e^{\frac{1}{2}\boldsymbol{\chi}^\top\mathbf{K}\boldsymbol{\chi}}.
\end{equation}
\end{enumerate}

The proof is in Appendix~\ref{Proof:TheoHfPf}.
\end{theorem}

Within the Berezin integral formalism over Grassmann variables, the transition from fermionic Hafnian to fermionic Pfaffian relies on the Kasteleyn condition $\epsilon^{i_1\dots i_{2L}}S_{i_1 i_2}\cdots S_{i_{2L-1}i_{2L}}=1$, which in fermionic terms becomes:
\[
\int\left(\prod_{i=1}^{2L}d\chi_i\right) \chi_{i_1}\cdots\chi_{i_{2L}} S_{i_1 i_2}\cdots S_{i_{2L-1}i_{2L}} = 1.
\]

To complete our mathematical toolkit for combinatorial functions, we now introduce the Hafnianinho (Definition~\ref{def:hafnianinho}):

\begin{equation*}
    \haf\left(\mathbf{A}_{\left[I^c_{[2r]}\right]}\right):=\frac{1}{2^{L-r}\left(L-r\right)!}\sum_{j_1,\dots,j_{2(L-r)}=1}^{2L}\epsilon^{i_1\cdots i_{2r}j_1\cdots j_{2(L-r)}}\epsilon^{i_1\cdots i_{2r}j_1\cdots j_{2(L-r)}}A_{j_1j_2}\cdots A_{j_{2(L-r)-1}j_{2(L-r)}},
\end{equation*}

where $I=\{i_1,\dots,i_{2r}\}\subseteq[2L]$ with $i_1<\dots<i_{2r}$. This function enumerates perfect matchings within vertex-induced subgraphs, providing a generalized counting mechanism for subsystems with even cardinality.

The Hafnianinho computes the number of perfect matchings in the vertex-induced subgraph defined by the complement set $I^c$, which contains $2(L-r)$ vertices. Formally:

\begin{itemize}
    \item The double Levi-Civita symbol $\epsilon^{i_1\cdots i_{2r}j_1\cdots j_{2(L-r)}}$ ensures that:
    \begin{enumerate}
        \item The indices $\{j_1,\dots,j_{2(L-r)}\}$ form a permutation of $I^c$
        \item All contributing terms enter with positive sign, analogous to the standard Hafnian
    \end{enumerate}
    
    \item The product $A_{j_1j_2}\cdots A_{j_{2(L-r)-1}j_{2(L-r)}}$ represents the weight of a potential perfect matching configuration restricted to the subgraph induced by $I^c$
    
    \item The normalization factor $2^{L-r}(L-r)!$ accounts for the overcounting inherent in summing over all permutations of the $2(L-r)$ vertices
\end{itemize}
The transition from Hafnianinho to Pfaffianinho requires an additional ingredient beyond the Kasteleyn orientation and sign matrix. Consider a graph where every pair of vertices shares an edge, we are excluding pure dimers. Following the same procedure as for the Hafnian, we assign an orientation to the edges and obtain the sign subsystem of the submatrix used for counting dimers or perfect matchings. This yields the expression:
\[
\epsilon^{i_1\cdots i_{2r}j_1\cdots j_{2(L-r)}}S_{j_1 j_2}\cdots S_{j_{2(L-r)-1}j_{2(L-r)}}.
\]
We observe that additional signs are required to achieve the identity:
\[
\epsilon^{i_1\cdots i_{2r}j_1\cdots j_{2(L-r)}}S_{i_1 i_2}\cdots S_{i_{2r-1}i_{2r}}S_{j_1 j_2}\cdots S_{j_{2(L-r)-1}j_{2(L-r)}}=1.
\]
These missing signs are obtained by introducing the product $S_{ij}K_{ij}=1$, demonstrating that the Hafnianinho in this case equals the product of oriented edges $K_{ij}$ with the Pfaffianinho of the Kasteleyn matrix.

We now introduce the \textit{Pfaffian of a submatrix}, or \textit{Pfaffianinho}, which extends the Pfaffian to principal submatrices. For a $2L\times2L$ antisymmetric matrix $\mathbf{K}$ with rows and columns $I=\{i_1,\dots,i_{2r}\}\subseteq[2L]$ removed, where $i_1<\dots<i_{2r}$, the Pfaffianinho is defined as:
\begin{equation}
\pf\left(\mathbf{K}_{\left[I^c_{[2r]}\right]}\right):=\frac{\epsilon\left(I_{[2r]}\right)}{(L-r)!2^{L-r}}\sum_{j_1,\dots,j_{2(L-r)}=1}^{2L}\varepsilon^{i_1\cdots i_{2r}j_1\cdots j_{2(L-r)}}K_{j_1j_2}\cdots K_{j_{2(L-r)-1}j_{2(L-r)}},
\end{equation}
where $1\leq r<L$. For the special case $r=L$, we define:
\begin{equation}
    \pf\left(\mathbf{K}_{[\emptyset]}\right):=1.
\end{equation}

In the more general scenario where extracted vertices are not necessarily adjacent (forming monomers rather than dimers), the orientation structure changes fundamentally. The expression:
\[
\epsilon^{i_1\cdots i_{2r}j_1\cdots j_{2(L-r)}}S_{j_1 j_2}\cdots S_{j_{2(L-r)-1}j_{2(L-r)}}
\]
cannot be completed with missing signs because monomers represent separate points without common edges. The solution employs \textit{disordering operators}, which generate new configurations with corrected signs, transforming:
\[
\epsilon^{i_1\cdots i_{2r}j_1\cdots j_{2(L-r)}}S_{j_1 j_2}\cdots S_{j_{2(L-r)-1}j_{2(L-r)}}K_{j_1j_2}\cdots K_{j_{2(L-r)-1}j_{2(L-r)}}
\]
into $K'_{j_1j_2}\cdots K'_{j_{2(L-r)-1}j_{2(L-r)}}$, where $K'_{ij}$ are components of the modified Kasteleyn matrix.

\begin{definition}[Modified Kasteleyn Matrix]
For a planar graph $G$ with Kasteleyn matrix $\mathbf{K}$ and monomer set $M \subset V$, the modified Kasteleyn matrix $\mathbf{K}'$ is constructed by selecting paths from each monomer to the boundary and reversing the sign of the matrix components $K_{ij}$ for each edge $(i,j)$ along these paths. The perfect matching count with monomers at $M$ is given by $\left|\pf\left(\mathbf{K}'_{\left[I^c_{[2r]}\right]}\right)\right|$, where $I^c_{[2r]}= V \setminus M$.
\end{definition}

These constructions use Jacobi's identity for Pfaffians (Appendix~\ref{PropPf}).

\begin{lemma}[Hafnianinho-Pfaffianinho correspondence for Planar Dimers]\label{PropHafPf}
Let $G = (V,E)$ be a planar graph with $|V|$ even, equipped with a symmetric weight matrix $\mathbf{A}$ satisfying $A_{ij} \geq 0$ and $A_{ii} = 0$. Let $\mathbf{K}$ be a Kasteleyn matrix for $G$ constructed via $K_{ij}=\epsilon^{ij}A_{ij}$, where $\mathbf{S}$ is a Kasteleyn sign matrix. Let $\mathbf{K}'$ be the modified Kasteleyn matrix incorporating monomer effects. Then:

\begin{enumerate}[label=(\roman*)]
\item \textbf{Combinatorial form}: The number of perfect matchings of the subgraph with the set of vertices $I^c_{[2r]}$ is:
\begin{equation}
\haf\!\left(\mathbf{A}_{\left[I^c_{[2r]}\right]}\right)=\begin{cases}
(-1)^r\pf(\mathbf{K'})\pf\left((\mathbf{K}'^{-1})_{\left[I_{[2r]}\right]}\right)\\
\epsilon\left(I_{[2r]}\right)\pf\left(\mathbf{K}'_{\left[I^c_{[2r]}\right]}\right)\\
\left(\prod_{\ell=1}^r K_{i_\ell j_\ell}\right)\epsilon\left(I_{[2r]}\right)\pf\left(\mathbf{K}_{\left[I^c_{[2r]}\right]}\right)&\text{if $I_{[2r]}$ is a set of dimers}
\end{cases}.
\end{equation}

\item \textbf{Grassmann-Berezin form}:
\begin{equation}
\int \mathbf{D}(\boldsymbol{\chi},\bar{\boldsymbol{\chi}})\left(\prod_{\alpha=1}^{2r}\bar{\chi}_{i_\alpha}\chi_{i_\alpha}\right)
e^{\frac{1}{2}(\bar{\boldsymbol{\chi}}\boldsymbol{\chi})^\top\mathbf{A}(\bar{\boldsymbol{\chi}}\boldsymbol{\chi})}
=\int\mathbf{D}\boldsymbol{\chi} 
\left(\prod_{\alpha=1}^{2r}\chi_{i_\alpha}\right)e^{\frac{1}{2}\boldsymbol{\chi}^\top\mathbf{K}'\boldsymbol{\chi}}.
\end{equation}
\end{enumerate}

The proof is in Appendix~\ref{ProofPropHafPf}.
\end{lemma}

\subsubsection{Planar dimer model}
The partition function of the planar dimer model enumerates the total number of perfect matchings (or dimer coverings) of the graph, as established in Theorem~\ref{thm:main}. For completeness and to gather all essential components of the planar dimer system in one place, we restate the partition function in function of the Kasteleyn matrix $\mathbf{K}$:
\begin{equation*}\label{eq:pf-int}
    Z_D(\mathbf{K}) = \int \mathbf{D}\boldsymbol{\chi} \, 
    e^{\frac{1}{2}\boldsymbol{\chi}^\top\mathbf{K}\boldsymbol{\chi}} = \mathrm{pf}(\mathbf{K}),
    \qquad 
    \mathbf{D}\boldsymbol{\chi} := d\chi_{2L} \cdots d\chi_1,
\end{equation*}

The correlation functions for this model are:
\begin{theorem}[Hafnianinho-Pfaffianinho correspondence for Planar Dimers (Correlation functions)]
The $r$-point edge correlation function for planar graphs~\cite{CimasoniReshetikhin2008} is given by:

\begin{equation}
\langle e_1 e_2 \cdots e_k \rangle_D=(-1)^k\left(\prod_{\ell=1}^k K_{i_\ell j_\ell}\right)\pf\left((\mathbf{K}^{-1})_{[I_{[k]}]}\right)=\left(\prod_{\ell=1}^k K_{i_\ell j_\ell} \right)\left\langle\prod_{\alpha
=1}^k\chi_{i_\alpha}\chi_{j_\alpha}\right\rangle_{\mathrm{D}},
\end{equation}
where each fermion pair $\chi_{i_\alpha}\chi_{j_\alpha}$ corresponds to edge $e_\alpha = (i_\alpha, j_\alpha)$, and $I_{[k]}=\{i_1, j_1,\dots,i_k, j_k\}$. The Grassmann expectation value is defined as:
\begin{equation}
\left\langle\prod_{\alpha=1}^k\chi_{i_\alpha}\chi_{j_\alpha}\right\rangle_{\mathrm{D}} := \frac{1}{Z_D(\mathbf{K})}\int\mathbf{D}\boldsymbol{\chi}\left(\prod_{\alpha=1}^k\chi_{i_\alpha}\chi_{j_\alpha}\right)e^{\frac{1}{2}\boldsymbol{\chi}^\top\mathbf{K}\boldsymbol{\chi}}.
\end{equation}

The $2r$-point monomer correlation for planar graphs~\cite{AllegraFortin2014} is given by:
\begin{equation}
\langle v_1 v_2 \cdots v_{2r} \rangle_D=\left\langle\left(\prod_{\alpha=1}^{2r}\chi_{i_\alpha}\right)e^{2\sum_{m \in V_m} \sum_{p \in \mathcal{P}_m} \sum_{(a,b) \in p} K_{ab} \chi_a \chi_b} \right\rangle_{D}.
\end{equation}
The term in the exponential represents the sum over all paths from the monomers to the boundary of the graph, where:
\begin{itemize}
\item $V_m$: set of vertices where monomers are located
\item $\mathcal{P}_m$: set of all paths from monomer $m$ to the graph boundary
\item $(a,b)$: edges along path $p$
\item $K_{ab}$: Kasteleyn matrix elements for edge $(a,b)$
\item $\chi_a, \chi_b$: Grassmann variables at vertices $a$ and $b$
\end{itemize}
and the expectation value is
\begin{equation}\label{2rmonomercorrelation}
\begin{split}
&\left\langle\left(\prod_{\alpha=1}^{2r}\chi_{i_\alpha}\right)e^{\sum_{m \in V_m} \sum_{p \in \mathcal{P}_m}\sum_{(a,b) \in p} K_{ab} \chi_a \chi_b} \right\rangle_{D}\\
&\qquad\qquad:=\frac{1}{Z_D(\mathbf{K})}\int\mathbf{D}\boldsymbol{\chi} 
    \left(\prod_{\alpha=1}^{2r}\chi_{i_\alpha}\right)e^{\frac{1}{2}\boldsymbol{\chi}^\top\mathbf{K}'\boldsymbol{\chi}}=(-1)^r\frac{\pf(\mathbf{K'})}{\pf(\mathbf{K})}\pf\left((\mathbf{K}'^{-1})_{\left[I_{[2r]}\right]}\right),
\end{split}
\end{equation}
with $I_{[2r]}=\{i_1,\dots,i_{2r}\}$, and the term in the exponent is
\begin{equation}
    \frac{1}{2}\boldsymbol{\chi}^\top\mathbf{K}'\boldsymbol{\chi}=\frac{1}{2}\boldsymbol{\chi}^\top\mathbf{K}\boldsymbol{\chi}-2\sum_{m \in V_m} \sum_{p \in \mathcal{P}_m} \sum_{(a,b) \in p} K_{ab} \chi_a \chi_b.
\end{equation}
$\mathbf{K'}$ is the modified version of the Kasteleyn matrix $\mathbf{K}$; a simple way to obtain it is by changing the sign in the original matrix of all the components that one traverses from the monomer to the edge of the graph.

\textbf{Proof Sketch:} The derivation proceeds by direct expansion of the Berezin integrals over Grassmann variables, revealing the underlying Pfaffianinho structure. The identification of these quantities with physical dimer-dimer, monomer-monomer, and mixed correlations is established through the combinatorial machinery developed in Lemma~\ref{PropHafPf}, which provides the essential mapping between Grassmann expectations and combinatorial enumeration.
\end{theorem}

The monomer-monomer correlation function presented in equation~\eqref{2rmonomercorrelation} admits an alternative interpretation as the partition function for a dimer model with fixed monomer constraints. Specifically, this partition function can be expressed as:
\begin{equation}
    Z_{D-M} = \pf(\mathbf{K}) \left\langle\left(\prod_{\alpha=1}^{2r}\chi_{i_\alpha}\right)
    e^{\sum_{m \in V_m} \sum_{p \in \mathcal{P}_m}\sum_{(a,b) \in p} K_{ab} \chi_a \chi_b} \right\rangle_{D}=(-1)^r\pf(\mathbf{K'})\pf\left((\mathbf{K}'^{-1})_{\left[I_{[2r]}\right]}\right),
\end{equation}
where the Grassmann expectation value encodes the statistical weight of dimer configurations subject to the condition that monomers are fixed at vertices $\{i_1,\dots,i_{2r}\}$. This formulation reveals the deep connection between correlation functions and constrained partition functions in the dimer model. A comprehensive analysis of the dimer model with fixed monomers, including explicit constructions of the modified Kasteleyn matrix $\mathbf{K}'$ and detailed combinatorial interpretations, is developed in the subsequent section.

For the complete mathematical framework of graph planarity and Kasteleyn orientations, we direct the reader to the foundational works in this field. The extension of Kasteleyn's theorem to non-planar graphs of genus $g$ is presented in~\cite{CimasoniReshetikhin2007,CimasoniReshetikhin2008}, which derive partition and correlation functions for such surfaces; our work focuses specifically on the planar case $g=0$. Linear-time planarity testing algorithms are thoroughly treated in Hopcroft and Tarjan~\cite{HopcroftTarjan1974} and Boyer and Myrvold~\cite{BoyerMyrvold2004}, while the fundamental characterizations of planarity via forbidden subgraphs and minors are established in Kuratowski~\cite{Kuratowski1930} and Wagner~\cite{Wagner1937}. For Kasteleyn orientations in planar graphs, the original constructions by Kasteleyn~\cite{Kasteleyn1961,Kasteleyn1963} provide the essential methods, with extensions to non-planar bipartite graphs through Little's forbidden subgraph characterization~\cite{Little1973}. The complete theory of combinatorial embeddings, rotation systems, face enumeration, and practical implementation aspects can be found in these references, along with treatments of special cases including multigraphs and certification methods for Pfaffian orientations.

\section{The Monomer--Dimer Model, Planar Monomer-Dimer model, and Pfaffian Techniques II}\label{Sec:3}

Having established the dimer model framework in the previous section, we now extend our analysis to include monomer configurations. The general monomer-dimer partition function is formulated using the \textit{Monobisyzexant} (Definition~\ref{def:Mbsz}) of the adjacency matrix $\mathbf{A}$ and a diagonal matrix $\mathbf{D}$ containing monomer fugacities. The correlation functions can be obtained by the \textit{monobisyzexantinho} (Definition~\ref{def:Mbszinho}), which is the Monobisyzexant of the submatrix of the adjacency matrix $\mathbf{A}$ and a submatrix of the diagonal matrix $\mathbf{D}$. While Theorem~\ref{thm:main2} provides a theoretical foundation for the planar monomer-dimer model through the Kasteleyn matrix formalism, the computational complexity remains formidable even for planar graphs. This inherent difficulty motivates our focus on two computationally tractable special cases: dimer with fixed monomer configurations with fermionic sources and fixed monomer configurations with bosonic sources. Both cases represent constrained versions of the full planar monomer-dimer model that admit efficient computation through their proportionality to Pfaffians, thereby bypassing the computational intractability of the general problem while retaining physical relevance.

\subsection{The monomer--dimer model}

The monomer--dimer model's significance stems from its deep connections to the Ising model. The two-dimensional Ising model maps exactly to a dimer model on a decorated lattice~\cite{Kasteleyn1961,Fisher1961,TemperleyFisher1961,McCoy1973}, while the square-lattice dimer model itself maps to an eight-vertex model~\cite{Baxter1968,Wu1971}. Furthermore, the Ising model in an external magnetic field admits a direct reformulation as a monomer-dimer system~\cite{HeilmannLieb1972}.

For finite monomer densities, closed-form solutions exist only in one dimension (via Chebyshev polynomials~\cite{Alberici2012}) and on complete or tree-like graphs~\cite{Alberici2013}. The transfer-matrix method expresses the partition function via its dominant eigenvalue~\cite{Lieb1967}, while Baxter's corner transfer matrix approach gives accurate approximations for thermodynamic quantities like dimer density~\cite{Baxter1968}. In D-dimensions with $D\geq3$, no exact solution exists, even for the close-packed case. Analytical results are available for single monomers at the boundary~\cite{Wu_2011,Wu2006}, boundary monomer correlations~\cite{Priezzhev2008}, and bulk monomer localization~\cite{Bouttier2007,Poghosyan2011}. Rigorous studies include partition function zeros~\cite{Heilmann1970,HeilmannLieb1972}, series expansions~\cite{Nagle1966}, and recursion relations~\cite{Ahrens1981}.

The monomer--dimer model generalizes the dimer model by incorporating a finite monomer density, substantially increasing the combinatorial complexity of the system. For general graphs, the monomer-dimer partition function is \#P-complete~\cite{Jerrum1987}, rendering it computationally intractable and markedly more complex than the pure dimer case. Numerical studies of monomer-monomer correlations at finite density show an exponential decay, consistent with mean-field predictions from Grassmann-variable formulations~\cite{Papanikolaou2007,Krauth2003}.

\begin{definition}[Monomer-Dimer Partition Function]
The partition function of the monomer-dimer model on a graph $G = (V,E)$ with edge weights $\{w_e\}_{e \in E}$ and monomer weights $\{x_v\}_{v \in V}$ is defined as:
\[
Z_{MD}(G) = \sum_{M \in \mathcal{M}(G)} \left(\prod_{e \in M} w_e\right) \left(\prod_{v \in V \setminus V(M)} x_v\right),
\]
where $\mathcal{M}(G)$ denotes the set of all matchings (including the empty matching) of $G$, and $V(M)$ is the set of vertices covered by matching $M$. For unweighted graphs ($w_e = 1$, $x_v = 1$), this counts all monomer-dimer configurations.
\end{definition}

\begin{definition}[Monomer Correlation Function in Monomer-Dimer Model]
For a set of vertices $\{v_1, v_2, \dots, v_k\}$, the monomer correlation function gives the probability that all specified vertices are unoccupied by dimers:
\[
\langle v_1 v_2 \cdots v_k \rangle_{MD} = \frac{\left(\prod_{v \in S} x_v\right) Z_{MD}(G \setminus S)}{Z_{MD}(G)},
\]
where $S = \{v_1, \dots, v_k\}$ and $G \setminus S$ denotes the graph with vertices $S$ and all incident edges removed.
\end{definition}

\begin{definition}[Dimer Correlation Function in Monomer-Dimer Model]
For a set of distinct edges $\{e_1, e_2, \dots, e_m\}$, the dimer correlation function gives the probability that all specified edges are occupied by dimers:
\[
\langle e_1 e_2 \cdots e_m \rangle_{MD} = \frac{\left(\prod_{e \in T} w_e\right) Z_{MD}(G \setminus V(T))}{Z_{MD}(G)},
\]
where $T = \{e_1, \dots, e_m\}$ and $G \setminus V(T)$ denotes the graph with all vertices incident to edges in $T$ removed.
\end{definition}

\begin{definition}[Mixed Monomer-Dimer Correlation Function]
For a set of vertices $\{v_1, \dots, v_k\}$ and distinct edges $\{e_1, \dots, e_m\}$ with $\{v_1, \dots, v_k\} \cap V(\{e_1, \dots, e_m\}) = \emptyset$, the mixed correlation function gives the joint probability that all specified vertices are monomers and all specified edges are dimers:
\[
\langle v_1 \cdots v_k; e_1 \cdots e_m \rangle_{MD} = \frac{\left(\prod_{v \in S} x_v\right) \left(\prod_{e \in T} w_e\right) Z_{MD}(G \setminus (S \cup V(T)))}{Z_{MD}(G)},
\]
where $S = \{v_1, \dots, v_k\}$, $T = \{e_1, \dots, e_m\}$, and $G \setminus (S \cup V(T))$ denotes the graph with both specified vertices and edges removed.
\end{definition}

We now formulate the monomer-dimer model within the framework of Berezin integration over Grassmann variables:

\begin{theorem}[Berezin integral over Grassmann variables representation of the monomer--dimer model~\cite{Samuel1980, Samuel1980b}]
The partition function for the monomer--dimer model (or full monomer--dimer model) sums over all possible monomer-dimer coverings, from the empty configuration (all monomers) to perfect matchings (all dimers). This partition function is given by
\begin{equation}
    Z_{MD}(\mathbf{D},\mathbf{A})=\operatorname{Mbsz}(\mathbf{D},\mathbf{A}),
\end{equation}
where $\operatorname{Mbsz}(\mathbf{D},\mathbf{A})$ is the Monobisyzexant function (Definition~\ref{def:Mbsz}) of the Adjacency matrix $\mathbf{A}$ (for dimer connections) and the diagonal matrix $\mathbf{D}$ encoding monomer locations (1 for a monomer, 0 otherwise).

The partition function admits a Grassmann--Berezin representation:
\begin{equation}\label{Mbsz-int}
\operatorname{Mbsz}(\mathbf{D},\mathbf{A})=\int\mathbf{D}(\boldsymbol{\chi},\bar{\boldsymbol{\chi}}),
e^{\bar{\boldsymbol{\chi}}^{\top}\mathbf{D}\boldsymbol{\chi}
+\frac{1}{2}(\bar{\boldsymbol{\chi}}\boldsymbol{\chi})^{\top}\mathbf{A}(\bar{\boldsymbol{\chi}}\boldsymbol{\chi})},\quad\mathbf{D}(\boldsymbol{\chi},\bar{\boldsymbol{\chi}}):=\prod_{i=1}^L d\bar{\chi}_id\chi_i.
\end{equation}

The correlation function is defined as
\begin{equation} \left\langle\prod_{\alpha=1}^{r}\Bar{\chi}_{i_\alpha}\chi_{i_\alpha}\right\rangle_{MD}:=\frac{1}{Z_{MD}(\mathbf{D},\mathbf{A})}\int\mathbf{D}(\boldsymbol{\chi},\bar{\boldsymbol{\chi}})\left(\prod_{\alpha=1}^{r}\Bar{\chi}_{i_\alpha}\chi_{i_\alpha}\right)\, e^{\Bar{\boldsymbol{\chi}}^\top\mathbf{D}\boldsymbol{\chi} +\frac{1}{2} (\Bar{\boldsymbol{\chi}}\boldsymbol{\chi})^\top\mathbf{A} (\Bar{\boldsymbol{\chi}}\boldsymbol{\chi})} =\operatorname{Mbsz}\left(\left(\mathbf{D},\mathbf{A}\right)_{\left[I^c_{[r]}\right]}\right)
\end{equation}
where $I=\{i_1,\dots,i_r\}$, and is equivalent to a Monobysexantinho (Definition~\ref{def:Mbszinho}).

From this fundamental monomer correlation, we can derive all other correlation functions in the monomer--dimer model:

\begin{itemize}
\item \textbf{Monomer correlations}: For vertices $v_1,\dots,v_r$, the correlation $\langle v_1 v_2 \cdots v_r \rangle$ gives the probability that all specified vertices are unoccupied by dimers (i.e., host monomers). In the Grassmann representation, each monomer at vertex $v$ corresponds to an insertion of $\Bar{\chi}_v\chi_v$, representing the absence of dimer coverage at that site:
\[
\langle v_1 v_2 \cdots v_r \rangle_{MD}:= \left\langle \prod_{\alpha=1}^r \Bar{\chi}_{v_\alpha}\chi_{v_\alpha} \right\rangle_{MD}.
\]
These correlations can involve odd numbers of vertices, unlike in pure dimer models.

\item \textbf{Dimer correlations}: For edges $e_1=(i_1,j_1),\dots,e_k=(i_k,j_k)$, the correlation $\langle e_1 e_2 \cdots e_k \rangle$ measures the probability that all specified edges are occupied by dimers. In the Grassmann formalism, each dimer on edge $(i,j)$ corresponds to the quartic term $\Bar{\chi}_i\chi_i\Bar{\chi}_j\chi_j$:
\[
\langle e_1 e_2 \cdots e_k \rangle_{MD}:= \left\langle \prod_{\ell=1}^k \Bar{\chi}_{i_\ell}\chi_{i_\ell}\Bar{\chi}_{j_\ell}\chi_{j_\ell} \right\rangle_{MD}.
\]
These correlations are always between pairs of vertices and represent the presence of specific dimer configurations.

\item \textbf{Mixed correlations}: For combinations of monomers and dimers, the correlation $\langle v_1 \cdots v_r; e_1 \cdots e_k \rangle$ gives the joint probability that specified vertices are unoccupied while specified edges are dimer-covered. This hybrid correlation combines both monomer and dimer insertions:
\[
\langle v_1 \cdots v_r; e_1 \cdots e_k \rangle_{MD}:= \left\langle \left(\prod_{\alpha=1}^r \Bar{\chi}_{v_\alpha}\chi_{v_\alpha}\right) \left(\prod_{\ell=1}^k \Bar{\chi}_{i_\ell}\chi_{i_\ell}\Bar{\chi}_{j_\ell}\chi_{j_\ell}\right) \right\rangle_{MD}.
\]
These mixed correlations capture the interplay between monomer and dimer degrees of freedom, essential for understanding the full statistical mechanics of the system.
\end{itemize}

\textbf{Proof Strategy:} The verification proceeds through systematic expansion of the exponential in the partition function, followed by term-wise Berezin integration. This process reveals the combinatorial structure encoded in the Monobisyzexant function (Definition~\ref{def:Mbsz}), with correlation functions naturally emerging as Monobisyzexantinhos through analogous expansion and integration techniques.
\end{theorem}

The quartic term makes evaluation nontrivial, analogous to the Hafnian problem. Applying a fermionic dual transformation followed by bosonization yields more accurate approximations than direct methods~\cite{najafi2024, Sourlas2019}.

\subsection{From Monobisyzexant to Modified Kasteleyn Pfaffians and the Planar monomer-dimer model}

The Kasteleyn matrix formalism enables the enumeration of monomer--dimer coverings on planar graphs by extending techniques developed for the pure dimer model. Building on the established mappings from Hafnians to Pfaffians and Hafnianinhos to Pfaffianinhos, the Monobisyzexant function—which expands in terms of minors, Hafnianinhos, and Hafnians (Lemma~\ref{Mbszexpanded})—admits a Pfaffian representation through Theorems~\ref{thm:main} and~\ref{PropHafPf}. This representation expresses the Mbsz as a sum over all monomer configurations, where each term corresponds to a dimer model with fixed monomers. Despite this elegant mathematical formulation, the partition function expands into an exponential number of Pfaffian terms, rendering exact computation infeasible for large systems and reflecting the \#P-completeness of the monomer-dimer model~\cite{Jerrum1987}.

\begin{theorem}[Monobisyzexant-Pfaffians correspondence for Planar monomer-dimers]\label{thm:main2}
Let $G = (V,E)$ be a planar graph with $|V|$ even, equipped with a symmetric weight matrix $\mathbf{A}$ satisfying $A_{ij} \geq 0$ and $A_{ii} = 0$, and a diagonal matrix $\mathbf{D}$ (1 for a monomer, 0 otherwise). Let $\mathbf{K}$ be a Kasteleyn matrix for $G$ constructed via $K_{ij}=\epsilon^{ij}A_{ij}$, where $\mathbf{S}$ is a Kasteleyn sign matrix. Then the following are equivalent:

\begin{enumerate}[label=(\roman*)]
\item \textbf{Combinatorial form}: The monomer--dimer partition function using Kasteleyn matrices is:

(1) For even $L$:
\begin{equation}
\operatorname{Mbsz}\!\left(\mathbf{D},\mathbf{A}\right)
=\det(\mathbf{D})
+\sum_{r=1}^{L/2-1}
\sum_{i_1<\cdots<i_{2r}=1}^L
\det\!\left(\mathbf{D}_{[I_{[2r]}]}\right)
(-1)^r\pf(\mathbf{K'}_{[2r]})\pf\left((\mathbf{K'}^{-1}_{[2r]})_{\left[I_{[2r]}\right]}\right)
+\pf(\mathbf{K}),
\end{equation}
where $\mathbf{K'}_{[\alpha]}$ is the modified Kasteleyn matrix with $\alpha$ monomers, and $I_{[2r]}=\{i_1,\dots,i_{[2r]}\}$.

\noindent
(2) For odd $L$:
\begin{equation}
\operatorname{Mbsz}\left(\mathbf{D},\mathbf{A}\right)
=\det(\mathbf{D})
+\sum_{r=1}^{(L-1)/2}
\sum_{i_1<\cdots<i_{[2r-1]}=1}^L
\epsilon(I_{[2r-1]})\det\left(\mathbf{D}_{[I_{[2r-1]}]}\right)
\pf\left((\mathbf{K'}_{[2r-1]})_{\left[I^c_{[2r-1]}\right]}\right),
\end{equation}
where $I_{[2r-1]}=\{i_1,\dots,i_{[2r-1]}\}$.

\item \textbf{Grassmann-Berezin form}: The relation between the integrals is

(1) For even $L$ (\cite{AllegraFortin2014} suggest this application):
\begin{equation}
\begin{split}
\int\mathbf{D}(\boldsymbol{\chi},\bar{\boldsymbol{\chi}})&
e^{\Bar{\boldsymbol{\chi}}^\top\mathbf{D}\boldsymbol{\chi}+\frac{1}{2}(\bar{\boldsymbol{\chi}}\boldsymbol{\chi})^\top\mathbf{A}(\bar{\boldsymbol{\chi}}\boldsymbol{\chi})}\\
&=\prod_{i=1}^LD_{ii}+\sum_{r=1}^{L/2-1}
\sum_{i_1<\cdots<i_{2r}=1}^L
\left(\prod_{\alpha=1}^{2r}D_{i_\alpha i_\alpha}\right)\int\mathbf{D}\boldsymbol{\chi} 
\left(\prod_{\alpha=1}^{2r}\chi_{i_\alpha}\right)e^{\frac{1}{2}\boldsymbol{\chi}^\top\mathbf{K}'_{[2r]}\boldsymbol{\chi}}+\int \mathbf{D}\boldsymbol{\chi} e^{\frac{1}{2}\boldsymbol{\chi}^\top\mathbf{K}\boldsymbol{\chi}}.
\end{split}
\end{equation}

(2) For odd $L$:
\begin{equation}
\begin{split}
\int\mathbf{D}(\boldsymbol{\chi},\bar{\boldsymbol{\chi}})&
e^{\Bar{\boldsymbol{\chi}}^\top\mathbf{D}\boldsymbol{\chi}+\frac{1}{2}(\bar{\boldsymbol{\chi}}\boldsymbol{\chi})^\top\mathbf{A}(\bar{\boldsymbol{\chi}}\boldsymbol{\chi})}\\
&=\prod_{i=1}^LD_{ii}+\sum_{r=1}^{(L-1)/2}
\sum_{i_1<\cdots<i_{[2r-1]}=1}^L
\left(\prod_{\alpha=1}^{2r-1}D_{i_\alpha i_\alpha}\right)\int\mathbf{D}\boldsymbol{\chi} 
\left(\prod_{\alpha=1}^{2r-1}\chi_{i_\alpha}\right)e^{\frac{1}{2}\boldsymbol{\chi}^\top\mathbf{K}'_{[2r-1]}\boldsymbol{\chi}}.
\end{split}
\end{equation}
\end{enumerate}

For $\mathbf{D}=0$ we recover the Theorem~\ref{thm:main}.

The proof is in Appendix~\ref{Proof:thm:main2}.
\end{theorem}

To analyse fixed monomer configurations in the dimer model, we apply Theorem~\ref{thm:main2} to both even and odd monomer cases:

\begin{align*}
&\text{Even monomers (2r):} \quad \left(\prod_{\alpha=1}^{2r}D_{i_\alpha i_\alpha}\right)
\int\mathbf{D}\boldsymbol{\chi} 
\left(\prod_{\alpha=1}^{2r}\chi_{i_\alpha}\right)
e^{\frac{1}{2}\boldsymbol{\chi}^\top\mathbf{K}'_{[2r]}\boldsymbol{\chi}}, \\
&\text{Odd monomers (2r-1):} \quad \left(\prod_{\alpha=1}^{2r-1}D_{i_\alpha i_\alpha}\right)
\int\mathbf{D}\boldsymbol{\chi} 
\left(\prod_{\alpha=1}^{2r-1}\chi_{i_\alpha}\right)
e^{\frac{1}{2}\boldsymbol{\chi}^\top\mathbf{K}'_{[2r-1]}\boldsymbol{\chi}}.
\end{align*}

We can rewrite the term
\begin{equation*}
\prod_{\alpha=1}^{2r}D_{i_\alpha i_\alpha}\chi_{i_\alpha}=\int\mathbf{D}\boldsymbol{\psi}e^{\sum_{\alpha=1}^{2r}\psi_{i_\alpha}D_{i_\alpha i_\alpha}\chi_{i_\alpha}}=\int\mathbf{D}\boldsymbol{\psi}e^{\boldsymbol{\psi}^T\mathbf{H}\boldsymbol{\chi}},
\end{equation*}
where $\boldsymbol{\psi}^T\mathbf{H}\boldsymbol{\chi}=\sum_{\alpha=1}^{2r}\sum_{j=1}^{L}\psi_{i_\alpha}H_{i_\alpha j}\chi_{j}=\sum_{\alpha=1}^{2r}\psi_{i_\alpha}D_{i_\alpha i_\alpha}\chi_{i_\alpha}$, we define $\mathbf{H}$ as the proyection matrix:
\[
\mathbf{H}_{i_\alpha j}:= 
\begin{cases}
D_{i_\alpha i_\alpha} & \text{if monomer } i_\alpha \text{ is at vertex } j \\
0 & \text{otherwise}
\end{cases}
\]

Then, for an even dimension, we can write the partition function for the dimer model with $2r$-monomer fixed as\cite{AllegraFortin2014}: 
\begin{equation*}
\int\mathbf{D}\boldsymbol{\chi} 
\left(\prod_{\alpha=1}^{2r}\chi_{i_\alpha}\right)e^{\frac{1}{2}\boldsymbol{\chi}^\top\mathbf{K}'_{[2r]}\boldsymbol{\chi}}=\int\mathbf{D}\boldsymbol{\psi}\int\mathbf{D}\boldsymbol{\chi}e^{\frac{1}{2}\boldsymbol{\chi}^\top\mathbf{K}'_{[2r]}\boldsymbol{\chi}+\boldsymbol{\psi}^T\mathbf{H}\boldsymbol{\chi}},
\end{equation*}
The expression above appears in a general way in Corollary~\ref{exponential with fermionic sources}.

For odd dimensions, taking just one monomer $D_{ii}=u_i$, we have
\begin{equation*}
    u_i
\int\mathbf{D}\boldsymbol{\chi}\chi_{i}
e^{\frac{1}{2}\boldsymbol{\chi}^\top\mathbf{K}'_{[2r-1]}\boldsymbol{\chi}}.
\end{equation*}

If we start from the Corollary~\ref{exponential with bosonic sources}, we obtain
\begin{equation*}
\int\mathbf{D}\boldsymbol{\chi} e^{\frac{1}{2}\boldsymbol{\chi}^\top\mathbf{K}\boldsymbol{\chi}+\mathbf{u}^\top \boldsymbol{\chi}}=\begin{cases}
\int\mathbf{D}\boldsymbol{\chi} e^{\frac{1}{2}\boldsymbol{\chi}^\top\mathbf{K}\boldsymbol{\chi}}, & \text{even $L$}, \\
\sum_{i=1}^L(-1)^{i+1}u_{i}\int\mathbf{D}\boldsymbol{\chi} 
\chi_{i}e^{\frac{1}{2}\boldsymbol{\chi}^\top\mathbf{K}\boldsymbol{\chi}}, & \text{odd $L$}.
\end{cases}
\end{equation*}
This result encodes the dimer model with fixed monomers, where $u_i$ represents the monomer weight at vertex $i$. For even $L$, the absence of monomer terms reflects the perfect matching constraint of dimer theory. For odd $L$, the explicit monomer contributions enable the extension to monomer-dimer configurations. While fermionic sources could alternatively construct this model, the bosonic field approach demonstrates the formalism's versatility and provides the foundation for our subsequent results.

In what follows, we present two distinct approaches for handling monomer configurations: fermionic sources and bosonic sources. These simplified cases serve as motivation for the more general structures developed in Section~\ref{Sec:6}, where we establish comprehensive frameworks that may find applications in other constrained statistical models beyond the dimer systems considered here.


\subsubsection{Planar dimer with fixed monomers model}
A Grassmann formulation~\cite{AllegraFortin2014} provides an exact solution as a product of two explicit Pfaffians, first for the close-packed limit and later with nonuniform fugacities. Early fermionic representations~\cite{Samuel1980,Samuel1980b} used Grassmann variable pairs per site to enforce single occupancy, leading to a quartic interaction and diagrammatic expansions.

Building upon the approach introduced in~\cite{AllegraFortin2014}, we employ Corollary~\ref{exponential with fermionic sources} to motivate the more general results of Theorem~\ref{exponential with linear fermionic and bosonic sources}. Integrating out the Grassmann variables yields the partition function for the dimer model with fixed monomers:
\begin{equation}
    Z_{D-M}(\mathbf{K}'):=\int\mathbf{D}\boldsymbol{\psi}\int\mathbf{D}\boldsymbol{\chi} 
    e^{\frac{1}{2}\boldsymbol{\chi}^T\mathbf{K}'\boldsymbol{\chi}+\boldsymbol{\psi}^T\mathbf{H}\boldsymbol{\chi}}
    =\pf(\mathbf{K}')\pf(\mathbf{C}),
\end{equation}
where $\mathbf{D}\boldsymbol{\chi}:= d\chi_{L} \cdots d\chi_1$, $\mathbf{D}\boldsymbol{\psi}:= d\psi_{r} \cdots d\psi_1$, $L = X \times Y$ is the total number of vertices in the square lattice with dimensions $X \times Y$, and $r$ is the number of monomers (both even). Here $\boldsymbol{\psi}^T$ is an $r$-dimensional Grassmann vector, $\mathbf{H}$ is an $r\times L$ matrix, $\mathbf{K}'$ is an $L\times L$ antisymmetric matrix encoding dimer interactions, and $\mathbf{C}=\mathbf{H}\mathbf{K}'^{-1}\mathbf{H}^\top=(\mathbf{K}'^{-\top})_{\left[I_{[2r]}\right]}$. This factorization cleanly separates the dimer contribution through $\pf(\mathbf{K}')$ from the monomer contribution through $\pf(\mathbf{C})$, providing a compact algebraic representation of the constrained system.

\paragraph{Modified Kasteleyn Matrix $\mathbf{K}'$ for Rectangular Lattices:}
For an $X \times Y$ square lattice with vertices labeled by coordinates $(x,y)$ where $x = 1,\dots,X$ and $y = 1,\dots,Y$, we modify the original Kasteleyn matrix $\mathbf{K}$ to $\mathbf{K}'$ following \cite{AllegraFortin2014}:

\begin{enumerate}
    \item \textbf{Define the frontier} as the entire right boundary: $F = \{(X,y) \mid y = 1,\dots,Y\}$.
    
    \item \textbf{For each monomer} at position $(x_i, y_i)$:
    \begin{itemize}
        \item If $x_i = X$ (monomer on frontier), no modification occurs.
        \item If $x_i < X$, construct a horizontal defect line from $(x_i, y_i)$ to $(X, y_i)$.
        \item For each horizontal edge in this path, i.e., edges between $(k, y_i)$ and $(k+1, y_i)$ for $k = x_i, \dots, X-1$, multiply the corresponding elements $\mathbf{K}_{uv}$ and $\mathbf{K}_{vu}$ by $-1$.
    \end{itemize}
    
    \item If multiple defect lines traverse the same edge, the sign changes accumulate multiplicatively ($-1$ for odd crossings, $+1$ for even).
\end{enumerate}

This specific construction ensures $\mathbf{K}'$ remains invertible and yields the correct enumeration of dimer configurations with fixed monomers.

\paragraph{Monomer Projection Matrix $\mathbf{H}$:}
The $r \times L$ matrix $\mathbf{H}$ encodes monomer positions, where $L = X \times Y$. For a monomer at coordinate $(x_i,y_i)$ with linear index $s = (x_i-1)Y + y_i$:
\[
\mathbf{H}_{ij} = 
\begin{cases}
1 & \text{if monomer } i \text{ is at vertex } j \\
0 & \text{otherwise}
\end{cases}
\]
Each row of $\mathbf{H}$ corresponds to a monomer and has a single 1 at the column corresponding to its vertex position.

This approach directly implements the method of \cite{AllegraFortin2014} for rectangular lattices, using their designated right boundary as the frontier and horizontal defect lines. The product $\pf(\mathbf{K}')\pf(\mathbf{C})$ then gives the exact enumeration of dimer configurations with the specified monomer positions.

\paragraph{Correlation Function and Disorder Operators:}
The correlation functions between monomers are obtained by introducing disorder operators that account for the defect lines. We define the correlation function for the dimer with fixed monomers model as:

\begin{equation}
\left\langle\left(\prod_{\alpha=1}^r\chi_{i_\alpha}\right)e^{\frac{1}{2}\boldsymbol{\chi}^T\mathbf{V}\boldsymbol{\chi}} \right\rangle_{D}:=\frac{Z_{D-M}(\mathbf{K}')}{Z_D(\mathbf{K})},
\end{equation}

where $\mathbf{V}$ is the antisymmetric matrix encoding the defect lines from monomers to the frontier, explicitly $\mathbf{V}=\mathbf{K}'-\mathbf{K}$, and $Z_D = \pf(\mathbf{K})$ is the pure dimer partition function. The exponential term $e^{\frac{1}{2}\boldsymbol{\chi}^T\mathbf{V}\boldsymbol{\chi}}$ represents the disorder operator that introduces the necessary sign changes along the defect paths. We saw this term in a more general form, as a monomer correlation function, in eq.~\eqref{2rmonomercorrelation}.

\begin{figure}[]
    \centering
    \includegraphics[width=1\linewidth]{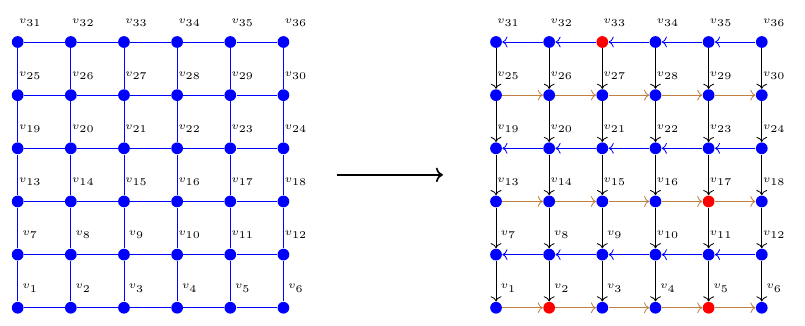}
    \caption{(Left) Original square lattice graph. (Right) Kasteleyn-oriented lattice with edge directions satisfying the odd clockwise orientation condition for each face. For square lattices, a convenient construction assigns all vertical edges downward and alternates horizontal edges between left and right directions. Monomers are inserted at vertices $v_2$, $v_5$, $v_{17}$, and $v_{33}$ in the right graph.}
    \label{Fig:KOrientationSquare}
\end{figure} 

Consider the $6\times6$ square lattice example shown in Figure~\ref{Fig:KOrientationSquare}. We define the frontier as the entire right boundary: $F=\{(6,y)\mid y=1,2,\dots,6\}$, corresponding to vertices $\{v_6,v_{12},v_{18},v_{24},v_{30},v_{36}\}$. 

The insertion of four monomers modifies the original Kasteleyn matrix through defect lines extending horizontally from each monomer to the frontier. Specifically:
\begin{itemize}
    \item From monomer $v_{33}$: flip signs of $K_{(33)(34)}$, $K_{(34)(35)}$, and $K_{(35)(36)}$
    \item From monomer $v_{17}$: flip sign of $K_{(17)(18)}$
    \item From monomers $v_2$ and $v_5$ (same row): flip signs from $v_2$ to $v_5$ only, as subsequent edge $K_{56}$ receives two sign flips and remains unchanged
\end{itemize}

The monomer projection matrix $\mathbf{H}$ is a $4\times 36$ matrix with nonzero entries at positions corresponding to monomer locations: $H_{12}=1$, $H_{25}=1$, $H_{3(17)}=1$, and $H_{4(33)}=1$, with all other entries zero.

Although the construction is illustrated for rectangular lattices, the method generalizes to arbitrary planar graphs. The specific path chosen from a monomer to the graph boundary is arbitrary and does not affect the final enumeration, as different paths related by gauge transformations yield equivalent sign patterns. The essential requirement is that defect lines terminate at the graph's natural boundary, ensuring proper counting of dimer configurations with fixed monomer positions.
\subsubsection{Planar dimer with monomer fugacity model}

For graphs with an even number of vertices $|V|$, the classical dimer model admits perfect matchings~\cite{LovaszPlummer1986,BondyMurty}. When $|V|$ is odd, perfect matchings cannot exist, and the maximal coverings leave exactly one vertex unmatched, forming configurations known as \textit{almost-perfect matchings}~\cite{LovaszPlummer1986}.

To incorporate monomer degrees of freedom, we introduce a bosonic source vector $\mathbf{u} \in \mathbb{C}^L$ and define the augmented Pfaffian partition function:
\begin{equation}
Z_{D-F}(\mathbf{K},\mathbf{u}):= \int \mathbf{D}\boldsymbol{\chi}\, e^{\frac{1}{2}\boldsymbol{\chi}^\top \mathbf{K} \boldsymbol{\chi} + \mathbf{u}^\top \boldsymbol{\chi}}
=
\begin{cases}
\pf(\mathbf{K}), & \text{even $L$}, \\
\pf(\mathbf{X}), & \text{odd $L$},
\end{cases}
\end{equation}
where for odd $L$, the extended antisymmetric matrix $\mathbf{X}$ is:
\begin{equation*}
\mathbf{X} = 
\begin{pmatrix}
0 & \mathbf{u}^\top \\
-\mathbf{u} & \mathbf{K}
\end{pmatrix}.
\end{equation*}

The Pfaffian of $\mathbf{X}$ admits a border expansion:
\begin{equation}\label{eq:pfaffian_border_expansion}
\pf(\mathbf{X}) = \sum_{j=1}^{L} (-1)^{j+1} u_j \, \pf\left(\mathbf{K}_{[\{j\}^c]}\right),
\end{equation}
where $\mathbf{K}_{[\{j\}^c]}$ denotes the principal submatrix obtained by deleting row and column $j$. Each term enumerates almost-perfect matchings with the monomer fixed at vertex $j$. The sign factor $(-1)^{j+1}$ emerges from the monomer insertions at positions $j$, consistent with our earlier discussion of the modified Kasteleyn matrix construction.

The source $\mathbf{u}$ serves as a \textit{monomer fugacity}, controlling the statistical weight of monomer positions. Correlation functions are defined by:
\begin{equation}
\left\langle \prod_{\alpha=1}^{r} \chi_{i_\alpha} \right\rangle_{D-F}
:=\frac{1}{Z_{D-F}} \int \mathbf{D}\boldsymbol{\chi} \left( \prod_{\alpha=1}^{r} \chi_{i_\alpha} \right) 
e^{\frac{1}{2}\boldsymbol{\chi}^\top \mathbf{K} \boldsymbol{\chi} + \mathbf{u}^\top \boldsymbol{\chi}},
\end{equation}
which evaluates to:
\begin{equation}
\left\langle \prod_{\alpha=1}^{r} \chi_{i_\alpha} \right\rangle_{D-F}
=
\begin{cases}
(-1)^{L}\epsilon(I_{[r]})\frac{\pf\left(\mathbf{K}_{\left[I_{[r]}^c\right]}\right)}{Z_{D-F}(\mathbf{K},\mathbf{u})}, & \text{even $r+L$}, \\
(-1)^{L+1}\epsilon(I_{[r]})\frac{\pf\left(\mathbf{X}_{\left[I_{[r]}^c\right]}\right)}{Z_{D-F}(\mathbf{K},\mathbf{u})}, & \text{odd $r+L$},
\end{cases}
\end{equation}
where $I_{[r]} = \{i_1, \dots, i_r\}$ denotes monomer positions and $\epsilon(I_{[r]})$ is the Grassmann reordering sign factor.

\paragraph{Even $L$ (Perfect Matchings)}
For even $L$, the partition function reproduces the pure dimer model, as all vertices can be perfectly paired~\cite{HeilmannLieb1972,LovaszPlummer1986,BondyMurty}. The bosonic source $\mathbf{u}$ couples only to unpaired vertices, which are absent in perfect matchings, leaving $Z_{D-F}(\mathbf{K},\mathbf{u})$ identical to the pure dimer partition function~\cite{Samuel1980, hayn1997}.

Correlation functions exhibit parity-dependent behavior:
\begin{itemize}
    \item \textbf{Odd-point correlators} ($r$ odd): Receive nontrivial contributions from $\mathbf{u}$, coupling to effectively unpaired sites
    \item \textbf{Even-point correlators} ($r$ even): Remain uncoupled from $\mathbf{u}$, reproducing pure dimer results~\cite{HeilmannLieb1972}
\end{itemize}

\paragraph{Odd $L$ (Almost-Perfect Matchings)}
For odd $L$, the bosonic source appears explicitly in both partition and correlation functions, as exactly one unpaired vertex exists for coupling. The partition function enumerates almost-perfect matchings: perfect matchings on $L-1$ vertices coupled to monomer fugacity $\mathbf{u}$ at the remaining vertex.

Correlation function behavior inverts:
\begin{itemize}
    \item \textbf{Even-point correlators} ($r$ even): Couple to $\mathbf{u}$ via the unpaired vertex
    \item \textbf{Odd-point correlators} ($r$ odd): Decouple from $\mathbf{u}$, reproducing pure dimer statistics
\end{itemize}

This formulation establishes a unified Pfaffian framework that interpolates continuously between pure dimer models and single-monomer configurations~\cite{Samuel1980,hayn1997,AllegraFortin2014}. 

Before proceeding to specific examples, it is crucial to address a fundamental aspect of the dimer model with fixed monomers. The introduction of monomers modifies the underlying graph structure, causing the standard Kasteleyn matrix to produce incorrect sign patterns in correlation functions. This is where the bosonic field plays an essential role: it systematically corrects these sign discrepancies, ensuring that the final calculation yields a proper sum over all configurations with positive weights. 

The bosonic source $\mathbf{u}$ provides the necessary degrees of freedom to compensate for the sign alterations induced by monomer insertions. When these fields are integrated out, the result is the exact enumeration of valid configurations. However, if the bosonic fields are retained in the calculation, they encode the full correlation structure rather than just the configuration count, offering additional physical insight into the monomer-dimer system.

\begin{figure}[]
    \centering
    \begin{subfigure}[b]{0.3\textwidth}
        \centering
        \includegraphics[width=\textwidth]{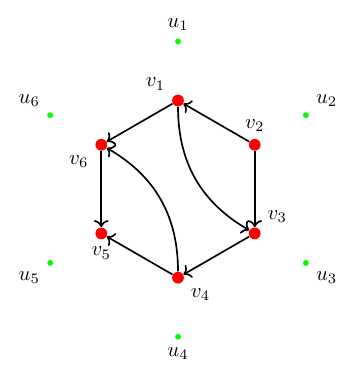}
        \caption{Kasteleyn-oriented graph $G$ with bosonic field vertices (green)}
        \label{Fig:adimerfugacity}
    \end{subfigure}
    \hfill
    \begin{subfigure}[b]{0.3\textwidth}
        \centering
        \includegraphics[width=\textwidth]{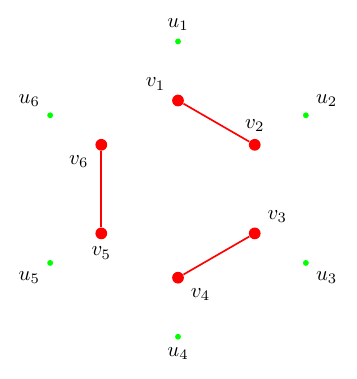}
        \caption{First perfect matching: bosonic vertices unpaired}
        \label{Fig:bdimerfugacity}
    \end{subfigure}
    \hfill
    \begin{subfigure}[b]{0.3\textwidth}
        \centering
        \includegraphics[width=\textwidth]{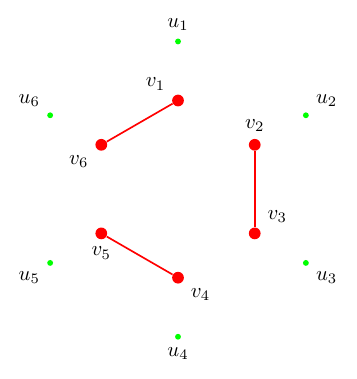}
        \caption{Second perfect matching: bosonic vertices unpaired}
        \label{Fig:cdimerfugacity}
    \end{subfigure}
    
    \vspace{1em}
    
    \begin{subfigure}[b]{0.3\textwidth}
        \centering
        \includegraphics[width=\textwidth]{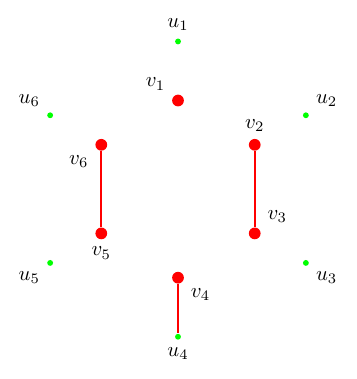}
        \caption{Almost perfect matching: bosonic dimer $(u_4,v_4)$}
        \label{Fig:ddimerfugacity}
    \end{subfigure}
    \hfill
    \begin{subfigure}[b]{0.3\textwidth}
        \centering
        \includegraphics[width=\textwidth]{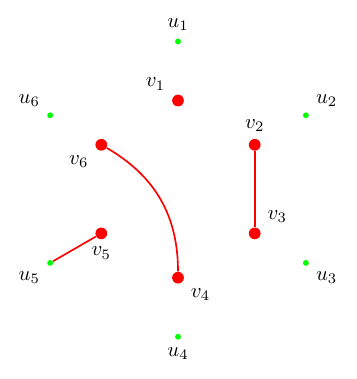}
        \caption{Almost perfect matching: bosonic dimer $(u_5,v_5)$}
        \label{Fig:edimerfugacity}
    \end{subfigure}
    \hfill
    \begin{subfigure}[b]{0.3\textwidth}
        \centering
        \includegraphics[width=\textwidth]{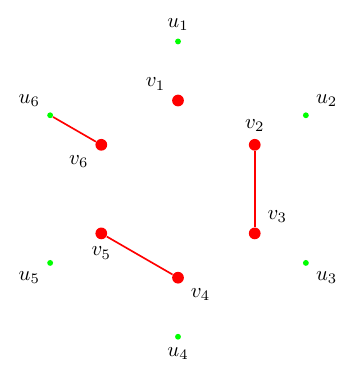}
        \caption{Almost perfect matching: bosonic dimer $(u_6,v_6)$}
        \label{Fig:fdimerfugacity}
    \end{subfigure}

    \hfill
    \begin{subfigure}[b]{0.3\textwidth}
        \centering
        \includegraphics[width=\textwidth]{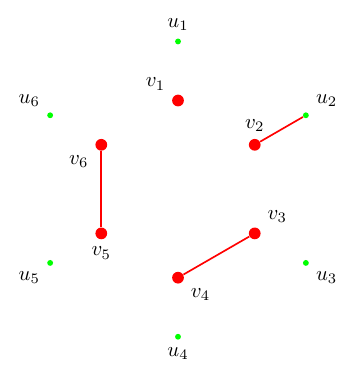}
        \caption{Almost perfect matching: bosonic dimer $(u_2,v_2)$}
        \label{Fig:gdimerfugacity}
    \end{subfigure}
    
    \caption{Dimer configurations with bosonic sources. (a-c) For even vertex count, perfect matchings leave bosonic vertices unpaired. (d-g) For odd vertex count with fixed monomer at $v_1$, bosonic vertices form dimers with fermionic vertices, generating almost perfect matchings.}
    \label{Fig:mainExample2}
\end{figure}

Consider the Kasteleyn-oriented graph in Figure~\ref{Fig:adimerfugacity} with bosonic field $\mathbf{u}=(u_1,u_2,u_3,u_4,u_5,u_6)$ represented by green vertices. The corresponding Kasteleyn matrix is
\begin{equation*}
    \mathbf{K}=
    \begin{pmatrix}
        0 & -K_{12} & K_{13} & 0 & 0 & K_{16} \\
        K_{12} & 0 & K_{23} & 0 & 0 & 0 \\
        -K_{13} & -K_{23} & 0 & K_{34} & 0 & 0 \\
        0 & 0 & -K_{34} & 0 & K_{45} & K_{46} \\
        0 & 0 & 0 & -K_{45} & 0 & -K_{56} \\
        -K_{16} & 0 & 0 & -K_{46} & K_{56} & 0
    \end{pmatrix},\quad \text{with } \pf(\mathbf{K})=K_{12}K_{34}K_{56}+K_{16}K_{23}K_{45}.
\end{equation*}
This yields two perfect matchings when we set $K_{ij}=1$ (Figures~\ref{Fig:bdimerfugacity} and~\ref{Fig:cdimerfugacity}), where bosonic vertices remain unpaired.

Now consider fixing a monomer at vertex $v_1$. In the classical dimer model, no perfect matching exists with an odd number of vertices. However, with bosonic sources, the bosonic field can couple with fermionic vertices to form almost perfect matchings, which means the Kasteleyn submatrix must be modified to preserve the sing. We compute these using the reduced matrices:
\begin{equation*}
\begin{split}
    \mathbf{K}'_{[\{1\}^c]}&=
    \begin{pmatrix}
        0 &{\color{red}-}K_{23} & 0 & 0 & 0 \\
        {\color{red}+}K_{23} & 0 &{\color{red}-}K_{34} & 0 & 0 \\
        0 &{\color{red}+}K_{34} & 0 &{\color{red}-}K_{45} & K_{46} \\
        0 & 0 &{\color{red}+}K_{45} & 0 &-K_{56} \\
        0 & 0 & -K_{46} &K_{56} & 0
    \end{pmatrix},\\
    \mathbf{X}_{[\{1\}^c]}&=
    \begin{pmatrix}
        0 & u_{2} & u_{3} & u_4 & u_5 & u_{6} \\
        -u_{2} & 0 &{\color{red}-}K_{23} & 0 & 0 & 0 \\
        -u_{3} &{\color{red}+}K_{23} & 0 &{\color{red}-}K_{34} & 0 & 0 \\
        -u_4 &0 &{\color{red}+}K_{34} & 0 &{\color{red}-}K_{45} & K_{46} \\
        -u_5 &0 & 0 &{\color{red}+}K_{45} & 0 &-K_{56} \\
        -u_{6} &0 & 0 & -K_{46} &K_{56} & 0
    \end{pmatrix}.
\end{split}
\end{equation*}
The Pfaffian "$\pf(\mathbf{X}_{[\{1\}^c]})=u_2K_{34}K_{56}+u_4K_{23}K_{56}+u_5K_{23}K_{46}+u_6K_{23}K_{45}$" enumerates almost perfect matchings. Choosing $u_i=1$ and setting $K_{ij}=1$ yields four almost perfect matchings illustrated in Figures~\ref{Fig:ddimerfugacity}--\ref{Fig:gdimerfugacity}, where bosonic vertices form dimers with fermionic vertices. 

It should be mentioned that we are treating the monomers as bosonic vertices that couple to the unpaired vertices; another way to tackle the problem is to place loops on all vertices.

\subsection{Hua Decomposition Method for Perfect and Almost Perfect Matchings}

Theorem~\ref{exponential with linear fermionic and bosonic sources} remains valid for singular skew-symmetric matrices $\mathbf{K}$, which commonly arise in physical systems through Kasteleyn matrices with nontrivial kernels. To handle these cases systematically, we employ a unitary block decomposition \cite{Hua1949} that separates the finite and null sectors of $\mathbf{K}$.

\begin{theorem}[Hua Decomposition for Fermionic Pfaffians with Mixed Sources]\label{HuaTheorem}
Let $\boldsymbol{\chi} = (\chi_1,\dots,\chi_{L})$ and $\boldsymbol{\psi} = (\psi_1,\dots,\psi_{L})$ be real Grassmann variables, and $\mathbf{u} = (u_1,\dots,u_L)$ be real or complex bosonic sources. For a singular skew-symmetric matrix $\mathbf{K}$ of rank $r$, consider the unitary transformation:
\[
\boldsymbol{\chi}' = \mathbf{U}^\ast \boldsymbol{\chi}, \quad
\mathbf{D}\boldsymbol{\chi}' = \det(\mathbf{U}) \mathbf{D}\boldsymbol{\chi}, \quad
\mathbf{U}^\ast = \begin{pmatrix} \tilde{\mathbf{U}} & \bar{\mathbf{U}} \end{pmatrix},
\]
where $\tilde{\mathbf{U}}$ and $\bar{\mathbf{U}}$ span the range and kernel of $\mathbf{K}$ respectively, yielding the block-diagonal form:
\[
\mathbf{U}^\dagger \mathbf{K} \mathbf{U}^\ast = \boldsymbol{\Sigma} =
\begin{pmatrix}
\boldsymbol{\Sigma}'_{r\times r} & \mathbf{0} \\
\mathbf{0} & \mathbf{0}
\end{pmatrix}, \quad
\boldsymbol{\Sigma}' = \bigoplus_{\alpha=1}^{r/2}
\begin{pmatrix}
0 & d_\alpha \\
-d_\alpha & 0
\end{pmatrix}.
\]

Then the Berezin integral with mixed sources decomposes as:
\begin{equation}\label{eq:practical_noninvertible}
\begin{split}
\int \mathbf{D}\boldsymbol{\chi}'\; & e^{\frac{1}{2}\boldsymbol{\chi}'^\top\mathbf{K}\boldsymbol{\chi}' + \mathbf{u}^\top\boldsymbol{\chi}' + \boldsymbol{\psi}^\top\boldsymbol{\chi}'}
= (\det \mathbf{U})\,\pf(\boldsymbol{\Sigma}') \times \\
& \Bigg[ \frac{(-1)^{\frac{(L-r)(L-r+1)}{2}}}{(L-r)!} 
\sum_{i_1,\dots,i_{L-r}=1}^{L} \det\!\bigl(\bar{\mathbf{U}}_{[I_{[L-r]},\star]}\bigr) 
\prod_{\alpha=1}^{L-r}\psi_{i_\alpha} \\
& \quad + \frac{(-1)^{\frac{(L-r)(L-r-1)}{2}}}{(L-r-1)!} 
\sum_{i_1,\dots,i_{L-r}=1}^{L} \det\!\bigl(\bar{\mathbf{U}}_{[I_{[L-r]},\star]}\bigr) 
u_{i_1}\!\!\prod_{\alpha=2}^{L-r}\psi_{i_\alpha} \Bigg]
e^{\frac{1}{2}\boldsymbol{\psi}^\top\mathcal{U}\boldsymbol{\psi} + \mathbf{u}^\top\mathcal{U}\boldsymbol{\psi}},
\end{split}
\end{equation}
where $\mathcal{U} = \tilde{\mathbf{U}}\boldsymbol{\Sigma}'^{-1}\tilde{\mathbf{U}}^\top$ is the pseudoinverse operator on the range space.
\end{theorem}

The Kasteleyn matrix $\mathbf{K}$ frequently becomes singular in physical scenarios involving open boundaries, zero-weight edges, or disconnected components. Applying the decomposition $\mathbf{K} = \mathbf{U}^\top \boldsymbol{\Sigma} \mathbf{U}$ provides both numerical stability and combinatorial insight:

- For full-rank $\mathbf{K}$ ($r=L$), only the Pfaffian term contributes, recovering standard dimer theory
- For singular $\mathbf{K}$ ($r<L$), Grassmann monomials describe almost-perfect matchings with $(L-r)$ unmatched vertices

Combinatorially, the kernel dimension $L-r$ identifies the minimal vertex set whose removal restores perfect matchings, while $\bar{\mathbf{U}}$ specifies the responsible vertices.

If we simplify Theorem~\ref{HuaTheorem} by setting the bosonic source to zero, and then integrating the sources, we obtain the partition function for the dimer model for fixed monomers via Hua decomposition:
\[
\int\mathbf{D}\boldsymbol{\psi}\int\mathbf{D}\boldsymbol{\chi}'\; e^{\frac{1}{2}\boldsymbol{\chi}'^\top \mathbf{K}\boldsymbol{\chi}'+\boldsymbol{\psi}^\top\boldsymbol{\chi}'}=(\det \mathbf{U})\, \pf(\boldsymbol{\Sigma}')\frac{(-1)^{\frac{(L-r)(L-r+1)}{2}}}{(L-r)!} \sum_{i_1,\dots,i_{L-r}=1}^{L} \det\!\bigl(\bar{\mathbf{U}}_{[I_{[L-r]},\star]}\bigr)\pf\left(\mathcal{U}_{\left[I^c_{[L-r]}\right]}\right),
\]
where $I=\{i_1,\dots,i_{L-r}\}$ indexes the singular sector or the monomers.

This unified framework preserves Pfaffian structure on reduced spaces and connects perfect/almost-perfect matchings through null subspace geometry, providing both analytical resolution and combinatorial interpretation for singular cases. This generalizes the result in~\cite{AllegraFortin2014} for singular matrices.

\section{Rooted Spanning Forests, Spanning Trees, and Determinant Techniques}\label{Sec:4}
Spanning forests and trees are among the most fundamental combinatorial objects in graph theory, with roots in Kirchhoff's matrix–tree theorem and a long development in algebraic graph theory~\cite{Kirchhoff1847,Biggs1993}. Beyond their classical combinatorial role, spanning trees and forest ensembles play a central role in physics — they underpin electrical-network theory and random-walk representations, appear in scaling limits and conformal-invariance problems, and serve as tractable toy models for transport, avalanches, and self-organized criticality~\cite{DoyleSnell1984,LyonsPeres2016}. This section gives a concise, pedagogical review of these models with emphasis on their interpretation as partition functions and correlators, and it introduces an alternative Berezin integral representation of the spanning-tree partition function and correlation function (Theorem) that complements algebraic approaches such as the all-minors matrix–tree and matrix–forest theorems~\cite{Chaiken1982,Chebotarev2006}.

\subsection{The Weighted Laplacian Matrix}

The natural starting point is the weighted Laplacian matrix \(\mathbf{L}^w\)
\footnote{A more extensive explanation with examples on how to construct the weighted Laplacian matrix has been prepared in the Appendix \ref{Laplacianmatrix}.}. 
For a graph \(G=(V,E)\) with edge weights \(w_{ab}\), it is defined by
\[
L^w_{ab} = -w_{ab}, \qquad a\neq b, \qquad 
L^w_{aa} = \sum_{b=1}^N w_{ab}.
\]
This operator encodes how each vertex interacts with its neighbors and appears in diverse contexts, including electrical networks and diffusion phenomena~\cite{DoyleSnell1984,Chung1997}.
A crucial property is that \(\det(\mathbf{L}^w)=0\), reflecting the presence of a global symmetry or conservation law.

Fermionic formulations make this connection between Laplacians, determinants, and combinatorial enumeration explicit~\cite{Caracciolo2007,Abdesselam2004}, yielding compact functional-integral expressions for spanning structures.

\begin{figure}
    \centering
    \includegraphics[width=0.4\linewidth]{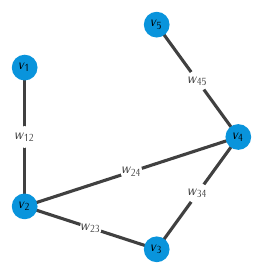}
    \caption{Weighted graph \(W\) with vertices \((v_1,v_2,v_3,v_4,v_5)\) and weighted edges \((w_{12},w_{23},w_{24},w_{34},w_{45})\).}
    \label{exmpG}
\end{figure}

\subsection{Rooted Spanning Forests}

A \textit{rooted spanning forest} is a disjoint union of trees, each with one distinguished root. 
These structures are enumerated by the principal minors of the Laplacian matrix (Theorem~\ref{teo:principalminorstree}). 
Physically, forests describe systems composed of multiple connected components anchored at fixed reference sites, appearing in models of diffusion, hierarchical organization, and in certain field-theoretic contexts where Laplacians and zero-mode regularization arise~\cite{ZinnJustin2002}.

The partition function of a rooted forest with root set \(I=\{i_1,\dots,i_r\}\) is given by
\begin{equation}
    Z_{SF}(I)=\int\mathbf{D}(\boldsymbol{\chi},\bar{\boldsymbol{\chi}})\,
    \Big(\prod_{\alpha=1}^{r}\bar{\chi}_{i_\alpha}\chi_{i_\alpha}\Big)
    e^{\bar{\boldsymbol{\chi}}^T \mathbf{L}^w\boldsymbol{\chi}}
    =\det\!\left(\mathbf{L}^w_{[I^c]}\right).
\end{equation}
When \(r=1\), this reduces to the case of a spanning tree. 

\begin{theorem}[Principal-Minors Tree Theorem~\cite{ChebotarevAgaev2002}]\label{teo:principalminorstree}
For $I=J$, we write $\mathbf{L}^w_{[I^c]}$. Then
\begin{equation}
     Z_{SF}(I)=\det\left(\mathbf{L}^w_{[I^c]}\right)
     \equiv \sum_{F\in \mathcal{F}}\prod_{e\in F}w_e,
\end{equation}
where the sum runs over all spanning forests $F$ composed of $r$ disjoint trees, each containing exactly one root vertex from $I$.
\end{theorem}

In combinatorial literature, the total number (or weighted count) of spanning trees or forests is usually denoted by $\tau(G)$. Throughout this work, we have used the symbol $Z$ to emphasize its interpretation as a Grassmann or field-theoretic partition function. Both notations are completely equivalent, and the choice between them is only conventional.

\subsection{Spectral Decomposition Method for Spanning Forest Enumeration}
\label{subsec:spectral_forests}

The enumeration of spanning forests—disjoint unions of trees that together cover all vertices—requires extra care for disconnected graphs.  While the Matrix–Tree Theorem treats connected graphs via cofactors of the Laplacian \cite{Kirchhoff1847}, its extension to general (possibly disconnected) graphs is most naturally formulated in spectral or block form: by detecting disconnected components (either combinatorially or via the spectral projector onto the zero-eigenspace of the Laplacian), one permutes/blocks the Laplacian into independent submatrices and factorizes the forest enumeration into contributions from each component.  This spectral/block approach is closely related to the matrix–forest theorems and their all-minors formulations, which give explicit formulas for forest weights in terms of Laplacian minors and resolvents~\cite{Chebotarev2006,Chaiken1982}.  From a numerical and algebraic viewpoint the same idea is used when one isolates zero modes by block diagonalization or projection: standard matrix-decomposition techniques identify the null and non-null subspaces and allow one to treat zero modes through projection rather than by forming formal inverses~\cite{Chung1997}. Building on our unified spectral framework, we now automate the detection of graph components, which allows us to block-decompose the Laplacian and factor the spanning forest count. This technique parallels the earlier treatment of Kasteleyn matrices and offers a robust computational pathway for forest generating functions, thereby completing the picture.

\subsubsection*{Component Identification via Spectral Analysis.}

The Laplacian matrix $\mathbf{L}$ encodes the connectivity structure of a graph through its spectrum. For a graph with $k$ connected components $G_1, G_2, \ldots, G_k$, the Laplacian has exactly $k$ zero eigenvalues, one per component, and the corresponding eigenvectors act as \textit{component indicators}. 
Let
\[
\mathbf{L} = \mathbf{Q}\Lambda\mathbf{Q}^\top,
\qquad
\Lambda= \mathrm{diag}(\lambda_1, \ldots, \lambda_N),
\]
be its spectral decomposition, where $\mathbf{Q}$ is orthogonal and $\Lambda$ contains the eigenvalues of $\mathbf{L}$.
The number of connected components is then
\begin{equation}
k = \dim(\ker(\mathbf{L})) = \#\{\lambda_i = 0\}.
\end{equation}

We denote by $\mathbf{Q}_0$ the $N\times k$ submatrix of $\mathbf{Q}$ formed by the eigenvectors corresponding to the zero eigenvalues of $\mathbf{L}$, and by $\mathbf{Q}_+$ the submatrix of eigenvectors with nonzero eigenvalues:
\begin{equation}
\mathbf{Q} = \bigl[\mathbf{Q}_0 \;\; \mathbf{Q}_+\bigr],
\qquad
\Lambda = 
\begin{pmatrix}
0 & 0\\
0 & \Lambda_+
\end{pmatrix},
\end{equation}
so that $\mathbf{Q}_0$ spans $\ker(\mathbf{L})$ and $\mathbf{Q}_+$ spans its regular (non-null) subspace.

Each zero eigenvector $\mathbf{v}_\alpha$ is constant on the vertices of a connected component and zero elsewhere:
\begin{equation}
(\mathbf{v}_\alpha)_i = 
\begin{cases}
\displaystyle \frac{1}{\sqrt{|V(G_\alpha)|}} & \text{if } i \in V(G_\alpha), \\[6pt]
0 & \text{otherwise.}
\end{cases}
\end{equation}
Hence, by inspecting the rows of $\mathbf{Q}_0$, vertices that share identical entries belong to the same component. This provides an entirely algebraic way to identify connected components directly from $\mathbf{L}$, without reconstructing the graph topology.

\subsubsection*{Block Diagonalization and Factorization.}

Once the components are identified, one can reorder vertices so that all vertices of the same component appear consecutively. This permutation $\mathbf{P}$ transforms the Laplacian into an explicit block-diagonal form:
\begin{equation}
\label{eq:blockLaplacian}
\mathbf{L}_{\text{block}} = 
\mathbf{P}^\top \mathbf{L} \mathbf{P} =
\begin{pmatrix}
\mathbf{L}_{G_1} & 0 & \cdots & 0 \\
0 & \mathbf{L}_{G_2} & \cdots & 0 \\
\vdots & \vdots & \ddots & \vdots \\
0 & 0 & \cdots & \mathbf{L}_{G_k}
\end{pmatrix},
\end{equation}
where each block $\mathbf{L}_{G_\alpha}$ is the Laplacian of component $G_\alpha$. 

This structure reveals that a disconnected graph behaves as the disjoint union of its connected subgraphs. From this, the enumeration of spanning forests with exactly one tree per component factorizes naturally:
\begin{equation}
\label{eq:forest_factorization}
\tau_{\text{forest}}(G)
= \prod_{\alpha=1}^{k} \tau(G_\alpha),
\end{equation}
where each $\tau(G_\alpha)$ is the number of spanning trees of component $G_\alpha$. 

Equation~\eqref{eq:forest_factorization} is the algebraic manifestation of the generalized Matrix--Tree theorem: taking one vertex from each component (removing $k$ rows and columns) annihilates all zero modes and yields a finite determinant equal to the product of tree counts of the connected subgraphs.

\subsubsection{Spectral Projection and Regularization}

The spectral decomposition provides an algebraic representation of the block-diagonal structure through projection operators:
\[
\mathbf{L} = \mathbf{Q}_0\mathbf{0}\mathbf{Q}_0^\top + \mathbf{Q}_+\Lambda_+\mathbf{Q}_+^\top.
\]
The projection operator
\[
\mathbf{P} = \mathbf{Q}_+\mathbf{Q}_+^\top = \mathbf{I} - \mathbf{Q}_0\mathbf{Q}_0^\top
\]
filters out the zero modes associated with disconnected components. The regularized Laplacian
\[
\mathbf{L}_{\text{reg}} = \mathbf{P}\mathbf{L}\mathbf{P} = \mathbf{Q}_+\Lambda_+\mathbf{Q}_+^\top
\]
contains only the invertible sector, analogous to the reduced antisymmetric matrix in Kasteleyn theory. This spectral approach provides an operator-theoretic treatment of connected and disconnected graphs within a unified framework.

\subsubsection*{Algorithmic Implementation.}

\begin{enumerate}
\item \textbf{Diagonalize the Laplacian:} compute $\mathbf{L} = \mathbf{Q}\Lambda\mathbf{Q}^\top$.
\item \textbf{Identify zero modes:} select eigenvectors in $\mathbf{Q}_0$ with $|\lambda_i| < \epsilon$.
\item \textbf{Cluster rows of $\mathbf{Q}_0$:} identical rows correspond to vertices in the same component.
\item \textbf{Extract sub-Laplacians:} for each component $G_\alpha$ with vertex set $V_\alpha$, set $\mathbf{L}_\alpha = \mathbf{L}[V_\alpha,V_\alpha]$.
\item \textbf{Compute component tree numbers:} delete one row and column from each $\mathbf{L}_\alpha$, compute $\tau(G_\alpha)=\det\left(\mathbf{L}_{\alpha[\{i\}^c]}\right)$.
\item \textbf{Multiply:} obtain $\tau_{\text{forest}}(G)$ through Eq.~\eqref{eq:forest_factorization}.
\end{enumerate}

This procedure can be implemented entirely from $\mathbf{L}$ or its spectrum, requiring no prior knowledge of the graph’s topology. An illustration of this procedure, including the explicit construction of $Q_0$, the identification of connected components, and the factorization of the Laplacian, is presented in Appendix~\ref{sec:SpectralIllustration}.

\subsubsection*{Physical and Combinatorial Interpretation.}

The factorization \eqref{eq:forest_factorization} has several complementary interpretations:

\begin{itemize}
\item \textbf{Combinatorial:} Each connected component $G_\alpha$ contributes a factor equal to its own number of spanning trees $\tau(G_\alpha)$. The product $\prod_{\alpha=1}^k \tau(G_\alpha)$ therefore enumerates all possible combinations of one spanning tree per component, i.e., all spanning forests that collectively cover the entire vertex set of $G$.
\item \textbf{Spectral:} Each zero mode of $\mathbf{L}$ corresponds to a disconnected region; projecting onto the non-null subspace eliminates global shifts.
\item \textbf{Physical:} In the fermionic formalism, each component behaves as an independent sector whose Grassmann integral factorizes:
\[
e^{\bar{\boldsymbol{\chi}}^\top \mathbf{L} \boldsymbol{\chi}}
= \prod_{\alpha=1}^{k} 
e^{\bar{\boldsymbol{\chi}}_\alpha^\top \mathbf{L}_{G_\alpha} \boldsymbol{\chi}_\alpha},
\]
so that
\[
Z_{\mathrm{SF}}
= \prod_{\alpha=1}^{k}
\int \mathbf{D}(\boldsymbol{\chi}_\alpha,\bar{\boldsymbol{\chi}}_\alpha)\,
\bar{\chi}_{i_\alpha}\chi_{i_\alpha}\,
e^{\bar{\boldsymbol{\chi}}_\alpha^\top \mathbf{L}_{G_\alpha}\boldsymbol{\chi}_\alpha}.
\]
Each block integral counts the spanning trees of $G_\alpha$, and their product yields the total spanning-forest weight.
\end{itemize}

The spectral decomposition plays for the Laplacian the same regularizing role that the block-Pfaffian construction plays for the Kasteleyn matrix. Both separate the null and regular sectors, allowing well-defined evaluation of Determinants or Pfaffians even when the original matrices are non-invertible. This unified viewpoint provides a consistent algebraic and physical framework for treating both fermionic and combinatorial systems with disconnected structures.

\subsection{Spanning Trees}
A \textit{spanning tree} is a connected acyclic subgraph that contains every vertex of the original graph. The celebrated \textit{Matrix--Tree Theorem} (Theorem~\ref{teo:matrixtree}) states that the number (or weighted sum) of spanning trees equals any cofactor of the graph Laplacian $\mathbf{L}^w$; see the original formulation of Kirchhoff and standard algebraic treatments~\cite{Kirchhoff1847,Biggs1993}, as well as a combinatorial all-minors proof~\cite{Chaiken1982}. Physically, spanning trees represent the minimal connectivity backbone of a system and play a central role in network theory, random walks, and electrical-circuit models~\cite{DoyleSnell1984,LyonsPeres2016}.

The spanning-tree partition function can be expressed as
\begin{equation}\label{eq:SpanningTreefermionic}
    Z_{ST}(\mathbf{L}^w)=\int\mathbf{D}(\boldsymbol{\chi},\bar{\boldsymbol{\chi}})\,
    \bar{{\chi}} _{a}{\chi}_{a}\,
    e^{\bar{\boldsymbol{\chi}}^T \mathbf{L}^w\boldsymbol{\chi}}
    =\det\!\left(\mathbf{L}^w_{[\{a\}^c]}\right),
\end{equation}
which counts all spanning trees of \(G\), each weighted by the product of its edge weights.

\begin{theorem}[Matrix--Tree Theorem, connected graphs \cite{Kirchhoff1847}]\label{teo:matrixtree}
Let $G$ be a connected graph with weighted Laplacian $\mathbf{L}^w$. Then any cofactor of $\mathbf{L}^w$ generates spanning trees:
\begin{equation}
    Z_{ST}\mathbf{L}^w=(-1)^{a+b}\det\!\left(\mathbf{L}^w_{[\{b\}^c|\{a\}^c]}\right)
    \equiv \sum_{T\in\mathcal{T}} \prod_{e\in T} w_e,
\end{equation}
where $\mathcal{T}$ is the set of spanning trees of $G$. For $a=b$, this reduces to a principal minor.
\end{theorem}

The correlation function defined by
\begin{equation}
\begin{split}
&\left\langle\prod_{\alpha=1}^{r}\bar{\chi}_{i_{\alpha}}\chi_{j_{\alpha}}\right\rangle_{\mathrm{ST}}\\
&\qquad:=\frac{1}{Z_{ST}(\mathbf{L}^w)}\int\mathbf{D}(\boldsymbol{\chi},\bar{\boldsymbol{\chi}})\left(\prod_{\alpha=1}^{r}\bar{\chi}_{i_{\alpha}}\chi_{j_{\alpha}}\right)\bar{{\chi}} _{a}{\chi}_{a}
    e^{\bar{\boldsymbol{\chi}}^T \mathbf{L}^w\boldsymbol{\chi}}=\epsilon\left(I_{[r]},J_{[r]}\right)\frac{\det\!\left(\mathbf{L}^w_{\left[(I_{[r]}\cup\{a\})^c|(J_{[r]}\cup\{a\})^c\right]}\right)}{\det\!\left(\mathbf{L}^w_{[\{a\}^c]}\right)}
\end{split}
\end{equation}
carries fundamental physical significance in the statistical mechanics of spanning trees.

Each Grassmann bilinear $\bar{\chi}_{i_\alpha}\chi_{j_\alpha}$ represents the presence of a specific edge $(i_\alpha,j_\alpha)$ in the spanning tree configuration. The product over all $\alpha$ from 1 to $r$ therefore corresponds to requiring that all edges in the set $E' = \{(i_1,j_1), \dots, (i_r,j_r)\}$ appear simultaneously in the tree. The expectation value computes the joint probability of this occurrence in the ensemble of uniformly random spanning trees.

From a network theory perspective, this correlation function measures connectivity patterns in minimal connected networks. In electrical network contexts, it represents the probability that all specified edges carry current in a random resistor network configuration. For communication networks, it quantifies the likelihood that multiple communication links are simultaneously active in a minimally connected infrastructure.

The regularization term $\bar{\chi}_a\chi_a$ plays a crucial physical role by removing the zero mode of the weighted Laplacian $\mathbf{L}^w$, which arises from its singular nature. This regularization corresponds to fixing a root vertex $a$ in the spanning tree enumeration, effectively breaking the global symmetry and making the partition function and correlations well-defined.
\subsection{Berezin integral over Grassmann variables representation of the spanning tree}
\begin{theorem}[All-minors Tree theorem~\cite{Chaiken1982}]\label{theoAllminor}
Let $G$ be a connected graph with weighted Laplacian $\mathbf{L}^w$. For any subsets $I,J$ of equal cardinality $r$,
\begin{equation}
Z_{SF}(I,J)=\epsilon\left(I_{[r]},J_{[r]}\right)\det\left(\mathbf{L}_{\left[I_{[r]}^c,J_{[r]}^c\right]}\right)
=\sum_{F\in \mathcal{F}}\prod_{e\in F}w_e,
\end{equation}
where the sum runs over spanning forests $F$ composed of $r$ disjoint trees, each containing exactly one vertex from $I$ and one from $J$; $\epsilon\left(I_{[r]},J_{[r]}\right)=(-1)^{\sum_{i \in I} i+\sum_{j\in J} j}$.
\end{theorem}

Beyond the classical algebraic formulations, we introduce a compact functional-integral representation for the spanning-tree generating function (see Corollary~\ref{Corollary2.6} for a broader generalization):
\begin{theorem}[Fermionic Sanning tree with bosonic sources]\label{TheoST}
The partition function for the spanning tree is
\begin{equation}\label{eq:NewST}
    Z_{ST}=\frac{1}{\sum_{i,j=1}^L \bar{u}_j u_i}
    \int \mathbf{D}(\boldsymbol{\chi},\bar{\boldsymbol{\chi}})\,
    e^{\bar{\boldsymbol{\chi}}^\top \mathbf{L} \boldsymbol{\chi}}\,
    e^{\bar{\boldsymbol{\chi}}^\top \mathbf{u}}\,
    e^{\bar{\mathbf{u}}^\top \boldsymbol{\chi}}.
\end{equation}

The correlation functions using bosonic source fields:

\begin{equation}
\begin{split}
\left\langle\prod_{\alpha=1}^{r}\bar{\chi}_{i_{\alpha}}\chi_{j_{\alpha}}\right\rangle_{\mathrm{ST}}
:=\frac{1}{Z_{ST}}\frac{1}{\sum_{i,j=1}^L \bar{u}_j u_i}
    \int\mathbf{D}(\boldsymbol{\chi},\bar{\boldsymbol{\chi}})\,
    \left(\prod_{\alpha=1}^{r}\bar{\chi}_{i_{\alpha}}\chi_{j_{\alpha}}\right)
    e^{\bar{\boldsymbol{\chi}}^\top \mathbf{L} \boldsymbol{\chi}}\,
    e^{\bar{\boldsymbol{\chi}}^\top \mathbf{u}}\,
    e^{\bar{\mathbf{u}}^\top \boldsymbol{\chi}},
\end{split}
\end{equation}

The Berezin integral expansion reveals the underlying combinatorial structure:

\begin{equation}
\begin{split}
    \int\mathbf{D}(\boldsymbol{\chi},\bar{\boldsymbol{\chi}})&
    \left(\prod_{\alpha=1}^{r}\bar{\chi}_{i_{\alpha}}\chi_{j_{\alpha}}\right)
    e^{\bar{\boldsymbol{\chi}}^\top \mathbf{L} \boldsymbol{\chi}}
    e^{\bar{\boldsymbol{\chi}}^\top \mathbf{u}}\,
    e^{\bar{\mathbf{u}}^\top \boldsymbol{\chi}}\\
    &=\epsilon\left(I_{[r]},J_{[r]}\right)\det\left(\mathbf{L}_{\left[I_{[r]}^c,J_{[r]}^c\right]}\right)
    +\sum_{j_0,i_0=1}^L\epsilon\left(I_{[0,r]},J_{[0,r]}\right)
    \bar{u}_{j_0}u_{i_0}\det\left(\mathbf{L}_{\left[I_{[0,r]}^c,J_{[0,r]}^c\right]}\right).
\end{split}
\end{equation}

The correlation function can be written in terms of the spanning forests partition functions by using Theorem~\ref{theoAllminor}:
\begin{equation}
\left\langle\prod_{\alpha=1}^{r}\bar{\chi}_{i_{\alpha}}\chi_{j_{\alpha}}\right\rangle_{\mathrm{ST}}
=\frac{1}{Z_{ST}}\frac{Z_{SF}(I_{[r]},J_{[r]})+\sum_{j_0,i_0=1}^L\bar{u}_{j_0}u_{i_0}Z_{SF}(I_{[0,r]},J_{[0,r]})}{\sum_{i,j=1}^L \bar{u}_j u_i}.
\end{equation}

The proof follows directly from expanding the integrand and identifying the resulting minor structures. A more general treatment using complex fermions is provided in Section~\ref{Sec:5}, where the complete proof of this theorem can be found.
\end{theorem}

This \textit{source-ordered} formulation incorporates root insertions directly into the functional measure via an auxiliary bosonic vector $\mathbf{u}$, thereby avoiding the explicit deletion of rows or columns from the Laplacian required in standard proofs of the Matrix–Tree Theorem~\cite{Kirchhoff1847,Chaiken1982}. The bosonic sources $(\mathbf{u},\bar{\mathbf{u}})$ implement root insertions while preserving the full operator character of $\mathbf{L}$; in practice they regularize the Laplacian's zero mode and act as a convenient device for analytic continuation and perturbative expansions. Crucially, the spanning-tree measure produced by the source-augmented Berezin integral is independent of the particular choice of source vectors (the insertion is auxiliary or ``gauge-like''), so physical observables do not depend on this choice.

From a computational and field-theoretic viewpoint, eq. \eqref{eq:NewST} casts rooted spanning-tree enumeration into a Gaussian Grassmann framework, making standard diagrammatic and functional techniques available for analytic approximations and symbolic or numerical manipulations~\cite{ZinnJustin2002,negele1988quantum}. The representation therefore complements algebraic approaches to trees and forests (matrix–tree and matrix–forest theorems) and suggests alternative, operator-level methods for both exact and approximate computations~\cite{Chaiken1982,Chebotarev2006}.

This formulation enables correlation computation through Laplacian minors but reveals a fundamental physical distinction: while the partition function $Z_{ST}$ reproduces the classical spanning tree count, the correlation functions exhibit qualitatively different behavior from conventional spanning tree correlations. The bosonic sources $\mathbf{u}$ and $\bar{\mathbf{u}}$ transform the correlation structure, introducing explicit field dependence—a radical departure from standard edge inclusion probabilities in uniform spanning trees.

The bosonic fields serve dual purposes: they regularize the Laplacian while introducing gauge-like freedom, but simultaneously alter the physical interpretation of correlations. This suggests the bosonic formulation describes a different physical system—spanning forests rather than pure spanning trees—despite partition function equivalence.

Spanning forests and trees thus provide a natural bridge between combinatorics and field theory~\cite{Caracciolo2004}. Fermionic integrals capture the determinant structure of connectivity, while bosonic sources extend this framework to functional-integral representations with modified correlation structures. This combinatorial foundation motivates the generalized Berezin integral identities in Section~\ref{Sec:5}, extending classical Gaussian formulas to mixed bosonic-fermionic systems.

\section{Monobisyzexant and Hafnian}\label{Sec:5}
This section introduces the \textit{Monobisyzexant} function and its special cases—the \textit{Hafnian}, the \textit{Monobisyzexantinho}, and the \textit{Hafnianinho}. It also presents two theorems formulated in terms of complex fermions. We begin by defining the Hafnian and the Monobisyzexant function, then revisiting the fermionic representations of the Hafnian and its reduced form, the \textit{hafnianinho}, now extended to include fermionic source terms. The fermionic formulations presented in this section were originally developed in~\cite{Samuel1980}. We generalize these identities by incorporating explicit source insertions within the Berezin integral framework.
\subsection{The Hafnian}
\begin{definition}[Hafnian]\label{Def:Hafnian}
The Hafnian of a $N\times N$ symmetric matrix $\mathbf{V}=(V_{ij})_{i,j=1}^{N}$, with even $N$, is 
\begin{equation}
    \haf(\mathbf{V}):=\frac{1}{2^{N/2}\left(\frac{N}{2}\right)!}\sum_{i_1,\dots,i_{N}=1}^{N}\epsilon^{i_1\dots i_{N}}\epsilon^{i_1\dots i_{L}}V_{i_1i_2}\cdots V_{i_{N-1}i_{N}},
\end{equation}
where $\varepsilon^{i_1\cdots i_{N}}$ is the Levi-Civita symbol.
\end{definition}

For example, the Hafnian of a $2\times 2$ matrix is
\begin{equation*}
    \haf\begin{pmatrix}
        V_{11}&V_{12}\\
        V_{12}&V_{22}
    \end{pmatrix}=V_{12}.
\end{equation*}
For a $4\times 4$ matrix we have
\begin{equation*}
\haf\begin{pmatrix}
    V_{11}&V_{12}&V_{13}&V_{14}\\
    V_{12}&V_{22}&V_{23}&V_{24}\\
    V_{13}&V_{23}&V_{33}&V_{34}\\
    V_{14}&V_{24}&V_{34}&V_{44}
\end{pmatrix}=V_{12}V_{34}+V_{13}V_{24}+V_{14}V_{23}.
\end{equation*}

\begin{lemma}[Matrix Scaling]
Let $\mathbf{V}=(V_{ij})_{i,j=1}^N$ be an $N\times N$ symmetric matrix with zero diagonal. We say that $\mathbf{V}$ is obtained by scaling from another $N \times N$ symmetric matrix $\mathbf{S}=(S_{ij})_{i,j=1}^N$ if
\begin{equation}
    V_{ij}=g_i g_jS_{ij},
\end{equation}
for some constants $g_1, \dots, g_N$. If $N$ is even, then the Hafnians of $\mathbf{V}$ and $\mathbf{S}$ are well-defined and satisfy
\begin{equation}
    \haf(\mathbf{V})=\left(\prod_{i=1}^N g_i\right)\haf(\mathbf{S}).
\end{equation}

In the special case where $V_{ij}=gS_{ij}$ for a constant $g \in \mathbb{R}$, it follows that
\begin{equation}
    \haf(\mathbf{V})=\haf(g\mathbf{S})=g^{\frac{N}{2}}\haf(\mathbf{S}).
\end{equation}
\end{lemma}
\begin{proof}
The result follows directly from the definition of the Hafnian. For any perfect matching $M = \{(i_1,j_1),\dots,(i_{N/2},j_{N/2})\}$ in the complete graph on $N$ vertices, we have:
\[
\prod_{(i,j) \in M} V_{ij} = \prod_{(i,j) \in M} (g_i g_j S_{ij}) = \left(\prod_{i=1}^N g_i\right) \prod_{(i,j) \in M} S_{ij},
\]
where the equality holds because each vertex appears exactly once in any perfect matching. Summing over all perfect matchings yields the scaling relation. The constant scaling case follows by taking $g_i = \sqrt{g}$ for all $i$.
\end{proof}

We now define the Hafnian of a submatrix, which we refer to as the \textit{hafnianinho}. 
\begin{definition}[The hafnianinho]\label{def:hafnianinho}
The hafnianinho is the Hafnian of a $N\times N$ symmetric matrix $\mathbf{V}$ without $I$ rows and columns, where $I=\{i_1,\dots,i_{r}\}\subseteq[N]$ with $i_1<\dots<i_{r}$.
\begin{equation}
    \haf\left(\mathbf{V}_{[I^c]}\right):=\frac{1}{2^{(N-r)/2}\left(\frac{N-r}{2}\right)!}\sum_{j_1,\dots,j_{(N-r)}=1}^{N}\epsilon^{i_1\cdots i_{r}j_1\cdots j_{(N-r)}}\epsilon^{i_1\cdots i_{r}j_1\cdots j_{(N-r)}}V_{j_1j_2}\cdots V_{j_{(N-r-1)}j_{(N-r)}},
\end{equation}
where $N$ and $r$ are even numbers, and $2\leq r< N$.

In the case $r=N$ we have:
\begin{equation}
\haf\left(\mathbf{V}_{[\emptyset]}\right)=1.
\end{equation}
\end{definition}

\subsection{The Monobisyzexant}

The Monobisyzexant is a comprehensive matrix function defined for any symmetric matrix $\mathbf{V}$ and diagonal matrix $\mathbf{D}$, encoding monomer-dimer configurations through its diagonal and off-diagonal entries. While related to the loop-Hafnian for even-dimensional graph Laplacians, the Mbsz extends this concept to arbitrary symmetric matrices and reduces to the standard Hafnian when all diagonal entries vanish.

\begin{definition}[Monobisyzexant]\label{def:Mbsz}
For an $L\times L$ symmetric matrix $\mathbf{V}$ and diagonal matrix $\mathbf{D}$, the Monobisyzexant is defined as:
\begin{equation}
\operatorname{Mbsz}\left(\mathbf{D},\mathbf{V}\right):=\frac{1}{L!}\sum_{i_1,\dots,i_{L}=1}^{L}\epsilon^{i_1\dots i_{L}}\epsilon^{i_1\dots i_{L}}\mathcal{S}_{i_1 \cdots i_{L}}\left(\mathbf{D},\mathbf{V}\right),
\end{equation}
where $\epsilon^{a_1 \cdots a_{L}}$ denotes the Levi-Civita symbol and the tensor $\mathcal{S}_{i_1 \cdots i_L}(\mathbf{D},\mathbf{V})$ is given by the determinant:
\begin{equation}\label{equ_Stensor}
\mathcal{S}_{i_1 \cdots i_L}(\mathbf{D},\mathbf{V}) =
\begin{vmatrix}
D_{i_1 i_1} & \binom{L-1}{1} V_{i_1 i_2} & 0 & \cdots & 0 \\
-1 & D_{i_2 i_2} & \binom{L-2}{1} V_{i_2 i_3} & \ddots & \vdots \\
0 & -1 & D_{i_3 i_3} & \ddots & 0 \\
\vdots & \ddots & \ddots & \ddots & \binom{1}{1} V_{i_{L-1} i_L} \\
0 & \cdots & 0 & -1 & D_{i_L i_L}
\end{vmatrix}.
\end{equation}
Setting $\mathbf{D} = \mathbf{0}$ recovers the standard Hafnian.
\end{definition}

\begin{lemma}[Mbsz Expansion]\label{Mbszexpanded}
The Monobisyzexant admits an expansion in terms of principal minors of $\mathbf{D}$ and Hafnians of $\mathbf{V}$:

\noindent
(a) For odd $L$:
\begin{equation}
\operatorname{Mbsz}\!\left(\mathbf{D},\mathbf{V}\right)
=\det(\mathbf{D})
+\sum_{\alpha=1}^{(L-1)/2}
\sum_{i_1<\cdots<i_{[2\alpha-1]}=1}^L
\det\!\left(\mathbf{D}_{[I_{[2\alpha-1]}]}\right)
\haf\!\left(\mathbf{V}_{[I^c_{[2\alpha-1]}]}\right),
\end{equation}
where $I=\{i_1,\dots,i_{[2\alpha-1]}\}$.

\noindent
(b) For even $L$:
\begin{equation}
\operatorname{Mbsz}\!\left(\mathbf{D},\mathbf{V}\right)
=\det(\mathbf{D})
+\sum_{\alpha=1}^{L/2-1}
\sum_{i_1<\cdots<i_{[2\alpha]}=1}^L
\det\!\left(\mathbf{D}_{[I_{[2\alpha]}]}\right)
\haf\!\left(\mathbf{V}_{[I^c_{[2\alpha]}]}\right)
+\haf(\mathbf{V}).
\end{equation}

The result follows by induction on $L$. The base cases for small $L$ are verified by direct computation, and the inductive step establishes the general pattern.
\end{lemma}
Explicit examples for $L=4$ and $L=5$ are provided in Appendix~\ref{Mbszexample}.

\begin{definition}[Monobisyzexantinho]\label{def:Mbszinho}
For $1 \leq r < L$ and index set $I=\{i_1,\dots,i_{r}\}$ with $i_1<\dots<i_{r}$, the Monobisyzexantinho is defined as:
\begin{equation}
\operatorname{Mbsz}\left(\left(\mathbf{D},\mathbf{V}\right)_{\left[I^c\right]}\right):=\frac{1}{(L-r)!}\sum_{j_1,\dots,j_{L-r}=1}^L\epsilon^{i_1\cdots i_{r}j_1\cdots j_{L-r}}\epsilon^{i_1\cdots i_{r}j_1\cdots j_{L-r}}\mathcal{S}_{j_1\cdots j_{L-r}}\left(\mathbf{D},\mathbf{V}\right).
\end{equation}
For $r=L$, we define $\operatorname{Mbsz}\left(\left(\mathbf{D},\mathbf{V}\right)_{[\emptyset]}\right):=1$.
\end{definition}

\subsection{Fermionic Identities for Monobisyzexant and Hafnian}

\begin{theorem}[Fermionic Monobisyzexant Identities]\label{Theo:Mbsz}
For fermionic fields $\boldsymbol{\chi},\bar{\boldsymbol{\chi}}$, sources $\boldsymbol{\psi},\bar{\boldsymbol{\psi}}$, diagonal matrix $\mathbf{D}$, and symmetric matrix $\mathbf{V}$:

\noindent
(a) If $\operatorname{Mbsz}\left(\mathbf{D},\mathbf{V}\right)\neq0$, then:
\begin{equation}
\boxed{\int\mathbf{D}(\boldsymbol{\chi},\bar{\boldsymbol{\chi}})\, 
e^{\bar{\boldsymbol{\chi}}^\top\mathbf{D}\boldsymbol{\chi}
+\frac{1}{2} (\bar{\boldsymbol{\chi}}\boldsymbol{\chi})^\top\mathbf{V} (\bar{\boldsymbol{\chi}}\boldsymbol{\chi})
+\bar{\boldsymbol{\psi}}^\top\boldsymbol{\chi}
+\bar{\boldsymbol{\chi}}^\top\boldsymbol{\psi}}
=\operatorname{Mbsz}(\mathbf{D},\mathbf{V})\, 
e^{\bar{\boldsymbol{\psi}}^\top\mathbf{Y}\boldsymbol{\psi}
+\frac{1}{2}(\bar{\boldsymbol{\psi}}\boldsymbol{\psi})^\top\mathbf{W} (\bar{\boldsymbol{\psi}}\boldsymbol{\psi})},}
\end{equation}
where $\mathbf{Y}$ and $\mathbf{W}$ satisfy:
\begin{align}
    \mathcal{S}_{a_1\cdots a_{n}}(\mathbf{D},\mathbf{V})&=\frac{(-1)^n}{(2n)!}\frac{\operatorname{Mbsz}\left(\left(\mathbf{D},\mathbf{V}\right)_{[I^c]}\right)}{\operatorname{Mbsz}\left(\mathbf{D},\mathbf{V}\right)},\quad 1\leq n<N,\\
    \operatorname{Mbsz}\left(\mathbf{Y},\mathbf{W}\right)&=\frac{(-1)^N}{\operatorname{Mbsz}\left(\mathbf{D},\mathbf{V}\right)}.
\end{align}

\noindent
(b) For any index set $I=\{b_1,\dots,b_n\}$ with $b_1<\dots<b_n$:
\begin{equation}
\boxed{\int\mathbf{D}(\boldsymbol{\chi},\bar{\boldsymbol{\chi}})\left(\prod_{i=1}^{n}\bar{\chi}^{b_i}\chi^{b_i}\right)\, 
e^{\bar{\boldsymbol{\chi}}^\top\mathbf{D}\boldsymbol{\chi}
+\frac{1}{2} (\bar{\boldsymbol{\chi}}\boldsymbol{\chi})^\top\mathbf{V} (\bar{\boldsymbol{\chi}}\boldsymbol{\chi})}
=\operatorname{Mbsz}\left(\left(\mathbf{D},\mathbf{V}\right)_{[I^c]}\right).}
\end{equation}

\noindent
(c) If $\operatorname{Mbsz}\left(\mathbf{D},\mathbf{V}\right)\neq0$, then for any index set $I=\{b_1,\dots,b_n\}$:
\begin{equation}
\boxed{\int\mathbf{D}(\boldsymbol{\chi},\bar{\boldsymbol{\chi}})\left(\prod_{i=1}^{n}\bar{\chi}^{b_i}\chi^{b_i}\right)\, 
e^{\bar{\boldsymbol{\chi}}^\top\mathbf{D}\boldsymbol{\chi}
+\frac{1}{2} (\bar{\boldsymbol{\chi}}\boldsymbol{\chi})^\top\mathbf{V} (\bar{\boldsymbol{\chi}}\boldsymbol{\chi})}
=(-1)^nn!\operatorname{Mbsz}\left(\mathbf{D},\mathbf{V}\right)\mathcal{S}_{b_1\cdots b_{n}}\left(\mathbf{Y},\mathbf{W}\right).}
\end{equation}

The proof proceeds by direct induction on the integral. For the base case, we use small values of $L$ and expand explicitly, revealing a pattern that generalizes to arbitrary values of $L$.
\end{theorem}

\begin{corollary}[Fermionic Hafnian Identities]\label{Coro:Fermionichafnian}
For fermionic fields $\boldsymbol{\chi},\bar{\boldsymbol{\chi}}$ and sources $\boldsymbol{\psi},\bar{\boldsymbol{\psi}}$:

\noindent
(a) If $\haf(\mathbf{V})\neq0$, then:
\begin{equation}
\int \mathbf{D}(\boldsymbol{\chi},\bar{\boldsymbol{\chi}})\,
e^{\frac{1}{2} (\bar{\boldsymbol{\chi}}\boldsymbol{\chi})^\top\mathbf{V} (\bar{\boldsymbol{\chi}}\boldsymbol{\chi})
+ \bar{\boldsymbol{\psi}}^\top\boldsymbol{\chi}
+ \bar{\boldsymbol{\chi}}^\top\boldsymbol{\psi}}
=\haf(\mathbf{V})\,
e^{\frac{1}{2}(\bar{\boldsymbol{\psi}}\boldsymbol{\psi})^\top\mathbf{W} (\bar{\boldsymbol{\psi}}\boldsymbol{\psi})},
\end{equation}
where $\mathbf{W}$ satisfies:
\begin{align}
    (2n-1)!!W_{i_1i_2}\cdots W_{i_{2n-1}i_{2n}}&=\frac{1}{(2n)!}\frac{\haf\left(\mathbf{V}_{[I^c]}\right)}{\haf(\mathbf{V})},\quad 1\leq n<N,\\
    \haf(\mathbf{W})&=\frac{1}{\haf(\mathbf{V})}.
\end{align}

\noindent
(b) For any index set $I=\{i_1,\dots,i_{2n}\}$ with $i_1<\dots<i_{2n}$:
\begin{equation}
\int \mathbf{D}(\boldsymbol{\chi},\bar{\boldsymbol{\chi}})\left(\prod_{\alpha=1}^{2n}\bar{\chi}_{i_\alpha}\chi_{i_\alpha}\right)\,
e^{\frac{1}{2} (\bar{\boldsymbol{\chi}}\boldsymbol{\chi})^\top\mathbf{V} (\bar{\boldsymbol{\chi}}\boldsymbol{\chi})}
=\haf\left(\mathbf{V}_{[I^c]}\right).
\end{equation}

\noindent
(c) If $\haf(\mathbf{V})\neq0$, then for any index set $I=\{i_1,\dots,i_{2n}\}$:
\begin{equation}
\int \mathbf{D}(\boldsymbol{\chi},\bar{\boldsymbol{\chi}})\left(\prod_{\alpha=1}^{2n}\bar{\chi}_{i_\alpha}\chi_{i_\alpha}\right)\,
e^{\frac{1}{2} (\bar{\boldsymbol{\chi}}\boldsymbol{\chi})^\top\mathbf{V} (\bar{\boldsymbol{\chi}}\boldsymbol{\chi})}
=(2n)!(2n-1)!!W_{i_1i_2}\cdots W_{i_{2n-1}i_{2n}}\haf(\mathbf{V}).
\end{equation}

This is a direct corollary of Theorem~\ref{Theo:Mbsz}. Taking $\mathbf{D} = \mathbf{0}$ in the general Monobisyzexant construction recovers the Hafnian as a special case.
\end{corollary}
\section{Generalized Berezin–Pfaffian and Determinant Identities}\label{Sec:6}
This section introduces a class of new Grassmann integrals involving both real and complex fermionic variables. Some particular cases were previously studied in~\cite{caracciolo2013algebraic}, where integrals of exponentials of quadratic forms with additional linear Grassmann terms were evaluated.  
The integrals presented here extend those results, providing a broader framework and a wider family of identities applicable in various contexts.

Before presenting the new integrals, we summarize the known results from~\cite{caracciolo2013algebraic}.

\subsection*{Real Fermions}
Let $\mathbf{A}$ be a $2L\times 2L$ antisymmetric matrix. Then, for real Grassmann variables $\boldsymbol{\chi}$, the following results hold:

\begin{enumerate}
    \item If $\mathbf{A}$ is invertible,
    \begin{equation*}
        \int \mathbf{D}\boldsymbol{\chi} \;
        e^{\tfrac{1}{2}\boldsymbol{\chi}^{T}\mathbf{A}\boldsymbol{\chi}
        + \boldsymbol{\psi}^{T}\boldsymbol{\chi}}
        = \pf(\mathbf{A}) \,
        e^{\tfrac{1}{2}\boldsymbol{\psi}^{T}\mathbf{A}^{-1}\boldsymbol{\psi}}.
    \end{equation*}

    \item For any subset $I_{[r]} = \{i_{1},\dots,i_{r}\}\subseteq [2L]$ with $i_{1}<\cdots<i_{r}$,
    \begin{equation*}
        \int \mathbf{D}\boldsymbol{\chi}
        \left(\prod_{\alpha=1}^{r}\chi_{i_{\alpha}}\right)
        e^{\tfrac{1}{2}\boldsymbol{\chi}^{T}\mathbf{A}\boldsymbol{\chi}}
        =
        \begin{cases}
            0, & \text{if $r$ is odd},\\[4pt]
            \epsilon(I_{[r]})\pf\left(\mathbf{A}_{[I_{[r]}^{c}]}\right), & \text{if $r$ is even.}
        \end{cases}
    \end{equation*}

    \item For any sequence of indices $I=(i_{1},\dots,i_{r})$ in $[2L]$, if $\mathbf{A}$ is invertible,
    \begin{equation*}
        \int \mathbf{D}\boldsymbol{\chi}
        \left(\prod_{\alpha=1}^{r}\chi_{i_{\alpha}}\right)
        e^{\tfrac{1}{2}\boldsymbol{\chi}^{T}\mathbf{A}\boldsymbol{\chi}}
        =
        \begin{cases}
            0, & \text{if $r$ is odd},\\[4pt]
            \pf(\mathbf{A})\,
            \pf\!\left((\mathbf{A}^{-T})_{[I_{[r]}]}\right), & \text{if $r$ is even.}
        \end{cases}
    \end{equation*}

    \item For any $r\times 2L$ matrix $\mathbf{C}$,
    \begin{equation*}
        \int \mathbf{D}\boldsymbol{\chi}
        \left(\prod_{\alpha=1}^{r}(\mathbf{C}\boldsymbol{\chi})_{\alpha}\right)
        e^{\tfrac{1}{2}\boldsymbol{\chi}^{T}\mathbf{A}\boldsymbol{\chi}}
        =
        \begin{cases}
            0, & \text{if $r$ is odd},\\[4pt]
            \displaystyle\sum_{i_{1}<\cdots<i_{r}=1}^{L}
            \epsilon(I_{[r]})\,
            \det\!\big(\mathbf{C}_{[\star,I_{[r]}]}\big)\,
            \pf\!\left(\mathbf{A}_{[I_{[r]}^{c}]}\right), & \text{if $r$ is even.}
        \end{cases}
    \end{equation*}
\end{enumerate}

\subsection*{Complex Fermions}

Let now $\mathbf{A}$ be a $L\times L$ (not necessarily symmetric) matrix. Then, for complex Grassmann variables $(\boldsymbol{\chi},\bar{\boldsymbol{\chi}})$, the following hold:

\begin{enumerate}
    \item If $\mathbf{A}$ is invertible,
    \begin{equation*}
        \int \mathbf{D}(\boldsymbol{\chi},\bar{\boldsymbol{\chi}})\;
        e^{\bar{\boldsymbol{\chi}}^{T}\mathbf{A}\boldsymbol{\chi}
        + \bar{\boldsymbol{\psi}}^{T}\boldsymbol{\chi}
        + \bar{\boldsymbol{\chi}}^{T}\boldsymbol{\psi}}
        = \det(\mathbf{A})\,
        e^{-\bar{\boldsymbol{\psi}}^{T}\mathbf{A}^{-1}\boldsymbol{\psi}}.
    \end{equation*}

    \item For subsets $I_{[r]}=\{i_{1},\dots,i_{r}\}$ and $J_{[r]}=\{j_{1},\dots,j_{r}\}$ of $[L]$ with the same cardinality $r$,
    \begin{equation*}
        \int \mathbf{D}(\boldsymbol{\chi},\bar{\boldsymbol{\chi}})\;
        \Bigg(\prod_{\alpha=1}^{r}\bar{\chi}_{i_{\alpha}}\chi_{j_{\alpha}}\Bigg)
        e^{\bar{\boldsymbol{\chi}}^{T}\mathbf{A}\boldsymbol{\chi}}
        = \epsilon(I_{[r]},J_{[r]})\,
        \det\!\left(\mathbf{A}_{[I_{[r]}^{c},J_{[r]}^{c}]}\right).
    \end{equation*}

    \item For sequences $I=(i_{1},\dots,i_{r})$ and $J=(j_{1},\dots,j_{r})$ in $[L]$ of equal length $r$, if $\mathbf{A}$ is invertible,
    \begin{equation*}
        \int \mathbf{D}(\boldsymbol{\chi},\bar{\boldsymbol{\chi}})\;
        \Bigg(\prod_{\alpha=1}^{r}\bar{\chi}_{i_{\alpha}}\chi_{j_{\alpha}}\Bigg)
        e^{\bar{\boldsymbol{\chi}}^{T}\mathbf{A}\boldsymbol{\chi}}
        = \det(\mathbf{A})\,
        \det\!\left((\mathbf{A}^{-T})_{[I_{[r]},J_{[r]}]}\right).
    \end{equation*}

    \item For any $r\times L$ matrix $\mathbf{B}$ and $L\times r$ matrix $\mathbf{C}$,
    \begin{equation*}
        \int \mathbf{D}(\boldsymbol{\chi},\bar{\boldsymbol{\chi}})\;
        \Bigg(\prod_{\alpha=1}^{r}(\bar{\chi}\mathbf{C})_{\alpha}
        (\mathbf{B}\chi)_{\alpha}\Bigg)
        e^{\bar{\boldsymbol{\chi}}^{T}\mathbf{A}\boldsymbol{\chi}}
        =
        \sum_{i_{1}<\cdots<i_{r}=1}^{L}
        \sum_{j_{1}<\cdots<j_{r}=1}^{L}
        (\det \mathbf{B}_{[\star,J_{[r]}]})\,
        \epsilon(I_{[r]},J_{[r]})\,
        \det\!\left(\mathbf{A}_{[I_{[r]}^{c},J_{[r]}^{c}]}\right).
    \end{equation*}
\end{enumerate}

Having recalled the standard Gaussian formulas for real and complex fermions, we now turn to the general case. Our goal is to construct Grassmann--Berezin integrals that couple fermionic variables to external sources, thereby encompassing a broader class of functional representations.
\subsection{Grassmann-Berezin integrals for real fermions}
This subsection develops the most general formulation of Grassmann–Berezin integrals for real Grassmann variables coupled to both bosonic and fermionic sources. Explicit integral identities are derived, together with their reductions obtained by suppressing one or both classes of sources.
\subsubsection{Fermionic Pfaffian with Bosonic-Fermionic sources.}
\begin{theorem}\label{exponential with linear fermionic and bosonic sources}[Fermionic Pfaffian with Bosonic-Fermionic sources] 
Let $\chi_1,\dots,\chi_{L}$ and $\psi_1,\dots,\psi_{L}$ denote real fermionic variables, and let $u_1,\dots,u_L$ denote real (or complex) bosonic variables. Then the following statements hold:

\medskip
\noindent
\textbf{(a)} For any antisymmetric matrix $\mathbf{A}\in \mathbb{R}^{L\times L}$, we have
\begin{equation}\label{part (a)}
\boxed{
\begin{aligned}
\int\mathbf{D}\boldsymbol{\chi}&e^{\tfrac{1}{2}\boldsymbol{\chi}^T\mathbf{A}\boldsymbol{\chi}+\mathbf{u}^T\boldsymbol{\chi}+\boldsymbol{\psi}^T\boldsymbol{\chi}}\\
&=\begin{cases}
\begin{aligned}
&\pf(\mathbf{A})+\sum_{m=1}^{L/2}\left[
   \frac{(-1)^m}{(2m)!}
   \sum_{j_1,\dots,j_{2m}=1}^L
   \epsilon\left(J_{[2m]}\right)\pf\left(\mathbf{A}_{\left[J^c_{[2m]}\right]}\right)
   \left(\prod_{\alpha=1}^{2m}\psi_{j_\alpha}\right)\right.\\
&\qquad-\frac{(-1)^m}{(2m-1)!}
   \sum_{j_1,\dots,j_{2m-1}=1}^L
   \epsilon\left(J_{[2m-1]}\right)\pf\left(\mathbf{X}_{\left[J^c_{[2m-1]}\right]}\right)
   \left(\prod_{\alpha=1}^{2m-1}\psi_{j_\alpha}\right)
\Bigg], 
\end{aligned}&\text{even $L$,}\\[2ex]
\begin{aligned}
&\pf(\mathbf{X})+\sum_{m=1}^{(L-1)/2}\frac{(-1)^m}{(2m)!}\sum_{j_1,\dots,j_{2m}=1}^L\epsilon\left(J_{[2m]}\right)\pf\left(\mathbf{X}_{\left[J^c_{[2m]}\right]}\right)\left(\prod_{\alpha=1}^{2m}\psi_{j_\alpha}\right)\\
&\qquad-\sum_{m=1}^{(L+1)/2}\frac{(-1)^m}{(2m-1)!}\sum_{j_1,\dots,j_{2m-1}=1}^L\epsilon\left(J_{[2m-1]}\right)\pf\!\left(\mathbf{A}_{\left[J_{[2m-1]}^c\right]}\right)\left(\prod_{\alpha=1}^{2m-1}\psi_{j_\alpha}\right),\\
\end{aligned}&\text{odd $L$,}
\end{cases}
\end{aligned}}
\end{equation}
where $\mathbf{X}=\begin{pmatrix} 0 & \mathbf{u}^T\\ -\mathbf{u} & \mathbf{A} \end{pmatrix}$, and $J_{[k]}=\{j_1,\dots,j_{k}\}\subseteq[L]$ whose elements are ordered increasingly.
\medskip
\noindent

\textbf{(b)} For any subset $I=\{i_1,\dots,i_r\}\subseteq[L]$ with $i_1<i_2<\cdots<i_r$, we have:
\begin{equation}\label{part (b)}
\hspace*{-2cm}\boxed{
\begin{aligned}
&\int\mathbf{D}\boldsymbol{\chi}\left(\prod_{\alpha=1}^r\chi_{i_\alpha}\right)e^{\tfrac{1}{2}\boldsymbol{\chi}^T\mathbf{A}\boldsymbol{\chi}+\mathbf{u}^T\boldsymbol{\chi}+\boldsymbol{\psi}^T\boldsymbol{\chi}}\\
&=
\begin{cases}
\begin{aligned}
&(-1)^{L}\Bigg\{\epsilon(I_{[r]})\pf\left(\mathbf{A}_{[I_{[r]}^c]}\right)+\sum_{m=1}^{(L-r)/2}\left[
   \frac{(-1)^m}{(2m)!}
   \sum_{j_1,\dots,j_{2m}=1}^L
   \epsilon\left(I_{[r]}\cup J_{[2m]}\right)\pf\left(\mathbf{A}_{\left[I_{[r]}^c\cap J^c_{[2m]}\right]}\right)
   \left(\prod_{\alpha=1}^{2m}\psi_{j_\alpha}\right)\right.\\
&\qquad-\left.\frac{(-1)^m}{(2m-1)!}
   \sum_{j_1,\dots,j_{2m-1}=1}^L
   \epsilon\left(I_{[r]}\cup J_{[2m-1]}\right)\pf\left(\mathbf{X}_{\left[I_{[r]}^c\cap J^c_{[2m-1]}\right]}\right)
   \left(\prod_{\alpha=1}^{2m-1}\psi_{j_\alpha}\right)
\right]\Bigg\},
\end{aligned}&\text{even $r+L$,}\\[2ex]
\begin{aligned}
&(-1)^{L+1}\Bigg\{\epsilon(I_{[r]})\pf\left(\mathbf{X}_{[I_{[r]}^c]}\right)+\sum_{m=1}^{(L-r-1)/2}\frac{(-1)^m}{(2m)!}\sum_{j_1,\dots,j_{2m}=1}^L\epsilon\left(I_{[r]}\cup J_{[2m]}\right)\pf\left(\mathbf{X}_{\left[I_{[r]}^c\cap J^c_{[2m]}\right]}\right)\left(\prod_{\alpha=1}^{2m}\psi_{j_\alpha}\right)\\
&\qquad-\sum_{m=1}^{(L-r+1)/2}\frac{(-1)^m}{(2m-1)!}\sum_{j_1,\dots,j_{2m-1}=1}^L\epsilon\left(I_{[r]}\cup J_{[2m-1]}\right)\pf\left(\mathbf{A}_{\left[I_{[r]}^c\cap J^c_{[2m-1]}\right]}\right)\left(\prod_{\alpha=1}^{2m-1}\psi_{j_\alpha}\right)\Bigg\}.
\end{aligned} &\text{odd $r+L$.}
\end{cases}
\end{aligned}}
\end{equation}
\medskip
\noindent
\textbf{(c)} More generally, for any $r\times L$ matrix $\mathbf{C}$ with entries in $\mathbb{R}$, we have
\begin{equation}\label{part (d)}
\hspace*{-1.5cm}\boxed{
\begin{aligned}
&\int\mathbf{D}\boldsymbol{\chi}\left(\prod_{\alpha=1}^r (\mathbf{C}\boldsymbol{\chi})_\alpha\right)e^{\tfrac{1}{2}\boldsymbol{\chi}^T\mathbf{A}\boldsymbol{\chi}+\mathbf{u}^T\boldsymbol{\chi}+\boldsymbol{\psi}^T\boldsymbol{\chi}}\\
&=\begin{cases}
\begin{aligned}
&(-1)^{L}\sum_{i_1<\dots<i_{r}=1}^{L} 
\det(\mathbf{C}_{[\star,I_{[r]}]})\Bigg\{\epsilon(I_{[r]})\pf\left(\mathbf{A}_{[I_{[r]}^c]}\right)\\& \qquad+\sum_{m=1}^{(L-r)/2}\left[
   \frac{(-1)^m}{(2m)!}
   \sum_{j_1,\dots,j_{2m}=1}^L
   \epsilon\left(I_{[r]}\cup J_{[2m]}\right)\pf\left(\mathbf{A}_{\left[I_{[r]}^c\cap J^c_{[2m]}\right]}\right)
   \left(\prod_{\alpha=1}^{2m}\psi_{j_\alpha}\right)\right.\\
&\qquad \qquad-\left.\frac{(-1)^m}{(2m-1)!}
   \sum_{j_1,\dots,j_{2m-1}=1}^L
   \epsilon\left(I_{[r]}\cup J_{[2m-1]}\right)\pf\left(\mathbf{X}_{\left[I_{[r]}^c\cap J^c_{[2m-1]}\right]}\right)
   \left(\prod_{\alpha=1}^{2m-1}\psi_{j_\alpha}\right)
\right]\Bigg\},
\end{aligned}&\text{even $r+L$,}\\[2ex]
\begin{aligned}
&(-1)^{L+1}\sum_{i_1<\dots<i_{r}=1}^{L} 
\det(\mathbf{C}_{[\star,I_{[r]}]})\Bigg\{\epsilon(I_{[r]})\pf\left(\mathbf{X}_{[I_{[r]}^c]}\right)\\& \qquad+\sum_{m=1}^{(L-r-1)/2}\frac{(-1)^m}{(2m)!}\sum_{j_1,\dots,j_{2m}=1}^L\epsilon\left(I_{[r]}\cup J_{[2m]}\right)\pf\left(\mathbf{X}_{\left[I_{[r]}^c\cap J^c_{[2m]}\right]}\right)\left(\prod_{\alpha=1}^{2m}\psi_{j_\alpha}\right)\\
&\qquad \qquad-\sum_{m=1}^{(L-r+1)/2}\frac{(-1)^m}{(2m-1)!}\sum_{j_1,\dots,j_{2m-1}=1}^L\epsilon\left(I_{[r]}\cup J_{[2m-1]}\right)\pf\left(\mathbf{A}_{\left[I_{[r]}^c\cap J^c_{[2m-1]}\right]}\right)\left(\prod_{\alpha=1}^{2m-1}\psi_{j_\alpha}\right)\Bigg\},
\end{aligned} &\text{odd $r+L$,}
\end{cases}
\end{aligned}}
\end{equation}
where $C_{[\star,{I_{[r]}}]}$ denotes the submatrix of $C$ with columns in ${I}$.

The proof of the theorem is in Appendix \ref{Proof exponential with linear fermionic and bosonic sources}.
\end{theorem}

\begin{corollary}\label{exponential with linear fermionic and bosonic sources (invertible)}
    Fermionic Pfaffian with Bosonic-Fermionic sources (Invertible case)

     \medskip
\noindent
\textbf{(a)} If the matrices $\mathbf{A}$ and $\mathbf{X}$ are invertible then we have:
\begin{equation}
\hspace{-2cm}\boxed{\int\mathbf{D}\boldsymbol{\chi} e^{\frac{1}{2}\boldsymbol{\chi}^T\mathbf{A}\boldsymbol{\chi}+\mathbf{u}^T\boldsymbol{\chi}+\boldsymbol{\psi}^T\boldsymbol{\chi}}=\begin{cases}
    \pf(\mathbf{A})e^{\frac{1}{2}\boldsymbol{\psi}^T \mathbf{A}^{-1} \boldsymbol{\psi} +\mathbf{u}^T \mathbf{A}^{-1} \boldsymbol{\psi}},&\text{even $L$,}\\
    \sum\limits_{j=1}^{L} (-1)^{j} \, (\psi_j-u_j)\pf(A_{[j^c]})
\sum\limits_{m=0}^{\frac{L-1}{2}} \frac{1}{(2m+1)m!} 
\left( \frac{1}{2}\boldsymbol{\psi}_{[j^c]}^T \mathbf{A}_{[j^c]}^{-1} \boldsymbol{\psi}_{[j^c]} + \mathbf{u}_{[j^c]}^T\mathbf{A}_{[j^c]}^{-1} \boldsymbol{\psi}_{[j^c]} \right)^m,& \text{odd $L$.}
\end{cases}}
\end{equation}
\end{corollary}
\medskip
\noindent
\textbf{(b)} For any subset $I=\{i_1,\dots,i_r\}\subseteq[L]$ with $i_1<i_2<\cdots<i_r$, we have
\begin{equation}\label{part (c)}
\hspace*{-1.5cm}\boxed{
\begin{aligned}
&\int\mathbf{D}\boldsymbol{\chi}\left(\prod_{\alpha=1}^r\chi_{i_\alpha}\right)e^{\tfrac{1}{2}\boldsymbol{\chi}^T\mathbf{A}\boldsymbol{\chi}+\mathbf{u}^T\boldsymbol{\chi}+\boldsymbol{\psi}^T\boldsymbol{\chi}}\\&=
\begin{cases}
\begin{aligned}
&(-1)^{L}\epsilon(I_{[r]})\pf\left(\mathbf{A}_{[I_{[r]}^c]}\right)e^{\frac{1}{2}\boldsymbol{\psi}_{[I_{[r]}^c]}^T \mathbf{A}_{[I_{[r]}^c]}^{-1} \boldsymbol{\psi}_{[I_{[r]}^c]} +\mathbf{u}_{[I_{[r]}^c]}^T \mathbf{A}_{[I_{[r]}^c]}^{-1} \boldsymbol{\psi}_{[I_{[r]}^c]}},
\end{aligned}&\text{even $r+L$,}\\[2ex]
\begin{aligned}
&(-1)^{L+1} \sum\limits_{j=1}^{L} (-1)^{j} \, (\psi_j-u_j) \; \epsilon(I_{[r]}\setminus\{j\}) \operatorname{pf}(A_{[I_{[r]}^c\cap \{j\}^c]})\\
&\times
\sum\limits_{m=0}^{\left\lfloor \frac{L-1}{2} \right\rfloor} \frac{1}{(2m+1)\,  m!} 
\left( \frac{1}{2}\boldsymbol{\psi}_{[I_{[r]}^c\cap \{j\}^c]}^T \mathbf{A}_{[I_{[r]}^c\cap \{j\}^c]}^{-1} \boldsymbol{\psi}_{[I_{[r]}^c\cap \{j\}^c]}+\mathbf{u}_{[I_{[r]}^c\cap \{j\}^c]}^T \mathbf{A}_{[I_{[r]}^c\cap \{j\}^c]}^{-1} \boldsymbol{\psi}_{[I_{[r]}^c\cap \{j\}^c]} \right)^m,
\end{aligned} &\text{odd $r+L$.}
\end{cases}
\end{aligned}}
\end{equation}

\medskip
\noindent
\textbf{(c)} More generally, for any $r\times L$ matrix $\mathbf{C}$ with entries in $\mathbb{R}$, we have
\begin{equation}
\hspace*{-1.5cm}\boxed{
\begin{aligned}
&\int\mathbf{D}\boldsymbol{\chi}\left(\prod_{\alpha=1}^r (\mathbf{C}\boldsymbol{\chi})_\alpha\right)e^{\tfrac{1}{2}\boldsymbol{\chi}^T\mathbf{A}\boldsymbol{\chi}+\mathbf{u}^T\boldsymbol{\chi}+\boldsymbol{\psi}^T\boldsymbol{\chi}}\\
&=\begin{cases}
\begin{aligned}
&(-1)^{L}\sum_{i_1<\dots<i_{r}=1}^{L} 
\det(\mathbf{C}_{[\star,I_{[r]}]})\epsilon(I_{[r]})\pf\left(\mathbf{A}_{[I_{[r]}^c]}\right)e^{\frac{1}{2}\boldsymbol{\psi}_{[I_{[r]}^c]}^T \mathbf{A}_{[I_{[r]}^c]}^{-1} \boldsymbol{\psi}_{[I_{[r]}^c]} +\mathbf{u}_{[I_{[r]}^c]}^T \mathbf{A}_{[I_{[r]}^c]}^{-1} \boldsymbol{\psi}_{[I_{[r]}^c]}},
\end{aligned}&\text{even $r+L$,}\\[2ex]
\begin{aligned}
&(-1)^{L+1}\sum_{i_1<\dots<i_{r}=1}^{L} 
\det(\mathbf{C}_{[\star,I_{[r]}]}) \sum\limits_{j=1}^{L} (-1)^{j} \, (\psi_j-u_j) \; \epsilon(I_{[r]}\setminus\{j\}) \operatorname{pf}(A_{[I_{[r]}^c\cap \{j\}^c]}) \; \\
& \times 
\sum\limits_{m=0}^{\left\lfloor \frac{L-1}{2} \right\rfloor} \frac{1}{(2m+1)\,  m!} 
\left( \frac{1}{2}\boldsymbol{\psi}_{[I_{[r]}^c\cap \{j\}^c]}^T \mathbf{A}_{[I_{[r]}^c\cap \{j\}^c]}^{-1} \boldsymbol{\psi}_{[I_{[r]}^c\cap \{j\}^c]} + \mathbf{u}_{[I_{[r]}^c\cap \{j\}^c]}^T \mathbf{A}_{[I_{[r]}^c\cap \{j\}^c]}^{-1} \boldsymbol{\psi}_{[I_{[r]}^c\cap \{j\}^c]} \right)^m,
\end{aligned} &\text{odd $r+L$.}
\end{cases}
\end{aligned}}
\end{equation}
The proof of the corollary is in Appendix \ref{Proof exponential with linear fermionic and bosonic sources (invertible)}.

Theorem~\ref{exponential with linear fermionic and bosonic sources} thus provides the general Berezin integral for fermionic variables coupled to bosonic and fermionic sources. 
All subsequent corollaries follow by restricting to purely fermionic or bosonic sources, or by setting the sources to zero.

\subsubsection{Fermionic Pfaffian with Fermionic sources.}

The first simple case we address is setting the bosonic source to zero.
\begin{corollary}[Fermionic Pfaffian with Fermionic sources]\label{exponential with fermionic sources}
For any real fermions $\chi_1,\dots,\chi_{L}$, and $\psi_1,\dots,\psi_{L}$.

\medskip
\noindent
\textbf{(a)} For an $L\times L$ antisymmetric matrix $\mathbf{A}$, we have 
\begin{equation}
\hspace{-1cm}\boxed{
\begin{aligned}
\int\mathbf{D}\boldsymbol{\chi}&e^{\tfrac{1}{2}\boldsymbol{\chi}^T\mathbf{A}\boldsymbol{\chi}+\boldsymbol{\psi}^T\boldsymbol{\chi}}=
\begin{cases}
\begin{aligned}
&\pf(\mathbf{A})+\sum_{m=1}^{L/2}
   \frac{(-1)^m}{(2m)!}
   \sum_{j_1,\dots,j_{2m}=1}^L
   \epsilon\left(J_{[2m]}\right)\pf\left(\mathbf{A}_{\left[J^c_{[2m]}\right]}\right)
   \left(\prod_{\alpha=1}^{2m}\psi_{j_\alpha}\right),
\end{aligned}&\text{even $L$,}\\[2ex]
\begin{aligned}
&-\sum_{m=1}^{(L+1)/2}\frac{(-1)^m}{(2m-1)!}\sum_{j_1,\dots,j_{2m-1}=1}^L\epsilon\left(J_{[2m-1]}\right)\pf\!\left(\mathbf{A}_{\left[J_{[2m-1]}^c\right]}\right)\left(\prod_{\alpha=1}^{2m-1}\psi_{j_\alpha}\right),
\end{aligned}&\text{odd $L$,}
\end{cases}
\end{aligned}}
\end{equation}
where $J_{[k]}=\{j_1,\dots,j_{k}\}\subseteq[L]$, whose elements are ordered increasingly.

\medskip
\noindent
\textbf{(a$^{\prime}$)} When the matrices $\mathbf{A}$ and $\mathbf{X}$ are invertible, we have:
\begin{equation}
\boxed{\int\mathbf{D}\boldsymbol{\chi} e^{\frac{1}{2}\boldsymbol{\chi}^T\mathbf{A}\boldsymbol{\chi}+\boldsymbol{\psi}^T\boldsymbol{\chi}}=\begin{cases}
\pf(\mathbf{A})e^{\tfrac{1}{2}\boldsymbol{\psi}^T \mathbf{A}^{-1} \boldsymbol{\psi}},&\text{even $L$,}\\
\sum\limits_{j=1}^{L} (-1)^{j} \, \psi_j \; \operatorname{pf}(A_{[j^c]}) \; 
\sum\limits_{m=0}^{\frac{L-1}{2}} \frac{1}{(2m+1)\,  m!} 
\left( \frac{1}{2}\boldsymbol{\psi}_{[j^c]}^T \mathbf{A}_{[j^c]}^{-1} \boldsymbol{\psi}_{[j^c]} \right)^m,&\text{odd $L$.}
\end{cases}
    }
\end{equation}

\medskip
\noindent
\textbf{(b)} For any subset $I=\{i_1,\dots,i_r\}\subseteq[L]$ with $i_1<i_2<\cdots<i_r$, we have
\begin{equation}
\hspace{-2cm}\boxed{
\begin{aligned}
&\int\mathbf{D}\boldsymbol{\chi}\left(\prod_{\alpha=1}^r\chi_{i_\alpha}\right)e^{\tfrac{1}{2}\boldsymbol{\chi}^T\mathbf{A}\boldsymbol{\chi}+\boldsymbol{\psi}^T\boldsymbol{\chi}}\\
&=
\begin{cases}
\begin{aligned}
&(-1)^{L}\Bigg[\epsilon(I_{[r]})\pf\left(\mathbf{A}_{[I_{[r]}^c]}\right)+\sum_{m=1}^{(L-r)/2}\frac{(-1)^m}{(2m)!}\sum_{j_1,\dots,j_{2m}=1}^L\epsilon\left(I_{[r]}\cup J_{[2m]}\right)\pf\left(\mathbf{A}_{\left[I_{[r]}^c\cap J^c_{[2m]}\right]}\right)\left(\prod_{\alpha=1}^{2m}\psi_{j_\alpha}\right)\Bigg],
\end{aligned}&\text{even $r+L$,}\\[2ex]
\begin{aligned}
&(-1)^{L}\Bigg[\sum_{m=1}^{(L-r+1)/2}\frac{(-1)^m}{(2m-1)!}\sum_{j_1,\dots,j_{2m-1}=1}^L\epsilon\left(I_{[r]}\cup J_{[2m-1]}\right)\pf\left(\mathbf{A}_{\left[I_{[r]}^c\cap J^c_{[2m-1]}\right]}\right)\left(\prod_{\alpha=1}^{2m-1}\psi_{j_\alpha}\right)\Bigg],
\end{aligned} &\text{odd $r+L$.}
\end{cases}
\end{aligned}}
\end{equation}

\medskip
\noindent
\textbf{(b$^{\prime}$)} For the invertible $\mathbf{A}$ and $\mathbf{X}$ matrices, we have
\begin{equation}
\boxed{
\begin{aligned}
&\int\mathbf{D}\boldsymbol{\chi}\left(\prod_{\alpha=1}^r\chi_{i_\alpha}\right)e^{\tfrac{1}{2}\boldsymbol{\chi}^T\mathbf{A}\boldsymbol{\chi}+\boldsymbol{\psi}^T\boldsymbol{\chi}}\\&=
\begin{cases}
\begin{aligned}
&(-1)^{L}\epsilon(I_{[r]})\pf\left(\mathbf{A}_{[I_{[r]}^c]}\right)e^{\frac{1}{2}\boldsymbol{\psi}_{[I_{[r]}^c]}^T \mathbf{A}_{[I_{[r]}^c]}^{-1} \boldsymbol{\psi}_{[I_{[r]}^c]} },
\end{aligned}&\text{even $r+L$,}\\[2ex]
\begin{aligned}
&(-1)^{L+1}\sum\limits_{j=1}^{L} (-1)^{j} \, \epsilon(I_{[r]}\setminus\{j\}) \psi_j \;  \operatorname{pf}(A_{[I_{[r]}^c\cap \{j\}^c]}) \; \\& \times
\sum\limits_{m=0}^{\left\lfloor \frac{L-1}{2} \right\rfloor} \frac{1}{(2m+1)\,  m!} 
\left( \frac{1}{2}\boldsymbol{\psi}_{[I_{[r]}^c\cap \{j\}^c]}^T \mathbf{A}_{[I_{[r]}^c\cap \{j\}^c]}^{-1} \boldsymbol{\psi}_{[I_{[r]}^c\cap \{j\}^c]} \right)^m,
\end{aligned} &\text{odd $r+L$.}
\end{cases}
\end{aligned}}
\end{equation}

\medskip
\noindent
\textbf{(c)} More generally, for any $r\times L$ matrix $\mathbf{C}$ with entries in $\mathbb{R}$, we have
\begin{equation}
\boxed{
\begin{aligned}
&\int\mathbf{D}\boldsymbol{\chi}\left(\prod_{\alpha=1}^r (\mathbf{C}\boldsymbol{\chi})_\alpha\right)e^{\tfrac{1}{2}\boldsymbol{\chi}^T\mathbf{A}\boldsymbol{\chi}+\boldsymbol{\psi}^T\boldsymbol{\chi}}\\&=\begin{cases}
\begin{aligned}
&(-1)^{L}\sum_{i_1<\dots<i_{r}=1}^{L} 
\det(\mathbf{C}_{[\star,I_{[r]}]})\Bigg[\epsilon(I_{[r]})\pf\left(\mathbf{A}_{[I_{[r]}^c]}\right)\\&\qquad +\sum_{m=1}^{(L-r)/2}
   \frac{(-1)^m}{(2m)!}
   \sum_{j_1,\dots,j_{2m}=1}^L
   \epsilon\left(I_{[r]}\cup J_{[2m]}\right) \pf\left(\mathbf{A}_{\left[I_{[r]}^c\cap J^c_{[2m]}\right]}\right)
   \left(\prod_{\alpha=1}^{2m}\psi_{j_\alpha}\right)\Bigg],
\end{aligned}&\text{even $r+L$,}\\[2ex]
\begin{aligned}
(-1)^{L}\sum_{i_1<\dots<i_{r}=1}^{L} 
\det(\mathbf{C}_{[\star,I_{[r]}]})\Bigg[&\sum_{m=1}^{(L-r+1)/2}\frac{(-1)^m}{(2m-1)!}\sum_{j_1,\dots,j_{2m-1}=1}^L\epsilon\left(I_{[r]}\cup J_{[2m-1]}\right) \\ &\times\pf\left(\mathbf{A}_{\left[I_{[r]}^c\cap J^c_{[2m-1]}\right]}\right)\left(\prod_{\alpha=1}^{2m-1}\psi_{j_\alpha}\right)\Bigg],
\end{aligned} &\text{odd $r+L$.}
\end{cases}
\end{aligned}}
\end{equation}

\medskip
\noindent
\textbf{(c$^{\prime}$)} For the invertible $\mathbf{A}$ and $\mathbf{X}$ matrices, we have
\begin{equation}
\boxed{
\begin{aligned}
&\int\mathbf{D}\boldsymbol{\chi}\left(\prod_{\alpha=1}^r (\mathbf{C}\boldsymbol{\chi})_\alpha\right)e^{\tfrac{1}{2}\boldsymbol{\chi}^T\mathbf{A}\boldsymbol{\chi}+\boldsymbol{\psi}^T\boldsymbol{\chi}}\\&=\begin{cases}
\begin{aligned}
&(-1)^{L}\sum_{i_1<\dots<i_{r}=1}^{L} 
\det(\mathbf{C}_{[\star,I_{[r]}]})\epsilon(I_{[r]})\pf\left(\mathbf{A}_{[I_{[r]}^c]}\right)e^{\frac{1}{2}\boldsymbol{\psi}_{[I_{[r]}^c]}^T \mathbf{A}_{[I_{[r]}^c]}^{-1} \boldsymbol{\psi}_{[I_{[r]}^c]} },
\end{aligned}&\text{even $r+L$,}\\[2ex]
\begin{aligned}
&(-1)^{L+1}\sum_{i_1<\dots<i_{r}=1}^{L} 
\det(\mathbf{C}_{[\star,I_{[r]}]}) \sum\limits_{j=1}^{L} (-1)^{j} \, \epsilon(I_{[r]}\setminus\{j\}) \psi_j \;  \operatorname{pf}(A_{[I_{[r]}^c\cap \{j\}^c]}) \; \\& \times
\sum\limits_{m=0}^{\left\lfloor \frac{L-1}{2} \right\rfloor} \frac{1}{(2m+1)\,  m!} 
\left( \frac{1}{2}\boldsymbol{\psi}_{[I_{[r]}^c\cap \{j\}^c]}^T \mathbf{A}_{[I_{[r]}^c\cap \{j\}^c]}^{-1} \boldsymbol{\psi}_{[I_{[r]}^c\cap \{j\}^c]} \right)^m,
\end{aligned} &\text{odd $r+L$.}
\end{cases}
\end{aligned}}
\end{equation}

\end{corollary}
\subsubsection{Fermionic Pfaffian with Bosonic sources.}
We can set the fermionic sources to zero and have pure bosonic sources.

\begin{corollary}[Fermionic Pfaffian with Bosonic sources]\label{exponential with bosonic sources}
For any real fermions $\chi_1,\dots,\chi_{L}$, and real (or complex) bosons $u_1,\dots,u_L$.

\medskip
\noindent
\textbf{(a)} For an $L\times L$ antisymmetric matrix $\mathbf{A}$, we have 
\begin{equation}
\boxed{
\begin{aligned}
\int\mathbf{D}\boldsymbol{\chi}e^{\tfrac{1}{2}\boldsymbol{\chi}^T\mathbf{A}\boldsymbol{\chi}+\mathbf{u}^T\boldsymbol{\chi}}=
\begin{cases}
\begin{aligned}
&\pf(\mathbf{A}), 
\end{aligned}&\text{even $L$,}\\[2ex]
\begin{aligned}
&-\sum_{j_0=1}^L(-1)^{j_0}u_{j_0}\pf\left(\mathbf{A}_{[\{j_0\}^c]}\right) \equiv \operatorname{pf}(\mathbf{X}),
\end{aligned}&\text{odd $L$.}
\end{cases}
\end{aligned}}
\end{equation}
For the odd case, we can rewrite:
\begin{equation}
    \begin{split}
        -\sum_{j_0=1}^L(-1)^{j_0}u_{j_0}\pf\left(\mathbf{A}_{[\{j_0\}^c]}\right)=\sum_{j_0=1}^L(-1)^{j_0+1}u_{j_0}\pf\left(\mathbf{A}_{[\{j_0\}^c]}\right)=\sum_{k=2}^{L+1}(-1)^{k}u_{k-1}\pf\left(\mathbf{A}_{[\{k-1\}^c]}\right),
    \end{split}
\end{equation}
where $k=j_0+1$. From the structure of the matrix \( \mathbf{X} \), it follows that $u_{k-1}=X_{1k}$ and $A_{[\{k-1\}^c]}=X_{[\{1,k\}^c]}$, then we obtain:
\begin{equation}
    \begin{split}
    \sum_{k=2}^{L+1}(-1)^{k}X_{1k}\pf\left(\mathbf{X}_{[\{1,k\}^c]}\right)=\operatorname{pf}(\mathbf{X}).
    \end{split}
\end{equation}

\medskip
\noindent
\textbf{(b)} For any subset $I=\{i_1,\dots,i_r\}\subseteq[L]$ with $i_1<i_2<\cdots<i_r$, we have
\begin{equation}
\boxed{
\begin{aligned}
\int\mathbf{D}\boldsymbol{\chi}\left(\prod_{\alpha=1}^r\chi_{i_\alpha}\right)e^{\tfrac{1}{2}\boldsymbol{\chi}^T\mathbf{A}\boldsymbol{\chi}+\mathbf{u}^T\boldsymbol{\chi}}=
\begin{cases}
(-1)^{L}\epsilon(I_{[r]})\pf\left(\mathbf{A}_{[I_{[r]}^c]}\right),& \text{even $r+L$,}\\[2ex]
(-1)^{L+1}\epsilon(I_{[r]})\pf\left(\mathbf{X}_{[I_{[r]}^c]}\right),&\text{odd $r+L$.}
\end{cases}  
\end{aligned}}
\end{equation}

\medskip
\noindent
\textbf{(c)} More generally, for any $r\times L$ matrix $\mathbf{C}$ with entries in $\mathbb{R}$, we have
\begin{equation}
\boxed{
\int\mathbf{D}\boldsymbol{\chi}\left(\prod_{\alpha=1}^r (\mathbf{C}\boldsymbol{\chi})_\alpha\right)e^{\tfrac{1}{2}\boldsymbol{\chi}^T\mathbf{A}\boldsymbol{\chi}+\mathbf{u}^T\boldsymbol{\chi}}=\begin{cases}
\begin{aligned}
&(-1)^{L}\sum_{i_1<\dots<i_{r}=1}^{L} \epsilon(I_{[r]}) 
\det(\mathbf{C}_{[\star,I_{[r]}]}) \pf\left(\mathbf{A}_{[I_{[r]}^c]}\right),
\end{aligned}\text{even $r+L$,}\\[2ex]
\begin{aligned}
&(-1)^{L+1}\sum_{i_1<\dots<i_{r}=1}^{L}  \epsilon(I_{[r]})
\det(\mathbf{C}_{[\star,I_{[r]}]}) \pf\left(\mathbf{X}_{[I_{[r]}^c]}\right),
\end{aligned}\text{odd $r+L$.}
\end{cases} }
\end{equation}
\end{corollary}

\subsubsection{Special pure source bosonic and fermionic term.}
\begin{corollary}\label{pure fermionicbosonic sources}
Let $\mathbf{A}$ be an $N \times L$ matrix with $N \geq L$, and consider the Grassmann vectors and the complex vector, respectively
\[
\boldsymbol{\psi}=\begin{pmatrix} \psi_1 & \cdots & \psi_N \end{pmatrix}, 
\qquad 
\boldsymbol{\chi}=\begin{pmatrix} \chi_1 \\ \vdots \\ \chi_L \end{pmatrix},\qquad\boldsymbol{u}=\begin{pmatrix} u_1& \cdots & u_N \end{pmatrix}.
\]
Then, the following identity holds:
\begin{equation}
\hspace{-1.5cm}\boxed{\int\mathbf{D}\boldsymbol{\chi} \, e^{(\boldsymbol{\psi}+\boldsymbol{u})\mathbf{A}\boldsymbol{\chi}}
    = \frac{(-1)^{\tfrac{L(L+1)}{2}}}{L!}
      \sum_{i_1,\dots,i_L=1}^N
      \det\!\bigl(\mathbf{A}_{[I_{[L]},\star]}\bigr)\,
      \prod_{\alpha=1}^L\psi_{i_\alpha}+\frac{(-1)^{\tfrac{L(L-1)}{2}}}{(L-1)!}
      \sum_{i_1,\dots,i_L=1}^N
      \det\!\bigl(\mathbf{A}_{[I_{[L]},\star]}\bigr)u_{i_1}\prod_{\alpha=2}^L\psi_{i_\alpha},}
\end{equation}
where $I = \{i_1,\dots,i_L\} \subset [N]$, and $\mathbf{A}_{[I_{[L]},\star]}$ denotes the submatrix of $\mathbf{A}$ formed by selecting all columns and the rows indexed by $I$.
\end{corollary}

\begin{lemma}\label{pure fermionic sources}
Let $\mathbf{A}$ be an $N \times L$ matrix with $N \geq L$, and consider the Grassmann vectors
\[
\boldsymbol{\psi}=\begin{pmatrix} \psi_1 & \cdots & \psi_N \end{pmatrix}, 
\qquad 
\boldsymbol{\chi}=\begin{pmatrix} \chi_1 \\ \vdots \\ \chi_L \end{pmatrix}.
\]
Then, the following identity holds:
\begin{equation}
    \boxed{\int \mathbf{D}\boldsymbol{\chi} \, e^{\boldsymbol{\psi}\mathbf{A}\boldsymbol{\chi}}
    = (-1)^{\tfrac{L(L+1)}{2}}
      \sum_{i_1<\dots<i_L=1}^N
      \det\!\bigl(\mathbf{A}_{[I_{[L]},\star]}\bigr)\,
      \psi_{i_1}\cdots\psi_{i_L},}
\end{equation}
where $I = \{i_1,\dots,i_L\} \subset [N]$, and $\mathbf{A}_{[I_{[L]},\star]}$ denotes the submatrix of $\mathbf{A}$ formed by selecting all columns and the rows indexed by $I$. In the special case $N=L$, this reduces to
\begin{equation}
\int \mathbf{D}\boldsymbol{\chi} \, e^{\boldsymbol{\psi}\mathbf{A}\boldsymbol{\chi}}
= (-1)^{\tfrac{L(L+1)}{2}} \det(\mathbf{A}) \prod_{i=1}^L \psi_i.
\end{equation}

The proof is in Appendix \ref{Proof pure fermionic sources}.
\end{lemma}

\begin{lemma}\label{generalized pure fermionic sources}
Let $\mathbf{A}^{(j)}$ be an $n^{(j)}\times m$ matrix with $j=1,\dots,N$, and $n^{(j)}\geq m$. Consider the Grassmann vectors
\[
\boldsymbol{\psi}^{(j)}= \begin{pmatrix} \psi^{(j)}_1 & \cdots & \psi^{(j)}_n \end{pmatrix}, 
\qquad 
\boldsymbol{\chi} = \begin{pmatrix} \chi_1 \\ \vdots \\ \chi_m \end{pmatrix}.
\]
Then,
\begin{equation}
\hspace{-1.5cm}\boxed{\begin{split}
  \int&\mathbf{D}\boldsymbol{\chi} e^{\boldsymbol{\psi}^{(1)}\mathbf{A}^{(1)}\boldsymbol{\chi}+\boldsymbol{\psi}^{(2)}\mathbf{A}^{(2)}\boldsymbol{\chi}+\cdots+\boldsymbol{\psi}^{(N)}\mathbf{A}^{(N)}\boldsymbol{\chi}}\\
  &=(-1)^{\tfrac{m(m+1)}{2}}\sum_{i_1^{(1)}<\dots<i_m^{(1)}=1}^{n^{(1)}}\cdots\sum_{i_1^{(N)}<\dots<i_m^{(N)}=1}^{n^{(N)}}\det\left(\mathbb{A}_{\left[\left\{i_1^{(1)},\dots,i_m^{(1)},\dots,i_1^{(N)},\dots,i_m^{(N)}\right\},\star\right]}\right)\psi_{i_1^{(1)}}^{(1)}\cdots\psi_{i_m^{(1)}}^{(1)}\cdots\psi_{i_1^{(N)}}^{(N)}\cdots\psi_{i_m^{(N)}}^{(N)},
\end{split}}
\end{equation}
where $\mathbb{A} = \begin{pmatrix}
\mathbf{A}^{(1)} \\ \mathbf{A}^{(2)} \\ \vdots \\ \mathbf{A}^{(N)}
\end{pmatrix}$. The proof is in Appendix \ref{proof generalized pure fermionic sources}
\end{lemma}

\subsection{Grassmann-Berezin integrals for complex fermions}
This subsection develops the Grassmann–Berezin formalism for complex fermionic variables. We address the most general case, where complex Grassmann fields are coupled simultaneously to bosonic and fermionic sources. Building on the general derivation, we examine several instructive special cases, emphasizing how the presence of a complex structure alters the properties of the integrals compared with the real-fermion formulation, and clarify the correspondence between both frameworks.
\subsubsection{Fermionic Determinant with Bosonic-Fermionic sources}
\begin{theorem}[Fermionic Determinant with Bosonic-Fermionic sources]\label{theorem 2.2}
Let $\chi^1,\dots,\chi^N$, $\bar{\chi}^1,\dots,\bar{\chi}^N$, and $\psi^1,\dots,\psi^N$, $\bar{\psi}^1,\dots,\bar{\psi}^N$ denote complex fermionic variables, and let $\bar{u}_1,\dots,\bar{u}_L,u_1,\dots,u_L$ denote real (or complex) bosonic variables.

\medskip
\noindent
\textbf{(a)} For an $L\times L$ matrix $\mathbf{A}$, the following relations hold:
 \begin{equation}\label{complex fermion with all linear term}
\hspace{-2cm}\boxed{\begin{aligned}
&\int \mathbf{D}(\boldsymbol{\chi},\bar{\boldsymbol{\chi}}) e^{\bar{\boldsymbol{\chi}}^T \mathbf{A}\boldsymbol{\chi}+\bar{\boldsymbol{\psi}}^T\boldsymbol{\chi}+\bar{\boldsymbol{\chi}}^T \boldsymbol{\psi}+\bar{\mathbf{u}}^T\boldsymbol{\chi}+\bar{\boldsymbol{\chi}}^T \mathbf{u}}\\&=\det(\mathbf{A})+
   \sum_{m= 1}^L
   \frac{(-1)^{m}}{(m!)^2}
   \sum_{j_{1},\dots,j_{m}=1}^{L}
   \sum_{i_{1},\dots,i_{m}=1}^{L}
   \epsilon\left(I_{[m]},J_{[m]}\right) \det\left(\mathbf{A}_{\left[I_{[m]}^c,J_{[m]}^c\right]}\right)
     \prod_{\alpha=1}^{m} \bar{\psi}_{j_\alpha}
     \prod_{\beta=1}^{m}
     \psi_{i_\beta}\notag\\
&
 -\det(\bar{\mathbf{Y}})-\sum_{m= 1}^L
   \frac{1}{m!(m-1)!}
   \sum_{j_{1},\dots,j_{m-1}=1}^{L}
   \sum_{i_{2},\dots,i_{m}=1}^{L}
   \epsilon\left(I_{[m]}\setminus\{i_1\},J_{[m-1]}\right) \det\left(\bar{\mathbf{Y}}_{\left[(I_{[m]}\setminus \{i_1\})^c,J_{[m-1]}^c\right]}\right)
     \prod_{\alpha=1}^{m-1} \bar{\psi}_{j_\alpha}
     \prod_{\beta=2}^{m}
     \psi_{i_\beta} \notag\\
     &-\det(\mathbf{Y})+ \sum_{m= 1}^L
   \frac{(-1)^{m}}{m!(m-1)!}
   \sum_{j_{2},\dots,j_{m}=1}^{L}
   \sum_{i_{1},\dots,i_{m-1}=1}^{L}
   \epsilon\left(I_{[m-1]},J_{[m]}\setminus \{j_1\}\right)\det\left(\mathbf{Y}_{\left[I_{[m-1]}^c,(J_{[m]}\setminus \{j_1\})^c\right]}\right)
     \prod_{\alpha=2}^{m} \bar{\psi}_{j_\alpha}
     \prod_{\beta=1}^{m-1}
     \psi_{i_\beta},
\end{aligned}}
\end{equation}

where 
\begin{equation}
   \bar{\mathbf{Y}}=\begin{pmatrix} 0 & \bar{\mathbf{u}}^T \\ \boldsymbol{\psi} & \mathbf{A} \end{pmatrix}, \quad  \mathbf{Y}=\begin{pmatrix} 0 & \boldsymbol{\bar{\psi}}^T \\ \mathbf{u} & \mathbf{A} \end{pmatrix}.
\end{equation}

\medskip
\noindent
\textbf{(b)} For any subsets ${I}=\{i_1,...,i_r\}$ and ${J}=\{j_1,...,j_r\} \subseteq [L]$ having the same cardinality $r$, with $i_1<...<i_r$ and $j_1<...<j_r$, we have:
 \begin{equation}
 \hspace{-2.7cm}
\boxed{\begin{aligned}
&\int \mathbf{D}(\boldsymbol{\chi},\bar{\boldsymbol{\chi}}) \bigg(\prod_{\alpha=1}^{r} \bar{{\chi}} _{i_{\alpha}} {\chi}_{j_{\alpha}}\bigg)e^{\bar{\boldsymbol{\chi}}^T \mathbf{A}\boldsymbol{\chi}+\bar{\boldsymbol{\psi}}^T\boldsymbol{\chi}+\bar{\boldsymbol{\chi}}^T \boldsymbol{\psi}+\bar{\mathbf{u}}^T\boldsymbol{\chi}+\bar{\boldsymbol{\chi}}^T \mathbf{u}}=\epsilon(I_{[r]},J_{[r]})\det[\mathbf{A}_{[I_{[r]}^c,J_{[r]}^c]}]\\[6pt] &+
   \sum_{m= 1}^{L-r}
   \frac{(-1)^{m}(-1)^{r(L-r)}}{(m!)^2}
   \sum_{k_{1},\dots,k_{m}=1}^{L}
   \sum_{\ell_{1},\dots,\ell_{m}=1}^{L}
   \epsilon(I_{[r]}\cup \mathscr{L}_{[m]},J_{[r]}\cup K_{[m]}) \det[\mathbf{A}_{[I_{[r]}^c\cap \mathscr{L}_{[m]}^c ,J_{[r]}^c\cap K_{[m]}^c]}]
     \prod_{\alpha=1}^{m} \bar{\psi}_{k_\alpha}
     \prod_{\beta=1}^{m}
     \psi_{\ell_\beta}\\[6pt]
&-(-1)^{r}  \epsilon(I_{[r]},J_{[r]}) \det[\bar{\mathbf{Y}}_{[I_{[r]}^c,J_{[r]}^c]}] -(-1)^{r}  \epsilon(I_{[r]},J_{[r]})\det[\mathbf{Y}_{[I_{[r]}^c,J_{[r]}^c]} ]\\[6pt]
& - \sum_{m= 1}^{L-r}
   \frac{(-1)^{r(L-r)}}{m!(m-1)!}
   \sum_{k_{1},\dots,k_{m-1}=1}^{L}
   \sum_{\ell_{2},\dots,\ell_{m}=1}^{L}
   \epsilon(I_{[r]}\cup \mathscr{L}_{[m]}\setminus \{\ell_1\},J_{[r]}\cup K_{[m-1]}) \det[\bar{\mathbf{Y}}_{[I_{[r]}^c\cap (\mathscr{L}_{[m]}\setminus \{\ell_1\})^c ,J_{[r]}^c\cap K_{[m-1]}^c]}]
     \prod_{\alpha=1}^{m-1} \bar{\psi}_{k_\alpha}
     \prod_{\beta=2}^{m}
     \psi_{\ell_\beta}\\[6pt]
& + \sum_{m= 1}^{L-r}
   \frac{(-1)^{m}(-1)^{r(L-r)}}{m!(m-1)!}
   \sum_{k_{2},\dots,k_{m}=1}^{L}
   \sum_{\ell_{1},\dots,\ell_{m-1}=1}^{L}
   \epsilon(I_{[r]}\cup \mathscr{L}_{[m-1]},J_{[r]}\cup K_{[m]}\setminus \{k_1\}) \det[\mathbf{Y}_{[I_{[r]}^c\cap \mathscr{L}_{[m-1]}^c ,J_{[r]}^c\cap (K_{[m]}\setminus \{k_1\})^c]}]\times\\
   &\qquad\qquad\qquad\qquad\qquad\qquad\times\prod_{\alpha=2}^{m} \bar{\psi}_{k_\alpha}
     \prod_{\beta=1}^{m-1}
     \psi_{\ell_\beta}.
\end{aligned}}
\end{equation}

\medskip
\noindent
\textbf{(c)} For any $r\times L$ matrix $B$ and $L\times r$ matrix $C$, we have:
 \begin{equation}\label{eq: sixth integral}
\hspace{-2.4cm}
\boxed{\begin{split}
&\int \mathbf{D}(\boldsymbol{\chi},\bar{\boldsymbol{\chi}}) \bigg(\prod_{\alpha=1}^{r} (\bar{{\chi}}C)_{\alpha}  (B{\chi})_{\alpha}\bigg) e^{\bar{\boldsymbol{\chi}}^T \mathbf{A}\boldsymbol{\chi}+\bar{\boldsymbol{\psi}}^T\boldsymbol{\chi}+\bar{\boldsymbol{\chi}}^T \boldsymbol{\psi}+\bar{\mathbf{u}}^T\boldsymbol{\chi}+\bar{\boldsymbol{\chi}}^T \mathbf{u}}\\&=
   \sum_{i_1<\dots<i_{r}=1}^{L} \sum_{j_1<\dots<j_{r}=1}^{L} (\det B_{[\star,{J}_{[r]}]})  \Bigg\{ \epsilon(I_{[r]},J_{[r]})\det[\mathbf{A}_{[I_{[r]}^c,J_{[r]}^c]}]\\[6pt] &+
   \sum_{m= 1}^{L-r}
   \frac{(-1)^{m}(-1)^{r(L-r)}}{(m!)^2}
   \sum_{k_{1},\dots,k_{m}=1}^{L}
   \sum_{\ell_{1},\dots,\ell_{m}=1}^{L}
   \epsilon(I_{[r]}\cup \mathscr{L}_{[m]},J_{[r]}\cup K_{[m]}) \det[\mathbf{A}_{[I_{[r]}^c\cap \mathscr{L}_{[m]}^c ,J_{[r]}^c\cap K_{[m]}^c]}]
     \prod_{\alpha=1}^{m} \bar{\psi}_{k_\alpha}
     \prod_{\beta=1}^{m}
     \psi_{\ell_\beta}\notag\\[6pt]
&-(-1)^{r}  \epsilon(I_{[r]},J_{[r]}) \det[\bar{\mathbf{Y}}_{[I_{[r]}^c,J_{[r]}^c]}] \notag\\[6pt]
& - \sum_{m= 1}^{L-r}
   \frac{(-1)^{r(L-r)}}{m!(m-1)!}
   \sum_{k_{1},\dots,k_{m-1}=1}^{L}
   \sum_{\ell_{2},\dots,\ell_{m}=1}^{L}
   \epsilon(I_{[r]}\cup \mathscr{L}_{[m]}\setminus \{\ell_1\},J_{[r]}\cup K_{[m-1]}) \det[\bar{\mathbf{Y}}_{[I_{[r]}^c\cap (\mathscr{L}_{[m]}\setminus \{\ell_1\})^c ,J_{[r]}^c\cap K_{[m-1]}^c]}]
     \prod_{\alpha=1}^{m-1} \bar{\psi}_{k_\alpha}
     \prod_{\beta=2}^{m}
     \psi_{\ell_\beta}\notag\\[6pt]
& + \sum_{m= 1}^{L-r}
   \frac{(-1)^{m}(-1)^{r(L-r)}}{m!(m-1)!}
   \sum_{k_{2},\dots,k_{m}=1}^{L}
   \sum_{\ell_{1},\dots,\ell_{m-1}=1}^{L}
   \epsilon(I_{[r]}\cup \mathscr{L}_{[m-1]},J_{[r]}\cup K_{[m]}\setminus \{k_1\}) \det[\mathbf{Y}_{[I_{[r]}^c\cap \mathscr{L}_{[m-1]}^c ,J_{[r]}^c\cap (K_{[m]}\setminus \{k_1\})^c]}]\times\\[6pt] 
     &\qquad\qquad\times\prod_{\alpha=2}^{m} \bar{\psi}_{k_\alpha}
     \prod_{\beta=1}^{m-1}
     \psi_{\ell_\beta}-(-1)^{r}  \epsilon(I_{[r]},J_{[r]})\det[\mathbf{Y}_{[I_{[r]}^c,J_{[r]}^c]} ] \Bigg\} (\det C_{[{I}_{[r]} ,\star]}).
\end{split}}
\end{equation}
\end{theorem}
\begin{corollary}[Fermionic Determinant with Bosonic-Fermionic sources (invertible case)]
Let $\chi^1,\dots,\chi^N$, $\bar{\chi}^1,\dots,\bar{\chi}^N$, and $\psi^1,\dots,\psi^N$, $\bar{\psi}^1,\dots,\bar{\psi}^N$ denote complex fermionic variables, and let $\bar{u}_1,\dots,\bar{u}_L,u_1,\dots,u_L$ denote real (or complex) bosonic variables.

\medskip
\noindent
\textbf{(a)} For an invertible $L\times L$ matrix $\mathbf{A}$, the following relation hold:
\begin{equation}
\boxed{\begin{aligned}
\int \mathbf{D}(\boldsymbol{\chi},\bar{\boldsymbol{\chi}}) e^{\bar{\boldsymbol{\chi}}^T \mathbf{A}\boldsymbol{\chi}+\bar{\boldsymbol{\psi}}^T\boldsymbol{\chi}+\bar{\boldsymbol{\chi}}^T \boldsymbol{\psi}+\bar{\mathbf{u}}^T\boldsymbol{\chi}+\bar{\boldsymbol{\chi}}^T \mathbf{u}}=\det(\mathbf{A})e^{-\bar{\boldsymbol{\psi}}^T \mathbf{A}^{-1}\boldsymbol{\psi}-\bar{\mathbf{u}}^T \mathbf{A}^{-1}\boldsymbol{\psi}-\bar{\boldsymbol{\psi}}^T \mathbf{A}^{-1}\mathbf{u}}.
\end{aligned}}
\end{equation}

\medskip
\noindent
\textbf{(b)} For an invertible $L\times L$ matrix $\mathbf{A}$, and for any subsets ${I}=\{i_1,...,i_r\}$ and ${J}=\{j_1,...,j_r\} \subseteq [L]$ having the same cardinality $r$, with $i_1<...<i_r$ and $j_1<...<j_r$, we have:
\begin{equation}
\hspace{-2cm}\boxed{\begin{aligned}
\int \mathbf{D}(\boldsymbol{\chi},\bar{\boldsymbol{\chi}}) \bigg(\prod_{\alpha=1}^{r} \bar{{\chi}} _{i_{\alpha}} {\chi}_{j_{\alpha}}\bigg)e^{\bar{\boldsymbol{\chi}}^T \mathbf{A}\boldsymbol{\chi}+\bar{\boldsymbol{\psi}}^T\boldsymbol{\chi}+\bar{\boldsymbol{\chi}}^T \boldsymbol{\psi}+\bar{\mathbf{u}}^T\boldsymbol{\chi}+\bar{\boldsymbol{\chi}}^T \mathbf{u}}&= \epsilon(I_{[r]},J_{[r]})\det[\mathbf{A}_{[I_{[r]}^c,J_{[r]}^c]}]\\& \times e^{-\bar{\boldsymbol{\psi}}_{[J_{[r]}^c]}^T \mathbf{A}^{-1}_{[I_{[r]}^c,J_{[r]}^c]}\boldsymbol{\psi}_{[I_{[r]}^c]}-\bar{\bold{u}}_{[J_{[r]}^c]}^T \mathbf{A}^{-1}_{[I_{[r]}^c,J_{[r]}^c]}\boldsymbol{\psi}_{[I_{[r]}^c]}-\bar{\boldsymbol{\psi}}_{[J_{[r]}^c]}^T \mathbf{A}^{-1}_{[I_{[r]}^c,J_{[r]}^c]}\bold{u}_{[I_{[r]}^c]}}.
\end{aligned}}
\end{equation}

\medskip
\noindent
\textbf{(c)} For an invertible $L\times L$ matrix $\mathbf{A}$, and for any $r\times L$ matrix $B$ and $L\times r$ matrix $C$, we have:
\begin{equation} \hspace{-1.2cm}
\hspace{-1cm}\boxed{\begin{aligned}
\int \mathbf{D}(\boldsymbol{\chi},\bar{\boldsymbol{\chi}}) \bigg(\prod_{\alpha=1}^{r} (\bar{{\chi}}C)_{\alpha}  (B{\chi})_{\alpha}\bigg) &e^{\bar{\boldsymbol{\chi}}^T \mathbf{A}\boldsymbol{\chi}+\bar{\boldsymbol{\psi}}^T\boldsymbol{\chi}+\bar{\boldsymbol{\chi}}^T \boldsymbol{\psi}+\bar{\mathbf{u}}^T\boldsymbol{\chi}+\bar{\boldsymbol{\chi}}^T \mathbf{u}}=\sum_{i_1<\dots<i_{r}=1}^{L} \sum_{j_1<\dots<j_{r}=1}^{L} (\det B_{[\star,{J}_{[r]}]}) \epsilon(I_{[r]},J_{[r]})\det[\mathbf{A}_{[I_{[r]}^c,J_{[r]}^c]}] \\& \times e^{-\bar{\boldsymbol{\psi}}_{[J_{[r]}^c]}^T \mathbf{A}^{-1}_{[I_{[r]}^c,J_{[r]}^c]}\boldsymbol{\psi}_{[I_{[r]}^c]}-\bar{\bold{u}}_{[J_{[r]}^c]}^T \mathbf{A}^{-1}_{[I_{[r]}^c,J_{[r]}^c]}\boldsymbol{\psi}_{[I_{[r]}^c]}-\bar{\boldsymbol{\psi}}_{[J_{[r]}^c]}^T \mathbf{A}^{-1}_{[I_{[r]}^c,J_{[r]}^c]}\bold{u}_{[I_{[r]}^c]}} (\det C_{[{I}_{[r]} ,\star]}).
\end{aligned}}
\end{equation}
\end{corollary}

\subsubsection{Fermionic Determinant with fermionic sources}
\begin{corollary}
    (Fermionic Determinant with fermionic sources).
    
\medskip
\noindent
\textbf{(a)} For an $L\times L$ matrix $\mathbf{A}$, the following relations hold:
 \begin{equation}
\boxed{\begin{aligned}
&\int \mathbf{D}(\boldsymbol{\chi},\bar{\boldsymbol{\chi}}) e^{\bar{\boldsymbol{\chi}}^T \mathbf{A}\boldsymbol{\chi}+\bar{\boldsymbol{\psi}}^T\boldsymbol{\chi}+\bar{\boldsymbol{\chi}}^T \boldsymbol{\psi}}\\&=\det[\mathbf{A}]+
   \sum_{m= 1}^L
   \frac{(-1)^{m}}{(m!)^2}
   \sum_{j_{1},\dots,j_{m}=1}^{L}
   \sum_{i_{1},\dots,i_{m}=1}^{L}
   \epsilon(I_{[m]},J_{[m]}) \det[\mathbf{A}_{[I_{[m]}^c,J_{[m]}^c]}]
     \prod_{\alpha=1}^{m} \bar{\psi}_{j_\alpha}
     \prod_{\beta=1}^{m}
     \psi_{i_\beta}.
\end{aligned}}
\end{equation}

\medskip
\noindent
\textbf{(a$^{\prime}$)} For invertible $L\times L$ matrix $\mathbf{A}$, we have:
 \begin{equation}
\boxed{\begin{aligned}
\int \mathbf{D}(\boldsymbol{\chi},\bar{\boldsymbol{\chi}}) e^{\bar{\boldsymbol{\chi}}^T \mathbf{A}\boldsymbol{\chi}+\bar{\boldsymbol{\psi}}^T\boldsymbol{\chi}+\bar{\boldsymbol{\chi}}^T \boldsymbol{\psi}}=\det[\mathbf{A}]
   e^{-\bar{\boldsymbol{\psi}}^T \mathbf{A}^{-1}\boldsymbol{\psi}}.
\end{aligned}}
\end{equation}

\medskip
\noindent
\textbf{(b)} For any subsets ${I}=\{i_1,...,i_r\}$ and ${J}=\{j_1,...,j_r\} \subseteq [L]$ having the same cardinality $r$, with $i_1<...<i_r$ and $j_1<...<j_r$, we have:
 \begin{equation}
\hspace{-2cm}\boxed{\begin{aligned}
&\int \mathbf{D}(\boldsymbol{\chi},\bar{\boldsymbol{\chi}}) \bigg(\prod_{\alpha=1}^{r} \bar{{\chi}} _{i_{\alpha}} {\chi}_{j_{\alpha}}\bigg)e^{\bar{\boldsymbol{\chi}}^T \mathbf{A}\boldsymbol{\chi}+\bar{\boldsymbol{\psi}}^T\boldsymbol{\chi}+\bar{\boldsymbol{\chi}}^T \boldsymbol{\psi}}=\epsilon(I_{[r]},J_{[r]})\det[\mathbf{A}_{[I_{[r]}^c,J_{[r]}^c]}]\\[6pt] &+
   \sum_{m= 1}^{L-r}
   \frac{(-1)^{m}(-1)^{r(L-r)}}{(m!)^2}
   \sum_{k_{1},\dots,k_{m}=1}^{L}
   \sum_{\ell_{1},\dots,\ell_{m}=1}^{L}
   \epsilon(I_{[r]}\cup \mathscr{L}_{[m]},J_{[r]}\cup K_{[m]}) \det[\mathbf{A}_{[I_{[r]}^c\cap \mathscr{L}_{[m]}^c ,J_{[r]}^c\cap K_{[m]}^c]}]
     \prod_{\alpha=1}^{m} \bar{\psi}_{k_\alpha}
     \prod_{\beta=1}^{m}
     \psi_{\ell_\beta}.
\end{aligned}}
\end{equation}

\medskip
\noindent
\textbf{(b$^{\prime}$)} For invertible $L \times L$ matrix $\mathbf{A}$ and for any subsets ${I}=\{i_1,...,i_r\}$ and ${J}=\{j_1,...,j_r\} \subseteq [L]$ having the same cardinality $r$, with $i_1<...<i_r$ and $j_1<...<j_r$, we have:
 \begin{equation}
\boxed{\begin{aligned}
&\int \mathbf{D}(\boldsymbol{\chi},\bar{\boldsymbol{\chi}}) \bigg(\prod_{\alpha=1}^{r} \bar{{\chi}} _{i_{\alpha}} {\chi}_{j_{\alpha}}\bigg)e^{\bar{\boldsymbol{\chi}}^T \mathbf{A}\boldsymbol{\chi}+\bar{\boldsymbol{\psi}}^T\boldsymbol{\chi}+\bar{\boldsymbol{\chi}}^T \boldsymbol{\psi}}=\epsilon(I_{[r]},J_{[r]})\det[\mathbf{A}_{[I_{[r]}^c,J_{[r]}^c]}]e^{-\bar{\boldsymbol{\psi}}_{[J_{[r]}^c]}^T \mathbf{A}^{-1}_{[I_{[r]}^c,J_{[r]}^c]}\boldsymbol{\psi}_{[I_{[r]}^c]}}.
\end{aligned}}
\end{equation}

\medskip
\noindent
\textbf{(c)} For any $r\times L$ matrix $B$ and $L\times r$ matrix $C$, we have:
 \begin{equation}
 \hspace{-2.5cm}
\boxed{\begin{split}
&\int \mathbf{D}(\boldsymbol{\chi},\bar{\boldsymbol{\chi}}) \bigg(\prod_{\alpha=1}^{r} (\bar{{\chi}}C)_{\alpha}  (B{\chi})_{\alpha}\bigg) e^{\bar{\boldsymbol{\chi}}^T \mathbf{A}\boldsymbol{\chi}+\bar{\boldsymbol{\psi}}^T\boldsymbol{\chi}+\bar{\boldsymbol{\chi}}^T \boldsymbol{\psi}}=
   \sum_{i_1<\dots<i_{r}=1}^{L} \sum_{j_1<\dots<j_{r}=1}^{L} (\det B_{[\star,{J}_{[r]}]})  \Bigg\{ \epsilon(I_{[r]},J_{[r]})\det[\mathbf{A}_{[I_{[r]}^c,J_{[r]}^c]}]\\[6pt] &+
   \sum_{m= 1}^{L-r}
   \frac{(-1)^{m}(-1)^{r(L-r)}}{(m!)^2}
   \sum_{k_{1},\dots,k_{m}=1}^{L}
   \sum_{\ell_{1},\dots,\ell_{m}=1}^{L}
   \epsilon(I_{[r]}\cup \mathscr{L}_{[m]},J_{[r]}\cup K_{[m]}) \det[\mathbf{A}_{[I_{[r]}^c\cap \mathscr{L}_{[m]}^c ,J_{[r]}^c\cap K_{[m]}^c]}]
     \prod_{\alpha=1}^{m} \bar{\psi}_{k_\alpha}
     \prod_{\beta=1}^{m}
     \psi_{\ell_\beta} \Bigg\} (\det C_{[{I}_{[r]} ,\star]}).
\end{split}}
\end{equation}

\medskip
\noindent
\textbf{(c$^{\prime}$)} For invertible $L \times L$ matrix $\mathbf{A}$ and for any $r\times L$ matrix $B$ and $L\times r$ matrix $C$, we have:
 \begin{equation}
\hspace{-2cm}\boxed{\begin{split}
\int \mathbf{D}(\boldsymbol{\chi},\bar{\boldsymbol{\chi}}) \bigg(\prod_{\alpha=1}^{r} (\bar{{\chi}}C)_{\alpha}  (B{\chi})_{\alpha}\bigg) e^{\bar{\boldsymbol{\chi}}^T \mathbf{A}\boldsymbol{\chi}+\bar{\boldsymbol{\psi}}^T\boldsymbol{\chi}+\bar{\boldsymbol{\chi}}^T \boldsymbol{\psi}}&=
   \sum_{i_1<\dots<i_{r}=1}^{L} \sum_{j_1<\dots<j_{r}=1}^{L} (\det B_{[\star,{J}_{[r]}]})  \epsilon(I_{[r]},J_{[r]})\det[\mathbf{A}_{[I_{[r]}^c,J_{[r]}^c]}]\\& \times e^{-\bar{\boldsymbol{\psi}}_{[J_{[r]}^c]}^T \mathbf{A}^{-1}_{[I_{[r]}^c,J_{[r]}^c]}\boldsymbol{\psi}_{[I_{[r]}^c]}} (\det C_{[{I}_{[r]} ,\star]}).
\end{split}}
\end{equation}

\end{corollary}

\subsubsection{Fermionic Determinant with bosonic sources}
\begin{corollary}(Fermionic Determinant with bosonic sources).

\medskip
\noindent
\textbf{(a)} For an $L\times L$ matrix $\mathbf{A}$, the following relations hold:
 \begin{equation}
\boxed{\begin{aligned}
&\int \mathbf{D}(\boldsymbol{\chi},\bar{\boldsymbol{\chi}}) e^{\bar{\boldsymbol{\chi}}^T \mathbf{A}\boldsymbol{\chi}+\bar{\mathbf{u}}^T\boldsymbol{\chi}+\bar{\boldsymbol{\chi}}^T \mathbf{u}}=\det(\mathbf{A}).
\end{aligned}}
\end{equation}

\medskip
\noindent
\textbf{(b)} For any subsets ${I}=\{i_1,...,i_r\}$ and ${J}=\{j_1,...,j_r\} \subseteq [L]$ having the same cardinality $r$, with $i_1<...<i_r$ and $j_1<...<j_r$, we have:
 \begin{equation}
\boxed{\begin{aligned}
&\int \mathbf{D}(\boldsymbol{\chi},\bar{\boldsymbol{\chi}}) \bigg(\prod_{\alpha=1}^{r}\bar{{\chi}} _{i_{\alpha}} {\chi}_{j_{\alpha}}\bigg)e^{\bar{\boldsymbol{\chi}}^T \mathbf{A}\boldsymbol{\chi}+\bar{\mathbf{u}}^T\boldsymbol{\chi}+\bar{\boldsymbol{\chi}}^T \mathbf{u}}=\epsilon(I_{[r]},J_{[r]})\det[\mathbf{A}_{[I_{[r]}^c,J_{[r]}^c]}].
\end{aligned}}
\end{equation}

\medskip
\noindent
\textbf{(c)} For any $r\times L$ matrix $B$ and $L\times r$ matrix $C$, we have:
 \begin{equation}
 \hspace{-2.4cm}
\boxed{\begin{split}
\int \mathbf{D}(\boldsymbol{\chi},\bar{\boldsymbol{\chi}}) \bigg(\prod_{\alpha=1}^{r} (\bar{{\chi}}C)_{\alpha}  (B{\chi})_{\alpha}\bigg) e^{\bar{\boldsymbol{\chi}}^T \mathbf{A}\boldsymbol{\chi}+\bar{\mathbf{u}}^T\boldsymbol{\chi}+\bar{\boldsymbol{\chi}}^T \mathbf{u}}=
   \sum_{i_1<\dots<i_{r}=1}^{L} \sum_{j_1<\dots<j_{r}=1}^{L}\epsilon(I_{[r]},J_{[r]}) (\det B_{[\star,{J}_{[r]}]})  \det[\mathbf{A}_{[I_{[r]}^c,J_{[r]}^c]}] (\det C_{[{I}_{[r]} ,\star]}).
\end{split}}
\end{equation}

\end{corollary}
\subsubsection{Fermionic determinant with fermionic sources and separated bosonic sources.}

In this part, we establish the basic identities involving linear fermionic exponentials, which will serve as the foundation for the construction of fermionic determinants coupled to both fermionic and bosonic sources. The following relations are straightforward consequences of the nilpotency of Grassmann variables:
\begin{subequations}
    \begin{align}
        e^{\bar{\boldsymbol{\chi}}^T \mathbf{u}} e^{\bar{\mathbf{u}}^T\boldsymbol{\chi}}
        &= e^{\bar{\boldsymbol{\chi}}^T \mathbf{u}+\bar{\mathbf{u}}^T\boldsymbol{\chi}}
        +\big(\bar{\boldsymbol{\chi}}^T \mathbf{u}\big)\big(\bar{\mathbf{u}}^T\boldsymbol{\chi}\big), \\
        e^{\bar{\mathbf{u}}^T\boldsymbol{\chi}} e^{\bar{\boldsymbol{\chi}}^T \mathbf{u}}
        &= e^{\bar{\boldsymbol{\chi}}^T \mathbf{u}+\bar{\mathbf{u}}^T\boldsymbol{\chi}}
        -\big(\bar{\boldsymbol{\chi}}^T \mathbf{u}\big)\big(\bar{\mathbf{u}}^T\boldsymbol{\chi}\big), \\
        \tfrac{1}{2}\left\{ e^{\bar{\boldsymbol{\chi}}^T \mathbf{u}}, \, e^{\bar{\mathbf{u}}^T\boldsymbol{\chi}} \right\}
        &= e^{\bar{\boldsymbol{\chi}}^T \mathbf{u}+\bar{\mathbf{u}}^T\boldsymbol{\chi}}, \\
        \tfrac{1}{2}\left[ e^{\bar{\boldsymbol{\chi}}^T \mathbf{u}}, \, e^{\bar{\mathbf{u}}^T\boldsymbol{\chi}} \right]
        &= \big(\bar{\boldsymbol{\chi}}^T \mathbf{u}\big)\big(\bar{\mathbf{u}}^T\boldsymbol{\chi}\big).
    \end{align}
\end{subequations}

\begin{corollary}[Fermionic determinant with fermionic sources and separated bosonic sources]\label{Corollary2.6}
Let $\chi^1,\dots,\chi^N$, $\bar{\chi}^1,\dots,\bar{\chi}^N$, and $\psi^1,\dots,\psi^N$, $\bar{\psi}^1,\dots,\bar{\psi}^N$ denote complex fermionic variables, and let $\bar{u}_1,\dots,\bar{u}_L,u_1,\dots,u_L$ denote real (or complex) bosonic variables.

\medskip
\noindent
\textbf{(a)} For an $L\times L$ matrix $\mathbf{A}$, the following integral identities hold:
\begin{equation}
    \begin{split}
    \int \mathbf{D}(\boldsymbol{\chi},\bar{\boldsymbol{\chi}})\,
    e^{\bar{\boldsymbol{\chi}}^T \mathbf{A}\boldsymbol{\chi}+\bar{\boldsymbol{\psi}}^T\boldsymbol{\chi}+\bar{\boldsymbol{\chi}}^T \boldsymbol{\psi}}
    e^{\bar{\boldsymbol{\chi}}^T \mathbf{u}} e^{\bar{\mathbf{u}}^T\boldsymbol{\chi}}
    &= \int \mathbf{D}(\boldsymbol{\chi},\bar{\boldsymbol{\chi}})\,
    e^{\bar{\boldsymbol{\chi}}^T \mathbf{A}\boldsymbol{\chi}
    +\bar{\boldsymbol{\psi}}^T\boldsymbol{\chi}
    +\bar{\boldsymbol{\chi}}^T \boldsymbol{\psi}
    +\bar{\mathbf{u}}^T\boldsymbol{\chi}
    +\bar{\boldsymbol{\chi}}^T \mathbf{u}} \\
    &\quad+\int \mathbf{D}(\boldsymbol{\chi},\bar{\boldsymbol{\chi}})\,
    \big(\bar{\boldsymbol{\chi}}^T \mathbf{u}\big)\,
    \big(\bar{\mathbf{u}}^T\boldsymbol{\chi}\big)\,
    e^{\bar{\boldsymbol{\chi}}^T \mathbf{A}\boldsymbol{\chi}
    +\bar{\boldsymbol{\psi}}^T\boldsymbol{\chi}
    +\bar{\boldsymbol{\chi}}^T \boldsymbol{\psi}},
    \end{split}
\end{equation}
and
\begin{equation}
    \begin{split}
    \int \mathbf{D}(\boldsymbol{\chi},\bar{\boldsymbol{\chi}})\,
    e^{\bar{\boldsymbol{\chi}}^T \mathbf{A}\boldsymbol{\chi}+\bar{\boldsymbol{\psi}}^T\boldsymbol{\chi}+\bar{\boldsymbol{\chi}}^T \boldsymbol{\psi}}
    e^{\bar{\mathbf{u}}^T\boldsymbol{\chi}} e^{\bar{\boldsymbol{\chi}}^T \mathbf{u}}
    &= \int \mathbf{D}(\boldsymbol{\chi},\bar{\boldsymbol{\chi}})\,
    e^{\bar{\boldsymbol{\chi}}^T \mathbf{A}\boldsymbol{\chi}
    +\bar{\boldsymbol{\psi}}^T\boldsymbol{\chi}
    +\bar{\boldsymbol{\chi}}^T \boldsymbol{\psi}
    +\bar{\mathbf{u}}^T\boldsymbol{\chi}
    +\bar{\boldsymbol{\chi}}^T \mathbf{u}} \\
    &\quad-\int \mathbf{D}(\boldsymbol{\chi},\bar{\boldsymbol{\chi}})\,
    \big(\bar{\boldsymbol{\chi}}^T \mathbf{u}\big)\,
    \big(\bar{\mathbf{u}}^T\boldsymbol{\chi}\big)\,
    e^{\bar{\boldsymbol{\chi}}^T \mathbf{A}\boldsymbol{\chi}
    +\bar{\boldsymbol{\psi}}^T\boldsymbol{\chi}
    +\bar{\boldsymbol{\chi}}^T \boldsymbol{\psi}}.
    \end{split}
\end{equation}

\noindent
In both expressions, the first term on the right-hand side corresponds to the extension of Theorem~\ref{theorem 2.2}, while the second term accounts for the bilinear contribution of the bosonic sources. The latter is given explicitly by
\begin{equation}
    \begin{split}
    \int \mathbf{D}(\boldsymbol{\chi},\bar{\boldsymbol{\chi}})\,
    \big(\bar{\boldsymbol{\chi}}^T \mathbf{u}\big)\big(\bar{\mathbf{u}}^T\boldsymbol{\chi}\big)\,
    e^{\bar{\boldsymbol{\chi}}^T \mathbf{A}\boldsymbol{\chi}
    +\bar{\boldsymbol{\psi}}^T\boldsymbol{\chi}
    +\bar{\boldsymbol{\chi}}^T \boldsymbol{\psi}}
    &= \sum_{i_0,j_0=1}^{L} \epsilon(I_{[0]},J_{[0]})\,\bar{u}_{j_0}u_{i_0}\,
       \det\!\big(\mathbf{A}_{[I_{[0]}^c,J_{[0]}^c]}\big) \\
    &\quad+\sum_{m=1}^{L-1}\frac{(-1)^m(-1)^{L-1}}{(m!)^2}
      \sum_{i_0,i_1,\dots,i_m=1}^L \sum_{j_0,j_1,\dots,j_m=1}^L
      \epsilon\!\left(I_{[0,m]},J_{[0,m]}\right) \\
    &\qquad\times \det\!\left(\mathbf{A}_{[K_{[0,m]}^c,J_{[0,m]}^c]}\right)
      \Big(\bar{u}_{j_0}u_{i_0}\prod_{\alpha=1}^m \bar{\psi}_{j_\alpha}\psi_{i_\alpha}\Big).
    \end{split}
\end{equation}

By simplifying this expression, we obtain:
\begin{equation}
   \boxed{ \begin{split}
   & \int \mathbf{D}(\boldsymbol{\chi},\bar{\boldsymbol{\chi}})\,
    \big(\bar{\boldsymbol{\chi}}^T \mathbf{u}\big)\big(\bar{\mathbf{u}}^T\boldsymbol{\chi}\big)\,
    e^{\bar{\boldsymbol{\chi}}^T \mathbf{A}\boldsymbol{\chi}
    +\bar{\boldsymbol{\psi}}^T\boldsymbol{\chi}
    +\bar{\boldsymbol{\chi}}^T \boldsymbol{\psi}}
    \\&= \det\!\big(\bar{\mathbf{X}}\big) \quad+\sum_{m=1}^{L-1}\frac{(-1)^m(-1)^{L-1}}{(m!)^2}
      \sum_{i_1,\dots,i_m=1}^L \sum_{j_1,\dots,j_m=1}^L
      \epsilon\!\left(I_{[m]},J_{[m]}\right)  \det\!\left(\bar{\mathbf{X}}_{[K_{[m]}^c,J_{[m]}^c]}\right)
      \Big(\prod_{\alpha=1}^m \bar{\psi}_{j_\alpha}\psi_{i_\alpha}\Big),
    \end{split}}
\end{equation}
where 
\begin{equation}
   \bar{\mathbf{X}}=\begin{pmatrix} 0 & \bar{\mathbf{u}}^T \\ \mathbf{u} & \mathbf{A}^T \end{pmatrix}.
\end{equation}
For the invertible matrix $\bold{A}$, this formula can be written as:
\begin{equation}
   \boxed{ \begin{split}
   & \int \mathbf{D}(\boldsymbol{\chi},\bar{\boldsymbol{\chi}})\,
    \big(\bar{\boldsymbol{\chi}}^T \mathbf{u}\big)\big(\bar{\mathbf{u}}^T\boldsymbol{\chi}\big)\,
    e^{\bar{\boldsymbol{\chi}}^T \mathbf{A}\boldsymbol{\chi}
    +\bar{\boldsymbol{\psi}}^T\boldsymbol{\chi}
    +\bar{\boldsymbol{\chi}}^T \boldsymbol{\psi}}
    = \sum\limits_{i,j=1}^L (-1)^{i+j}u_i\bar{u}_j\det[\mathbf{A}_{[i^c,j^c]}]e^{-\bar{\boldsymbol{\psi}}_{[j^c]}^T \mathbf{A}^{-1}_{[i^c,j^c]}\boldsymbol{\psi}_{[i^c]}}.
    \end{split}}
\end{equation}
\end{corollary}

\begin{lemma} (Fermionic Determinant with separated Bosonic sources).
Let $\chi^1,\dots,\chi^N$, $\bar{\chi}^1,\dots,\bar{\chi}^N$ denote complex fermionic variables, and let $\bar{u}_1,\dots,\bar{u}_L,u_1,\dots,u_L$ denote real (or complex) bosonic variables.

\medskip
\noindent
\textbf{(a)} For an $L\times L$ matrix $\mathbf{A}$, the following relations hold:
\begin{equation}
    \begin{split}
    \int\mathbf{D}(\boldsymbol{\chi},\bar{\boldsymbol{\chi}})\,
    e^{\bar{\boldsymbol{\chi}}^T \mathbf{A}\boldsymbol{\chi}}
    e^{\bar{\boldsymbol{\chi}}^T \mathbf{u}} e^{\bar{\mathbf{u}}^T\boldsymbol{\chi}}
    &=\int\mathbf{D}(\boldsymbol{\chi},\bar{\boldsymbol{\chi}})\,
    e^{\bar{\boldsymbol{\chi}}^T\mathbf{A}\boldsymbol{\chi}
    +\bar{\mathbf{u}}^T\boldsymbol{\chi}
    +\bar{\boldsymbol{\chi}}^T\mathbf{u}}+\int\mathbf{D}(\boldsymbol{\chi},\bar{\boldsymbol{\chi}})\,
    \big(\bar{\boldsymbol{\chi}}^T \mathbf{u}\big)\,
    \big(\bar{\mathbf{u}}^T\boldsymbol{\chi}\big)\,
    e^{\bar{\boldsymbol{\chi}}^T\mathbf{A}\boldsymbol{\chi}},
    \end{split}
\end{equation}
and similarly
\begin{equation}
    \begin{split}
    \int\mathbf{D}(\boldsymbol{\chi},\bar{\boldsymbol{\chi}})\,
    e^{\bar{\boldsymbol{\chi}}^T \mathbf{A}\boldsymbol{\chi}}
    e^{\bar{\mathbf{u}}^T\boldsymbol{\chi}}e^{\bar{\boldsymbol{\chi}}^T \mathbf{u}}
    &=\int\mathbf{D}(\boldsymbol{\chi},\bar{\boldsymbol{\chi}})\,
    e^{\bar{\boldsymbol{\chi}}^T\mathbf{A}\boldsymbol{\chi}
    +\bar{\mathbf{u}}^T\boldsymbol{\chi}
    +\bar{\boldsymbol{\chi}}^T\mathbf{u}}-\int\mathbf{D}(\boldsymbol{\chi},\bar{\boldsymbol{\chi}})\,
    \big(\bar{\boldsymbol{\chi}}^T \mathbf{u}\big)\,
    \big(\bar{\mathbf{u}}^T\boldsymbol{\chi}\big)\,
    e^{\bar{\boldsymbol{\chi}}^T\mathbf{A}\boldsymbol{\chi}}.
    \end{split}
\end{equation}

\noindent
The second term on the right-hand side of each expression is given explicitly by:  

\begin{equation}
    \begin{split}
    \int\mathbf{D}&(\boldsymbol{\chi},\bar{\boldsymbol{\chi}})\big(\bar{\boldsymbol{\chi}}^T \mathbf{u}\big)\big(\bar{\mathbf{u}}^T\boldsymbol{\chi}\big)
    e^{\bar{\boldsymbol{\chi}}^T \mathbf{A}\boldsymbol{\chi}}=\sum_{i_0,j_0=1}^{L}\epsilon(I_{[0]},J_{[0]})\bar{u}_{j_0}u_{k_0}\det\left(\mathbf{A}_{[I_{[0]}^c,J_{[0]}^c]}\right)=\det\left(\bar{\mathbf{X}}\right).
    \end{split}
\end{equation}
\end{lemma}

\section*{Conclusions}

In this work, we have developed a comprehensive Grassmann--Berezin framework that establishes fundamental connections between combinatorial enumeration, statistical mechanics, and functional integration methods. Our main contributions can be summarized as follows:

\paragraph*{Unified Mathematical Framework}
We established two master Berezin integral over Grassmann variables identities (Theorems~\ref{exponential with linear fermionic and bosonic sources} and \ref{theorem 2.2}) that incorporate both bosonic and fermionic source terms for real and complex fermions. These identities remain valid in arbitrary dimensions and for singular matrices, generalizing the standard Gaussian integral framework. A key technical innovation is the practical block-decomposition method for handling non-invertible matrices, which isolates null and non-null subspaces while preserving the Pfaffian structure.

\paragraph*{Mapping Between Complex and Real Fermions}
In this work, we have established a comprehensive translation of combinatorial structures into the framework of Berezin integrals over Grassmann variables. Our central contributions include:

\begin{itemize}
\item The Hafnian (naturally expressed with complex fermions) is mapped to the Pfaffian (real fermions) through Grassmann--Berezin representations (Theorem~\ref{thm:main}).
\item Similarly, the Monobisyzexant (Mbsz) function for complex fermions admits an equivalent representation in terms of real fermions (Theorem~\ref{thm:main2}).
\item These mappings establish a fundamental connection between the Hafnian and Pfaffian formulations of the dimer model through Kasteleyn orientation, while simultaneously providing a real fermionic integral representation of the Mbsz. Crucially, our analysis demonstrates that the Mbsz admits no Pfaffian representation, highlighting the essential complexity of the monomer-dimer problem compared to the pure dimer case.
\end{itemize}
\begin{itemize}
\item The Hafnian (naturally expressed with complex fermions) is mapped to the Pfaffian (real fermions) through Grassmann--Berezin representations (Theorem~\ref{thm:main}).
\item Similarly, the Monobisyzexant (Mbsz) function for complex fermions admits an equivalent representation in terms of real fermions (Theorem~\ref{thm:main2}).
\item These mappings establish a fundamental connection between the Hafnian and Pfaffian formulations of the dimer model through Kasteleyn orientation, while simultaneously providing a real fermionic integral representation of the Mbsz. Crucially, our analysis demonstrates that the Mbsz admits no Pfaffian representation, highlighting the essential complexity of the monomer-dimer problem compared to the pure dimer case.
\end{itemize}

\paragraph*{Combinatorial Applications}
Our framework yields advances when applied to fundamental combinatorial models:
\begin{itemize}
\item A generalized partition function for planar dimer systems of arbitrary parity, incorporating fixed monomers through bosonic source terms and expressible via Pfaffians of extended Kasteleyn matrices
\item The introduction of the Monobiszczant (Mbsz) function, which generalizes the Hafnian to provide complete combinatorial encoding of monomer-dimer configurations, with potential extensions to more complex structures like Hyperhafnians through systematic modifications of the $\mathcal{S}_{i_1 \cdots i_L}(\mathbf{D},\mathbf{V})$ tensor~\eqref{equ_Stensor}
\item An alternative Berezin integral representation for spanning trees and forests utilizing complex bosonic sources, where these gauge-like terms fundamentally transform the correlation structure to reveal spanning forests rather than pure spanning trees
\end{itemize}

The Grassmann-Berezin framework developed in this work provides powerful mathematical tools for combinatorial enumeration and functional integration in statistical mechanics. The master identities and block-decomposition methods extend Gaussian integral techniques to singular matrices, while mappings between fermion representations offer alternative formulations of combinatorial problems. 

Furthermore, the fermionic representation employed here enables the application of advanced techniques such as fermionic duality, which yields improved approximations for complex systems. Notably, this approach is not limited to quartic interactions; ongoing research demonstrates its extensibility to interactions of arbitrary order.

This work potentially contributes to the mathematical foundations of lattice field theory and statistical mechanics.

\section*{Acknowledgements}
MAR acknowledges partial support by CNPq. Eddy acknowledges support from CNPq (Process No. 141672/2023-4).

\newpage

\appendix
\section*{\centering \Huge \textbf{Appendix}}
\addcontentsline{toc}{section}{Appendix}
\vspace*{2em}  
\addtocontents{toc}{\protect\setcounter{tocdepth}{1}}
\renewcommand{\theequation}{\thesection.\arabic{equation}}

\setcounter{equation}{0}
\setcounter{table}{0}
\renewcommand{\thetable}{A\arabic{table}}

\section{Graph Matrices: Adjacency and Laplacian}\label{Laplacianmatrix}

This section introduces the fundamental matrices associated with graphs: the adjacency matrix, the degree matrix, and the Laplacian matrix. These matrices encode the connectivity structure of graphs and play crucial roles in various physical and combinatorial contexts.

\subsection{Adjacency Matrix}

Let $G = (V,E)$ be a graph with vertex set $V = \{1,\dots,L\}$ and edge set $E$. The \textit{weighted adjacency matrix} $\mathbf{W}$ is defined by
\begin{equation}
W_{ij} = \text{weight of the edge between } i \text{ and } j.
\end{equation}
For undirected graphs, $\mathbf{W}$ is symmetric: $W_{ij} = W_{ji}$. If there is no edge between $i$ and $j$ (and $i \neq j$), then $W_{ij} = 0$.

\subsection{Degree Matrix}

The \textit{degree matrix} $\mathbf{D}$ is a diagonal matrix defined by
\begin{equation}
D_{ii} = \sum_{j} W_{ij},
\end{equation}
which represents the total weight of edges incident to vertex $i$.

\subsection{Laplacian Matrix}

The \textit{Laplacian matrix} $\mathbf{L}$ is defined as
\begin{equation}
\mathbf{L} = \mathbf{D} - \mathbf{W}.
\end{equation}
In components,
\begin{equation}
L_{ij} =
\begin{cases}
\displaystyle \sum_{k \neq i} W_{ik}, & \text{if } i=j, \\[6pt]
-W_{ij}, & \text{if } i \neq j \text{ and } (i,j) \in E, \\[4pt]
0, & \text{otherwise.}
\end{cases}
\end{equation}
The Laplacian matrix is symmetric and positive semidefinite. Each row and column sums to zero, so $\mathbf{L}$ has a zero eigenvalue corresponding to the eigenvector $(1,1,\dots,1)^T$.

\subsection{Laplacian for Graphs with Loops}

For graphs containing loops (edges from a vertex to itself), the definition of the Laplacian requires careful consideration. Let $w_{ii}$ denote the weight of a loop at vertex $i$. The adjacency matrix entry $W_{ii}$ equals the loop weight $w_{ii}$, and the degree $D_{ii}$ includes this loop weight. However, in the Laplacian matrix, the diagonal entry becomes:
\begin{equation}
L_{ii} = D_{ii} - W_{ii} = \left(\sum_{k \neq i} W_{ik} + W_{ii}\right) - W_{ii} = \sum_{k \neq i} W_{ik}.
\end{equation}
Thus, loops do not appear explicitly in the Laplacian matrix, though they contribute to the vertex degree in the adjacency structure.

\subsection{Physical Interpretation}

In physical contexts, the Laplacian matrix acts as a discrete version of the Laplace operator. For a function $\phi$ defined on the vertices,
\begin{equation}
(\mathbf{L} \, \phi)_i = \sum_{j} W_{ij} (\phi_i - \phi_j),
\end{equation}
which measures the difference between $\phi_i$ and its neighbors. This interpretation remains valid for graphs with loops, as the $j=i$ term vanishes.

\subsection{Construction Rules}

The Laplacian can be constructed by the following steps:

\begin{enumerate}
    \item \textbf{List all vertices.} Label the vertices as $1, \dots, L$.
    \item \textbf{Compute degrees.} For each vertex $i$, compute $D_{ii} = \sum_{j} W_{ij}$ (including loops).
    \item \textbf{Assign off-diagonal entries.} For each pair $(i,j)$ with $i \neq j$:
    \begin{itemize}
        \item If $i$ and $j$ are connected by an edge with weight $w_{ij}$, set $L_{ij} = -w_{ij}$.
        \item Otherwise, set $L_{ij} = 0$.
    \end{itemize}
    \item \textbf{Assign diagonal entries.} Set $L_{ii} = D_{ii} - W_{ii} = \sum_{j \neq i} W_{ij}$.
    \item \textbf{Verify the row-sum property.} Each row must sum to zero.
\end{enumerate}

\subsection{Examples}

\paragraph{(i) Graph with a Loop.}
Consider a graph with two vertices and a loop at vertex 1. Let the edges be: (1,1) with weight $a$, (1,2) with weight $b$. Then:
- Adjacency matrix:
\[
\mathbf{W} = \begin{pmatrix}
a & b \\
b & 0
\end{pmatrix}.
\]
- Degree matrix:
\[
\mathbf{D} = \begin{pmatrix}
a+b & 0 \\
0 & b
\end{pmatrix}.
\]
- Laplacian matrix:
\[
\mathbf{L} = \mathbf{D} - \mathbf{W} = \begin{pmatrix}
a+b - a & -b \\
-b & b - 0
\end{pmatrix} = \begin{pmatrix}
b & -b \\
-b & b
\end{pmatrix}.
\]
Note that the loop does not appear in the Laplacian.

\paragraph{(ii) Triangle Graph with a Loop.}
For three vertices connected pairwise with unit weights and a loop at vertex 1 with weight $c$:
- Adjacency matrix:
\[
\mathbf{W} = \begin{pmatrix}
c & 1 & 1 \\
1 & 0 & 1 \\
1 & 1 & 0
\end{pmatrix}.
\]
- Degree matrix:
\[
\mathbf{D} = \begin{pmatrix}
c+2 & 0 & 0 \\
0 & 2 & 0 \\
0 & 0 & 2
\end{pmatrix}.
\]
- Laplacian matrix:
\[
\mathbf{L} = \begin{pmatrix}
(c+2)-c & -1 & -1 \\
-1 & 2-0 & -1 \\
-1 & -1 & 2-0
\end{pmatrix} = \begin{pmatrix}
2 & -1 & -1 \\
-1 & 2 & -1 \\
-1 & -1 & 2
\end{pmatrix}.
\]
Again, the loop does not affect the Laplacian.

\subsection{Spanning Tree Interpretation}

Kirchhoff's matrix-tree theorem states that for a connected graph, the number of spanning trees is given by
\begin{equation}
\tau(G) = \det \mathbf{L}_{(i)} ,
\end{equation}
where $\mathbf{L}_{(i)}$ is the matrix obtained by deleting the $i$-th row and column of $\mathbf{L}$. This result extends to weighted graphs, giving the sum of the weights of all spanning trees. For graphs with loops, the theorem remains valid as loops are never included in spanning trees.

\subsection{Comparison with the Kasteleyn Matrix}

Both the Laplacian and Kasteleyn matrices encode topological information of graphs, but they emphasize different physical structures:
\begin{itemize}
    \item The Laplacian arises naturally from \textit{bosonic} fields and diffusion-type dynamics.
    \item The Kasteleyn matrix originates from \textit{fermionic} degrees of freedom and Pfaffian combinatorics.
\end{itemize}
Both matrices share local rule-based constructions and yield determinant or Pfaffian quantities that count fundamental combinatorial configurations.

\section{Properties of the Pfaffian}\label{PropPf}
Some basic properties of Pfaffians are shown below:
\begin{lemma}[Properties of the Pfaffian]\label{LemmaPf}
Let $\mathbf{A}$ be a real skew-symmetric $2N\times 2N$ matrix. Then: 
\begin{enumerate}
\item The Pfaffian of a tridiagonal matrix is
    \begin{equation}
        \pf\begin{pmatrix}
            0&a_{1}&0&0&&&\\
            -a_{1}&0&0&0&&&\\
            0&0&0&a_{2}&&&\\
            0&0&-a_{2}&0&\ddots&&\\
            &&&\ddots&\ddots&&\\
            &&&&&0&a_n\\
            &&&&&-a_n&0
    \end{pmatrix}=\prod_{i=1}^n a_i,
    \end{equation}
    if $a_1=\cdots=a_n=1$ we have
    \begin{equation}
        \pf\begin{pmatrix}
            0&1&0&0&&&\\
            -1&0&0&0&&&\\
            0&0&0&1&&&\\
            0&0&-1&0&\ddots&&\\
            &&&\ddots&\ddots&&\\
            &&&&&0&1\\
            &&&&&-1&0
    \end{pmatrix}=1.
    \end{equation}
    \item $\left[\pf(\mathbf{A})\right]^2=\det(\mathbf{A})$.
    \item $\pf\left(\mathbf{M}\mathbf{A}\mathbf{M}^T\right)=\det(\mathbf{M})\pf(\mathbf{A})$ for any $2N\times2N$ matrix $\mathbf{M}$.
    \item \textsc{(Minor summation formula for Pfaffians)} More generally, we have
    \begin{equation}
    \pf\left(\mathbf{M}\mathbf{A}\mathbf{M}^T\right)=\sum_{\substack{I^c\subseteq [2N]\\|I^c|=2m}}\det(\mathbf{M}_{[\star|I^c]})\pf\left(\mathbf{A}_{[I^c]}\right),\quad for~any~2m\times 2N~matrix~\mathbf{M}~(m\leq N),
    \end{equation}
    where $[2N]=\{1,\dots,2N\}$. Here $\mathbf{M}_{[\star|I^c]}$ denotes the subatrix of $\mathbf{M}$ with columns $I^c$ (and all its rows).
    \item \textsc{(Jacobi's identity for pfaffians)} If $\mathbf{A}$ is invertible, then $\pf\left(\left(\mathbf{A}^{-1}\right)^T\right)=\left[\pf(\mathbf{A})\right]^{-1}$ and more generally
    \begin{equation}
        \pf\left(\left(\mathbf{A}^{-1}\right)^T_{[I]}\right)=\frac{\epsilon(I)\pf\left(\mathbf{A}_{[I^c]}\right)}{\pf(\mathbf{A})},
    \end{equation}
    for any $I\subset [2N]$, where $\epsilon(I)=(-1)^{|I|\frac{|I|-1}{2}}(-1)^{\sum_{i\in I}i}$.
\end{enumerate}
\end{lemma}
The proof of some of the properties above can be found here \cite{lomont1985properties}.

\section{Spectral Illustration: Disconnected Graph Examples}\label{sec:SpectralIllustration}

To complement Section~\ref{subsec:spectral_forests}, we present two examples showing how the spectral
decomposition of the Laplacian identifies disconnected components and factorizes the enumeration of
spanning forests. The procedure is entirely algebraic and applies equally well to graphs with any number
of components.

\subsection*{Example 1: Two Components}

Consider a weighted graph $G$ with three vertices $V = \{1,2,3\}$ and edges
\(
E = \{(1,2)\}
\),
with uniform weight $w_{12}=1$. Vertex $3$ is isolated, so the graph consists of two connected
components:
\[
G_1 = (\{1,2\}, \{(1,2)\}), \qquad
G_2 = (\{3\}, \emptyset).
\]

\paragraph{Step 1. Laplacian and spectrum.}
The weighted Laplacian is
\[
L =
\begin{pmatrix}
 1 & -1 & 0 \\[3pt]
-1 &  1 & 0 \\[3pt]
 0 &  0 & 0
\end{pmatrix}.
\]
Its eigenvalues and normalized eigenvectors are
\[
\lambda_1 = 0, \quad \lambda_2 = 0, \quad \lambda_3 = 2, \qquad
Q =
\begin{pmatrix}
\frac{1}{\sqrt{2}} & 0 & \frac{1}{\sqrt{2}}\\[3pt]
\frac{1}{\sqrt{2}} & 0 & -\frac{1}{\sqrt{2}}\\[3pt]
0 & 1 & 0
\end{pmatrix}.
\]
The two zero eigenvalues correspond to two disconnected components.

\paragraph{Step 2. Extracting zero modes.}
The matrix $Q_0$ formed by the zero-mode eigenvectors is
\[
Q_0 =
\begin{pmatrix}
\frac{1}{\sqrt{2}} & 0 \\[3pt]
\frac{1}{\sqrt{2}} & 0 \\[3pt]
0 & 1
\end{pmatrix}.
\]
Each row of $Q_0$ corresponds to a vertex:
\[
r_1 = (\tfrac{1}{\sqrt{2}}, 0), \quad
r_2 = (\tfrac{1}{\sqrt{2}}, 0), \quad
r_3 = (0, 1).
\]
Vertices $1$ and $2$ share identical rows (hence belong to the same component), while vertex $3$ has a distinct row.

\paragraph{Step 3. Block decomposition and factorization.}
Reordering vertices as $(1,2)$ and $(3)$ gives
\[
L_{\text{block}} =
\begin{pmatrix}
L_{G_1} & 0\\[3pt]
0 & L_{G_2}
\end{pmatrix}, \qquad
L_{G_1} =
\begin{pmatrix}
1 & -1\\[3pt]
-1 & 1
\end{pmatrix}, \quad
L_{G_2} = (0).
\]
Deleting one row and column from $L_{G_1}$ yields
\(
\det(L_{G_1[\{1\}^c]}) = 1,
\)
and by convention $\det(L_{G_2[\{1\}^c]}) = 1$. Thus,
\[
\tau_{\mathrm{forest}}(G) = \tau(G_1)\,\tau(G_2) = 1 \times 1 = 1.
\]

\paragraph{Interpretation.}
The spectral decomposition identifies both components purely algebraically through the structure of
$Q_0$. Vertices with identical rows in $Q_0$ form a connected region. In the Grassmann formulation,
the same factorization appears as a product of two independent Gaussian integrals
over $\chi_{1,2}$ and $\chi_3$, confirming Eq.~\eqref{eq:forest_factorization}.

\vspace{1em}

\subsection*{Example 2: Three Components}

We now illustrate a general case with three disconnected components.

\paragraph{Step 1. Graph and Laplacian.}
Let the vertex set be $V = \{1,2,3,4,5,6\}$ and edges
\[
E = \{(1,2), (2,3), (4,5)\}.
\]
Hence,
\[
G_1 = \{1,2,3\}, \qquad
G_2 = \{4,5\}, \qquad
G_3 = \{6\}.
\]
The Laplacian matrix is
\[
L =
\begin{pmatrix}
 1 & -1 &  0 &  0 &  0 &  0\\
-1 &  2 & -1 &  0 &  0 &  0\\
 0 & -1 &  1 &  0 &  0 &  0\\
 0 &  0 &  0 &  1 & -1 &  0\\
 0 &  0 &  0 & -1 &  1 &  0\\
 0 &  0 &  0 &  0 &  0 &  0
\end{pmatrix}.
\]

\paragraph{Step 2. Eigenvalues and zero modes.}
Diagonalizing $L$ gives the spectrum
\[
\lambda = (0,0,0,1,2,3).
\]
The three zero modes (normalized) are
\[
v^{(1)} = \tfrac{1}{\sqrt{3}}(1,1,1,0,0,0)^{\top}, \quad
v^{(2)} = \tfrac{1}{\sqrt{2}}(0,0,0,1,1,0)^{\top}, \quad
v^{(3)} = (0,0,0,0,0,1)^{\top}.
\]
Stacking them as columns gives
\[
Q_0 =
\begin{pmatrix}
\frac{1}{\sqrt{3}} & 0 & 0\\[2pt]
\frac{1}{\sqrt{3}} & 0 & 0\\[2pt]
\frac{1}{\sqrt{3}} & 0 & 0\\[2pt]
0 & \frac{1}{\sqrt{2}} & 0\\[2pt]
0 & \frac{1}{\sqrt{2}} & 0\\[2pt]
0 & 0 & 1
\end{pmatrix}.
\]

\paragraph{Step 3. Reading connectivity from \(Q_0\).}
Each vertex corresponds to one row of $Q_0$:
\[
\begin{array}{c|ccc}
i & (Q_0)_{i1} & (Q_0)_{i2} & (Q_0)_{i3}\\ \hline
1 & 1/\sqrt{3} & 0 & 0\\
2 & 1/\sqrt{3} & 0 & 0\\
3 & 1/\sqrt{3} & 0 & 0\\
4 & 0 & 1/\sqrt{2} & 0\\
5 & 0 & 1/\sqrt{2} & 0\\
6 & 0 & 0 & 1
\end{array}
\]
Vertices $1,2,3$ have identical rows, forming $G_1$; vertices $4,5$ share the same row, forming $G_2$;
and vertex $6$ stands alone, giving $G_3$.

\paragraph{Step 4. General criterion.}
Given any Laplacian $L$ with $k$ zero modes,
\[
L = Q \Lambda Q^{\top}, \qquad \Lambda = \mathrm{diag}(0,\dots,0,\lambda_{k+1},\dots,\lambda_N),
\]
the $N\times k$ matrix $Q_0$ contains all zero-eigenvalue eigenvectors.
Defining $r_i = (Q_0)_{i,1:k}$ as the $i$-th row,
\[
i \text{ and } j \text{ belong to the same component} \iff r_i \parallel r_j.
\]
In other words, identical or proportional rows in $Q_0$ indicate connectivity.
\paragraph{Step 5. Block decomposition and factorization.}
Reordering vertices as $(1,2,3)$, $(4,5)$ and $(6)$ gives
\[
L_{\text{block}} =
\begin{pmatrix}
L_{G_1} & 0 & 0\\[3pt]
0 & L_{G_2} & 0\\[3pt]
0 & 0 & L_{G_3}
\end{pmatrix}, \qquad
L_{G_1} =
\begin{pmatrix}
1 & -1 & 0\\[3pt]
-1 & 2 &-1\\[3pt]
0 & -1 &1
\end{pmatrix}, \qquad
L_{G_1} =
\begin{pmatrix}
1 & -1\\[3pt]
-1 & 1
\end{pmatrix}, \quad
L_{G_3} = (0).
\]
Deleting one row and column from $L_{G_1}$ yields
\(
\det(L_{G_1[\{1\}^c]}) = 1
\), deleting one row and column from $L_{G_2}$ yields
\(
\det(L_{G_2[\{1\}^c]}) = 1
\),
and by convention $\det(L_{G_3[\{1\}^c]}) = 1$. Thus,
\[
\tau_{\mathrm{forest}}(G) = \tau(G_1)\,\tau(G_2) \tau(G_3) = 1.
\]

\paragraph{Step 6. Interpretation.}
In $\mathbb{R}^k$, each vertex corresponds to one point given by $r_i$.
All vertices of the same component collapse to the same point,
while distinct components occupy orthogonal directions.
The Laplacian’s kernel therefore encodes the connected components geometrically,
and the regularized subspace $L_{\mathrm{reg}} = Q_+ \Lambda_+ Q_+^{\top}$
acts only on the nontrivial connected sectors.

---

Together, these examples show concretely how the spectral decomposition not only provides
a regularization mechanism for singular Laplacians but also encodes connectivity information
in the geometry of the zero-mode subspace.

\section{Examples of Monobisyzexant}\label{Mbszexample}

\subsection{Example 1: The case of a \(5\times5\) matrix}
Consider $5\times5$ matrices, where $\mathbf{D}$ is diagonal and $\mathbf{V}$ is symmetric.  
The symmetric tensor $\mathcal{S}_{i_1 i_2 i_3 i_4 i_5}(\mathbf{D},\mathbf{V})$ is defined as
\begin{equation*}
\begin{split}
\mathcal{S}_{i_1i_2i_3i_4i_5}(\mathbf{D},\mathbf{V})=&\begin{vmatrix}
    D_{i_1i_1}&\binom{4}{1}V_{i_1i_2}&0&0&0\\
    -1& D_{i_2i_2}&\binom{3}{1}V_{i_2i_3}&0&0\\
    0&-1& D_{i_3i_3}&\binom{2}{1}V_{i_3i_4}&0\\
    0&0&-1& D_{i_4i_4}&V_{i_4i_5}\\
    0&0&0&-1& D_{i_5i_5}
\end{vmatrix}\\
=&D_{i_1i_1}D_{i_2i_2}D_{i_3i_3}D_{i_4i_4}D_{i_5i_5}+4D_{i_3i_3}D_{i_4i_4}D_{i_5i_5}V_{i_1i_2}+3D_{i_1i_1}D_{i_4i_4}D_{i_5i_5}V_{i_2i_3}\\
&+2D_{i_1i_1}D_{i_2i_2}D_{i_5i_5}V_{i_3i_4}+8D_{i_5i_5}V_{i_1i_2}V_{i_3i_4}+D_{i_1i_1}D_{i_2i_2}D_{i_3i_3}V_{i_4i_5}\\
&+4D_{i_3i_3}V_{i_1i_2}V_{i_4i_5}+3D_{i_1i_1}V_{i_2i_3}V_{i_4i_5}.
\end{split}
\end{equation*}

The corresponding Monobisyzexant is
\begin{equation*}
\hspace{-3.7cm}
\begin{split}
\operatorname{Mbsz}\left(\mathbf{D},\mathbf{V}\right)=&\frac{1}{5!}\sum_{i_1,\dots,i_5=1}^5\rho^{i_1i_2i_3i_4i_5}\mathcal{S}_{i_1i_2i_3i_4i_5}\left(\mathbf{D},\mathbf{V}\right)\\
=&D_{11}D_{22}D_{33}D_{44}D_{55}+D_{33}D_{44}D_{55}V_{12}+D_{22}D_{44}D_{55}V_{13}+D_{22}D_{33}D_{55}V_{14}+D_{22}D_{33}D_{44}V_{15}\\
&+D_{11}D_{44}D_{55}V_{23}+D_{11}D_{33}D_{55}V_{24}+D_{11}D_{33}D_{44}V_{25}+D_{11}D_{22}D_{55}V_{34}+D_{11}D_{22}D_{44}V_{35}\\
&+D_{11}D_{22}D_{33}V_{45}+D_{55}(V_{14}V_{23}+V_{13}V_{24}+V_{12}V_{34})+D_{44}(V_{15}V_{23}+V_{13}V_{25}+V_{12}V_{35})\\
&+D_{33}(V_{15}V_{24}+V_{14}V_{25}+V_{12}V_{45})+ D_{22}(V_{15}V_{34}+V_{14}V_{35}+V_{13}V_{45})\\
&+D_{11}(V_{25}V_{34}+V_{24}V_{35}+V_{23}V_{45})\\
=&\det(\mathbf{D})+\det(\mathbf{D}_{[\{1,2\}^c]})\haf\left(\mathbf{V}_{[\{3,4,5\}^c]}\right)+\det(\mathbf{D}_{[\{1,3\}^c]})\haf\left(\mathbf{V}_{[\{2,4,5\}^c]}\right)\\
&+\det(\mathbf{D}_{[\{1,4\}^c]})\haf\left(\mathbf{V}_{[\{2,3,5\}^c]}\right)+\det(\mathbf{D}_{[\{1,5\}^c]})\haf\left(\mathbf{V}_{[\{2,3,4\}^c]}\right)+\det(\mathbf{D}_{[\{2,3\}^c]})\haf\left(\mathbf{V}_{[\{1,4,5\}^c]}\right)\\
&+\det(\mathbf{D}_{[\{2,4\}^c]})\haf\left(\mathbf{V}_{[\{1,3,5\}^c]}\right)+\det(\mathbf{D}_{[\{2,5\}^c]})\haf\left(\mathbf{V}_{[\{1,3,4\}^c]}\right)+\det(\mathbf{D}_{[\{3,4\}^c]})\haf\left(\mathbf{V}_{[\{1,2,5\}^c]}\right)\\
&+\det(\mathbf{D}_{[\{3,5\}^c]})\haf\left(\mathbf{V}_{[\{1,2,4\}^c]}\right)+\det(\mathbf{D}_{[\{4,5\}^c]})\haf\left(\mathbf{V}_{[\{1,2,3\}^c]}\right)+\det(\mathbf{D}_{[\{1,2,3,4\}^c]})\haf\left(\mathbf{V}_{[\{5\}^c]}\right)\\
&+\det(\mathbf{D}_{[\{1,2,3,5\}^c]})\haf\left(\mathbf{V}_{[\{4\}^c]}\right)+\det(\mathbf{D}_{[\{1,2,4,5\}^c]})\haf\left(\mathbf{V}_{[\{3\}^c]}\right)+\det(\mathbf{D}_{[\{2,3,4,5\}^c]})\haf\left(\mathbf{V}_{[\{1\}^c]}\right)\\
=&\det(\mathbf{D})+\sum_{i_1<i_2<i_3=1}^5\det(\mathbf{D}_{[\{i_1,i_2,i_3\}]})\haf\left(\mathbf{V}_{[\{i_1,i_2,i_3\}^c]}\right)+\sum_{i_1=1}^5\det(\mathbf{D}_{[\{i_1\}]})\haf\left(\mathbf{V}_{[\{i_1\}^c]}\right)\\
=&\det(\mathbf{D})+\sum_{\alpha=1}^2\sum_{i_1<\cdots<i_{2\alpha-1}=1}^5\det(\mathbf{D}_{[\{i_1,\dots,i_{2i-1}\}]})\haf\left(\mathbf{V}_{[\{i_1,\dots,i_{2\alpha-1}\}^c]}\right).
    \end{split}
\end{equation*}

\subsection{Example 2: The case of a \(4\times4\) matrix}

Consider $4\times4$ matrices, $\mathbf{D}$ is a diagonal matrix and $\mathbf{V}$ is a symmetric matrix. The symmetric tensor $\mathcal{S}_{i_1i_2i_3i_4}(\mathbf{D},\mathbf{V})$ is
\begin{equation*}
\begin{split}
\mathcal{S}_{i_1i_2i_3i_4}(\mathbf{D},\mathbf{V})=&\begin{vmatrix}
    D_{i_1i_1}&\binom{3}{1}V_{i_1i_2}&0&0\\
    -1& D_{i_2i_2}&\binom{2}{1}V_{i_2i_3}&0\\
    0&-1& D_{i_3i_3}&V_{i_3i_4}\\
    0&0&-1& D_{i_4i_4}\\
\end{vmatrix}\\
=&D_{i_1i_1}D_{i_2i_2}D_{i_3i_3}D_{i_4i_4}+3D_{i_3i_3}D_{i_4i_4}V_{i_1i_2}+2D_{i_1i_1}D_{i_4i_4}V_{i_2i_3}\\
&+D_{i_1i_1}D_{i_2i_2}V_{i_3i_4}+3V_{i_1i_2}V_{i_3i_4}.
\end{split}
\end{equation*}

The corresponding Monobisyzexant is
\begin{equation*}
\hspace{-3cm}
    \begin{split}
        \operatorname{Mbsz}\left(\mathbf{D},\mathbf{V}\right)=&\frac{1}{4!}\sum_{i_1,\dots,i_4=1}^4\rho^{i_1i_2i_3i_4}\mathcal{S}_{i_1i_2i_3i_4}\left(\mathbf{D},\mathbf{V}\right)\\
        =&D_{11}D_{22}D_{33}D_{44}+D_{33}D_{44}V_{12}+D_{22}D_{44} V_{13}+D_{22}D_{33}V_{14}+D_{11}D_{44}V_{23}+D_{11}D_{33}V_{24}\\
        &+D_{11}D_{22}V_{34}+V_{14}V_{23}+V_{13}V_{24}+V_{12}V_{34}\\
        =&\det(\mathbf{D})+\det(\mathbf{D}_{[\{1,2\}^c]})\haf\left(\mathbf{V}_{[\{3,4\}^c]}\right)+\det(\mathbf{D}_{[\{1,3\}^c]})\haf\left(\mathbf{V}_{[\{2,4\}^c]}\right)+\det(\mathbf{D}_{[\{1,4\}^c]})\haf\left(\mathbf{V}_{[\{2,3\}^c]}\right)\\
        &+\det(\mathbf{D}_{[\{2,3\}^c]})\haf\left(\mathbf{V}_{[\{1,4\}^c]}\right)+\det(\mathbf{D}_{[\{2,4\}^c]})\haf\left(\mathbf{V}_{[\{1,3\}^c]}\right)\\
        &+\det(\mathbf{D}_{[\{3,4\}^c]})\haf\left(\mathbf{V}_{[\{1,2\}^c]}\right)+\haf(\mathbf{V})\\
        =&\det(\mathbf{D})+\sum_{i_1<i_2=1}^4\det(\mathbf{D}_{[\{i_1,i_2\}]})\haf\left(\mathbf{V}_{[\{i_1,i_2\}^c]}\right)+\haf(\mathbf{V}).
    \end{split}
\end{equation*}
\section{Proofs}
\subsection{Kasteleyn parity for the combined sign}

We next show that the product of the Wick sign $\sigma(M)$ and the orientation sign $S(M):=\prod_{\{i,j\}\in M}S_{ij}$ is constant in $M$ on planar graphs. This is the crux of the Kasteleyn method.

\begin{lemma}[Kasteleyn parity lemma]\label{lem:kasteleyn-parity}
Fix a planar embedding and a Kasteleyn orientation. For any two perfect matchings $M$ and $M_0$,
\[
\frac{\sigma(M)\, S(M)}{\sigma(M_0)\, S(M_0)}=+1.
\]
Equivalently, the quantity $\sigma(M)\,S(M)$ is independent of $M$.
\end{lemma}

\begin{proof}
Consider the symmetric difference $M\triangle M_0$, which decomposes into a disjoint union of even-length cycles $C_1,\dots,C_r$ embedded in the plane. Passing from $M_0$ to $M$ flips the pairing along each cycle $C_k$.

\smallskip
\textit{(i) Wick sign ratio along a cycle.)} Reordering the Grassmann variables to implement the new pairing along a single even cycle of length $2\ell$ contributes a factor $(-1)^{\ell-1}$. (This is a standard computation: in a $2\ell$-tuple, the number of transpositions needed to swap pairings differs by $\ell-1$ modulo 2.) Therefore
\[
\frac{\sigma(M)}{\sigma(M_0)}=\prod_{k=1}^r (-1)^{\ell(C_k)-1},
\]
where $\ell(C)$ is half the length of cycle $C$.

\smallskip
\textit{(ii) Orientation sign ratio along a cycle.)} The orientation sign $S(M)$ is the product over edges in the matching of $S_{ij}$. On a cycle flip, exactly $\ell(C)$ edges of the matching on $C$ are replaced. The ratio of orientation signs contributed by cycle $C$ equals $(-1)^{c(C)}$, where $c(C)$ is the number of edges oriented \textit{clockwise} along $C$ in the embedding. (Indeed, each replaced edge contributes a factor $S_{ij}$, and traversing the cycle once counts the parity of clockwise edges.)

Thus
\[
\frac{S(M)}{S(M_0)}=\prod_{k=1}^r (-1)^{c(C_k)}.
\]

\smallskip
\textit{(iii) Combine and use the Kasteleyn condition.)} Multiply (i) and (ii):
\[
\frac{\sigma(M)\,S(M)}{\sigma(M_0)\,S(M_0)}
=\prod_{k=1}^r (-1)^{\ell(C_k)-1+c(C_k)}.
\]
For a Kasteleyn orientation on a planar embedding, one shows (see e.g.\ standard expositions of Kasteleyn's theorem) that along any even cycle $C$, the parity satisfies
\[
\ell(C)-1+c(C)\equiv 0 \pmod{2}.
\]
Intuitively, the boundary of each face has an odd number of clockwise edges; summing face-parities over the region enclosed by $C$ yields $c(C)\equiv \ell(C)-1\ (\text{mod }2)$. Hence each factor is $+1$, and the product is $+1$.
\end{proof}

The combinatorial content of this lemma finds its natural algebraic expression through Levi-Civita symbols and orientation signs. Consider the product structure:

\[
\epsilon^{i_1\dots i_{2L}}S_{i_1 i_2}\cdots S_{i_{2L-1}i_{2L}}
\]

where $\epsilon^{i_1\dots i_{2L}}$ is the Levi-Civita symbol and $S_{ij}$ are orientation signs. This product encodes both the permutation structure of the matching and the relative orientations of the edges.

The Levi-Civita symbol $\epsilon^{i_1\dots i_{2L}}$ serves two essential purposes:
\begin{itemize}
    \item It ensures the indices $(i_1,\dots,i_{2L})$ form a permutation of $\{1,\dots,2L\}$, thereby enforcing the condition that each vertex appears exactly once---the mathematical expression of a \textit{perfect matching}.
    
    \item In the Hafnian formulation, the square of the Levi-Civita symbol yields a constant positive sign for all contributing configurations, distinguishing it from the determinant, where permutation signs are preserved.
\end{itemize}

When we introduce an arbitrary orientation to the graph, we assign signs to the edges, transforming the weights from $A_{ij} = w_{ij}$ to $S_{ij}A_{ij} = \pm w_{ij}$. Defining a skew-symmetric matrix through this orientation is always possible, but the critical insight lies in the behavior of the combinatorial sign:
\[
\epsilon^{i_1\dots i_{2L}}S_{i_1 i_2}\cdots S_{i_{2L-1}i_{2L}}.
\]
For a general orientation, this product depends on the specific perfect matching $M$ and can be $\pm 1$. The remarkable property of the \textit{Kasteleyn orientation} is that it ensures this product equals $+1$ for \textit{every} perfect matching $M$.

The equivalence between these formulations arises from the identification:
\begin{itemize}
\item $\sigma(M) = \epsilon^{i_1\dots i_{2L}}$ (Wick sign = Levi-Civita permutation sign)
\item $S(M) = S_{i_1 i_2}\cdots S_{i_{2L-1}i_{2L}}$ (orientation sign product)
\end{itemize}

The Kasteleyn Parity Lemma guarantees that $\sigma(M)S(M)$ is constant, while the Levi-Civita formulation shows that this constant can be normalized to $+1$ through an appropriate Kasteleyn orientation. This normalization is precisely what enables the transformation from the Hafnian of the adjacency matrix to the Pfaffian of the Kasteleyn matrix, providing the mathematical foundation for exact computation in planar dimer models.

This equivalence demonstrates that the combinatorial invariance discovered by Kasteleyn corresponds to the existence of a special orientation that makes the Levi-Civita orientation product uniformly positive, thereby removing sign and enabling Pfaffian methods for enumeration.

\subsection{Proof of Theorem \ref{thm:main}}\label{Proof:TheoHfPf}
\begin{enumerate}[label=(\roman*)]
\item For a planar graph with a Kasteleyn orientation, consider any perfect matching on $2L$ vertices, denoted by the edge set $\{\{i_1,i_2\}, \{i_3,i_4\}, \dots, \{i_{2L-1},i_{2L}\}\}$. The product of the Kasteleyn orientation signs for this matching, expressed as $\epsilon^{i_1 i_2} \epsilon^{i_3 i_4} \cdots \epsilon^{i_{2L-1} i_{2L}}$, equals the full Levi-Civita symbol $\epsilon^{i_1 i_2 \cdots i_{2L-1} i_{2L}}$ that is independent of the specific matching. This statement is precisely equivalent to Lemma~\ref{lem:kasteleyn-parity}, which guarantees the invariance of the combined sign $\sigma(M)S(M)$ across all perfect matchings. This fundamental identity enables the transformation from the Hafnian of the adjacency matrix to the Pfaffian of the Kasteleyn matrix, providing the mathematical foundation for the exact solvability of planar dimer models.

Using the fundamental identity between the Levi-Civita symbol and Kasteleyn orientation signs, together with the definition of the Kasteleyn matrix elements $K_{ij} = \epsilon^{ij} A_{ij}$, we can directly relate the Hafnian and Pfaffian expressions:
\begin{equation*}
\begin{split}
\haf(\mathbf{A}) &= \frac{1}{2^L L!}\sum_{i_1,\dots,i_{2L}=1}^{2L}\epsilon^{i_1\dots i_{2L}}\epsilon^{i_1 i_2}\epsilon^{i_3 i_4}\cdots \epsilon^{i_{2L-1}i_{2L}}A_{i_1i_2}\cdots A_{i_{2L-1}i_{2L}}\\
&= \frac{1}{2^L L!}\sum_{i_1,\dots,i_{2L}=1}^{2L}\epsilon^{i_1\dots i_{2L}}(\epsilon^{i_1 i_2}A_{i_1i_2})(\epsilon^{i_3 i_4}A_{i_3i_4})\cdots (\epsilon^{i_{2L-1}i_{2L}}A_{i_{2L-1}i_{2L}})\\
&= \frac{1}{2^L L!}\sum_{i_1,\dots,i_{2L}=1}^{2L}\epsilon^{i_1\dots i_{2L}}K_{i_1i_2}K_{i_3i_4}\cdots K_{i_{2L-1}i_{2L}}\\
&= \pf(\mathbf{K}).
\end{split}
\end{equation*}

\item We can write an equivalent version of Lemma~\ref{lem:kasteleyn-parity} by using Grassmann--Berezin representation: 
\begin{equation*}
    \int \mathbf{D}\boldsymbol{\psi}\prod_{\alpha=1}^{2L}d\psi_{i_\alpha}=\epsilon^{i_1 i_2} \epsilon^{i_3 i_4} \cdots \epsilon^{i_{2L-1} i_{2L}}.
\end{equation*}
\end{enumerate}
By using the above and $K_{ij}=\epsilon^{ij} A_{ij}$, we can go from the fermionic Hafnian to the fermionic Pfaffian
\begin{equation}
\begin{split}
\int\mathbf{D}(\boldsymbol{\chi},\bar{\boldsymbol{\chi}})
e^{\frac{1}{2}(\bar{\boldsymbol{\chi}}\boldsymbol{\chi})^\top\mathbf{A}(\bar{\boldsymbol{\chi}}\boldsymbol{\chi})}&
=\int\mathbf{D}(\boldsymbol{\chi},\bar{\boldsymbol{\chi}})\left(\frac{1}{2}\sum_{i,j=1}^{2L}A_{ij}\bar{\chi}_i\chi_i\bar{\chi}_j\chi_j\right)^L\\
&=\frac{1}{2^L}\int\mathbf{D}\boldsymbol{\chi}\sum_{i_1,\dots,i_{2L}=1}^{2L}\int\mathbf{D}\bar{\boldsymbol{\chi}}\bar{\chi}_{i_1}\cdots\bar{\chi}_{i_{2L}}\prod_{\alpha=1}^LA_{i_{2\alpha-1}i_{2\alpha}}\chi_{i_{2\alpha-1}}\chi_{i_{2\alpha}}\\
&=\frac{1}{2^L}\int\mathbf{D}\boldsymbol{\chi}\sum_{i_1,\dots,i_{2L}=1}^{2L}\prod_{\alpha=1}^L(\epsilon^{i_{2\alpha-1}i_{2\alpha}}A_{i_{2\alpha-1}i_{2\alpha}})\chi_{i_{2\alpha-1}}\chi_{i_{2\alpha}}\\
&=\int\mathbf{D}\boldsymbol{\chi}\left(\frac{1}{2}\sum_{i,j=1}^{2L}K_{ij}\chi_i\chi_j\right)^L\\
&=\int \mathbf{D}\boldsymbol{\chi} e^{\frac{1}{2}\boldsymbol{\chi}^\top\mathbf{K}\boldsymbol{\chi}}.
\end{split}
\end{equation}

\subsection{Proof of Theorem~\ref{PropHafPf}}\label{ProofPropHafPf}

\begin{enumerate}[label=(\roman*)]
    \item The proof follows the strategy of Theorem~\ref{thm:main}, beginning with the Hafnian of the restricted adjacency matrix $\mathbf{A}_{[I^c]}$. The second line factorises the Levi-Civita symbol using the Kasteleyn orientation, introducing sign factors $\epsilon^{i_1i_2}\cdots\epsilon^{i_{2r-1}i_{2r}}$ associated with monomer insertions. The third line identifies the terms $\epsilon^{jk}A_{jk}$ as elements of the original Kasteleyn matrix $\mathbf{K}$. The product of monomer signs is then absorbed into a redefinition of the Kasteleyn matrix, yielding the modified matrix $\mathbf{K'}$. The fifth line recognizes the expression as the Pfaffian of $\mathbf{K'}_{[I^c]}$. The final result follows from Jacobi's identity for Pfaffians, which relates the Pfaffian of a submatrix to the product of the Pfaffian of the full matrix and the Pfaffian of the inverse submatrix.
\begin{equation*}
\hspace{-3.4cm}
\begin{split}
\haf\left(\mathbf{A}_{\left[I^c_{[2r]}\right]}\right)&=\frac{1}{2^{L-r}\left(L-r\right)!}\sum_{j_1,\dots,j_{2(L-r)}=1}^{2L}\epsilon^{i_1\cdots i_{2r}j_1\cdots j_{2(L-r)}}\epsilon^{i_1\cdots i_{2r}j_1\cdots j_{2(L-r)}}A_{j_1j_2}\cdots A_{j_{2(N-r)-1}j_{2(N-r)}}\\
&=\frac{1}{2^{L-r}\left(L-r\right)!}\sum_{j_1,\dots,j_{2(L-r)}=1}^{2L}\epsilon^{i_1\cdots i_{2r}j_1\cdots j_{2(L-r)}}\epsilon^{i_1i_{2}}\cdots\epsilon^{i_{2r-1}i_{2r}}\epsilon^{j_1 j_2}\cdots \epsilon^{j_{2L-1}j_{2L}}A_{j_1j_2}\cdots A_{j_{2(N-r)-1}j_{2(N-r)}}\\
&=\frac{1}{2^{L-r}\left(L-r\right)!}\sum_{j_1,\dots,j_{2(L-r)}=1}^{2L}\epsilon^{i_1\cdots i_{2r}j_1\cdots j_{2(L-r)}}\epsilon^{i_1i_{2}}\cdots\epsilon^{i_{2r-1}i_{2r}}(\epsilon^{j_1 j_2}A_{j_1j_2})\cdots (\epsilon^{j_{2L-1}j_{2L}}A_{j_{2(N-r)-1}j_{2(N-r)}})\\
&=\frac{\epsilon\left(I_{[2r]}\right)\epsilon\left(I_{[2r]}\right)}{2^{L-r}(L-r)!}\sum_{j_1,\dots,j_{2(L-r)}=1}^{2L}\varepsilon^{i_1\cdots i_{2r}j_1\cdots j_{2(L-r)}}K'_{j_1j_2}\cdots K'_{j_{2(L-r)-1}j_{2(L-r)}}\\
&=\epsilon\left(I_{[2r]}\right)\pf\left(\mathbf{K'}_{\left[I^c_{[2r]}\right]}\right)\\
&=(-1)^r\pf(\mathbf{K'})\pf\left((\mathbf{K}'^{-1})_{\left[I_{[2r]}\right]}\right).
\end{split}
\end{equation*}
\item Here we again follow the ideas of proof of Theorem~\ref{thm:main}.

\begin{equation}
\hspace{-3cm}
\begin{split}
\int\mathbf{D}&(\boldsymbol{\chi},\bar{\boldsymbol{\chi}})
\left(\prod_{\alpha=1}^{2r}\bar{\chi}_{k_\alpha}\chi_{k_\alpha}\right)e^{\frac{1}{2}(\bar{\boldsymbol{\chi}}\boldsymbol{\chi})^\top\mathbf{A}(\bar{\boldsymbol{\chi}}\boldsymbol{\chi})}\\
&=\int\mathbf{D}(\boldsymbol{\chi},\bar{\boldsymbol{\chi}})\left(\prod_{\alpha=1}^{2r}\bar{\chi}_{k_\alpha}\chi_{k_\alpha}\right)\left(\frac{1}{2}\sum_{i,j=1}^{2L}A_{ij}\bar{\chi}_i\chi_i\bar{\chi}_j\chi_j\right)^{L-r}\\
&=\frac{1}{2^{L-r}}\int\mathbf{D}\boldsymbol{\chi}\left(\prod_{\alpha=1}^{2r}\chi_{k_\alpha}\right)\sum_{j_1,\dots,j_{2(L-r)}=1}^{2L}\int\mathbf{D}\bar{\boldsymbol{\chi}}\bar{\chi}_{i_1}\cdots\bar{\chi}_{i_{2r}}\bar{\chi}_{j_1}\cdots\bar{\chi}_{j_{2(L-r)}}\prod_{\alpha=1}^{L-r}A_{j_{2\alpha-1}j_{2\alpha}}\chi_{j_{2\alpha-1}}\chi_{j_{2\alpha}}\\
&=\frac{1}{2^{L-r}}\int\mathbf{D}\boldsymbol{\chi}\left(\prod_{\alpha=1}^{2r}\chi_{k_\alpha}\right)\sum_{j_1,\dots,j_{2(L-r)}=1}^{2L}\epsilon^{i_1i_2}\cdots\epsilon^{i_{2r-1}i_{2r}}\prod_{\alpha=1}^{L-r}(\epsilon^{j_{2\alpha-1}j_{2\alpha}}A_{j_{2\alpha-1}j_{2\alpha}})\chi_{j_{2\alpha-1}}\chi_{j_{2\alpha}}\\
&=\int\mathbf{D}\boldsymbol{\chi}\left(\prod_{\alpha=1}^{2r}\chi_{k_\alpha}\right)\left(\frac{1}{2}\sum_{i,j=1}^{2L}K'_{ij}\chi_i\chi_j\right)^{L-r}\\
&=\int \mathbf{D}\boldsymbol{\chi}\left(\prod_{\alpha=1}^{2r}\chi_{k_\alpha}\right)e^{\frac{1}{2}\boldsymbol{\chi}^\top\mathbf{K'}\boldsymbol{\chi}}.
\end{split}
\end{equation}
\end{enumerate}
\subsection{Proof of Theorem~\ref{thm:main2}}\label{Proof:thm:main2}
Our starting point is Lemma~\ref{Mbszexpanded}, which provides an expansion of the Monobisyzexant function in terms of Determinants (which, for the diagonal matrix $\mathbf{D}$, reduce to products of its diagonal entries) and Hafnian, and hafnianinhos. Using Theorem~\ref{thm:main} and Proposition~\ref{PropHafPf}, we systematically rewrite this expansion in terms of Pfaffians and Pfaffinhos of the modified Kasteleyn matrix. This transformation is exclusively valid for planar graphs, as it relies on the combinatorial properties guaranteed by the Kasteleyn orientation. The Grassmann-Berezin integral formulation follows naturally by applying the same theorem and proposition to each term in the expansion, thereby establishing the equivalence between the combinatorial and Grassmann-Berezin representations of the monomer-dimer partition function.
\subsection{Proof of Theorem \ref{exponential with linear fermionic and bosonic sources}}\label{Proof exponential with linear fermionic and bosonic sources}
\begin{proof}
The proof of part (a) begins by decomposing the expression into separate exponentials. We then expand the exponential containing the linear terms, which yields:  
\begin{equation}
\begin{aligned}
\int \mathbf{D}\boldsymbol{\chi}\,
e^{\tfrac{1}{2}\boldsymbol{\chi}^T \mathbf{A} \boldsymbol{\chi} 
+ \mathbf{u}^T \boldsymbol{\chi} 
+ \boldsymbol{\psi}^T \boldsymbol{\chi}}
&= \int \mathbf{D}\boldsymbol{\chi}\, 
\left(1 + \mathbf{u}^T \boldsymbol{\chi}\right)\, \left[ 1 + \sum_{m=1}^{L} \frac{1}{m!} \big(\boldsymbol{\psi}^T\boldsymbol{\chi}\big)^m \right] 
e^{\tfrac{1}{2}\boldsymbol{\chi}^T \mathbf{A} \boldsymbol{\chi} 
} \\[1ex]
&= \int \mathbf{D}\boldsymbol{\chi}\, 
e^{\tfrac{1}{2}\boldsymbol{\chi}^T \mathbf{A} \boldsymbol{\chi} 
}
+ \sum_{k=1}^L u_k \int \mathbf{D}\boldsymbol{\chi}\,
\chi_k \, e^{\tfrac{1}{2}\boldsymbol{\chi}^T \mathbf{A} \boldsymbol{\chi} 
} \\[1ex]
&+ \sum_{m=1}^{L} \frac{1}{m!} \int \mathbf{D}\boldsymbol{\chi}\,  \big(\boldsymbol{\psi}^T\boldsymbol{\chi}\big)^m
e^{\tfrac{1}{2}\boldsymbol{\chi}^T \mathbf{A} \boldsymbol{\chi} 
} + \sum_{k=1}^L u_k \sum_{m=1}^{L} \frac{1}{m!} \int \mathbf{D}\boldsymbol{\chi}\, \chi_k \big(\boldsymbol{\psi}^T\boldsymbol{\chi}\big)^m
e^{\tfrac{1}{2}\boldsymbol{\chi}^T \mathbf{A} \boldsymbol{\chi} 
}.
\end{aligned}
\end{equation}
We now distinguish between the two possible scenarios for \(L\):

\paragraph{Case 1: \(L\) even.}  
In this case, the first contribution, as established in \cite{caracciolo2013algebraic}, is given by 
\(\operatorname{pf}[\mathbf{A}]\). The second contribution vanishes identically, while in the third summation only the even values of \(m\) yield nonzero terms. Conversely, in the fourth summation, only the odd values of \(m\) contribute. Collecting these observations, we obtain:
\begin{equation}
\begin{aligned}
\int \mathbf{D}\boldsymbol{\chi}\,
e^{\tfrac{1}{2}\boldsymbol{\chi}^T \mathbf{A} \boldsymbol{\chi} 
+ \mathbf{u}^T \boldsymbol{\chi} 
+ \boldsymbol{\psi}^T \boldsymbol{\chi}} & = \operatorname{pf}(\mathbf{A})+ \sum_{m=1}^{\tfrac{L}{2}} \frac{(-1)^{m}}{(2m)!}
\sum_{j_1,\ldots,j_{2m}=1}^L
\left( \int \mathbf{D}\boldsymbol{\chi} \, \prod_{\alpha=1}^{2m} \chi_{j_\alpha} 
e^{\frac{1}{2}\boldsymbol{\chi}^T\mathbf{A}\boldsymbol{\chi}} \right)
\prod_{\alpha=1}^{2m} \psi_{j_\alpha}\\
&  + \sum_{m=1}^{\tfrac{L}{2}} \frac{(-1)^{m}}{(2m-1)!}
\sum_{j_0,j_1,\ldots,j_{2m-1}=1}^L
\left( \int \mathbf{D}\boldsymbol{\chi} \, \prod_{\alpha=0}^{2m-1} \chi_{j_\alpha} 
e^{\frac{1}{2}\boldsymbol{\chi}^T\mathbf{A}\boldsymbol{\chi}} \right) u_{j_0}
\prod_{\alpha=1}^{2m-1} \psi_{j_\alpha}.
\end{aligned}
\end{equation}
By explicitly performing the Grassmann integrations, the expression simplifies to the following closed form:
\begin{equation}
    \begin{aligned}
        \int \mathbf{D}\boldsymbol{\chi}\,
e^{\tfrac{1}{2}\boldsymbol{\chi}^T \mathbf{A} \boldsymbol{\chi} 
+ \mathbf{u}^T \boldsymbol{\chi} 
+ \boldsymbol{\psi}^T \boldsymbol{\chi}} &= \pf(\mathbf{A})+\sum_{m=1}^{L/2}\left[
   \frac{(-1)^m}{(2m)!}
   \sum_{j_1,\dots,j_{2m}=1}^L
   \epsilon\left(J_{[2m]}\right)\pf\left(\mathbf{A}_{\left[J^c_{[2m]}\right]}\right)
   \left(\prod_{\alpha=1}^{2m}\psi_{j_\alpha}\right)\right.\\
&\qquad+\frac{(-1)^m}{(2m-1)!}
   \sum_{j_0,j_1,\dots,j_{2m-1}=1}^L
   \epsilon\left(J_{[0,2m-1]}\right)\pf\left(\mathbf{A}_{\left[J^c_{[0,2m-1]}\right]}\right)
   \left(u_{j_0}\prod_{\alpha=1}^{2m-1}\psi_{j_\alpha}\right)
\Bigg],
    \end{aligned}
\end{equation}
where $J_{[k]}=\{j_1,\dots,j_{k}\}\subseteq[L]$ and $J_{[0,k]}=\{j_0,j_1,\dots,j_{k}\}\subseteq[L]$. Using Corollary \ref{exponential with bosonic sources}, the result can be simplified as:
\begin{equation}
    \begin{aligned}
        \int \mathbf{D}\boldsymbol{\chi}\,
e^{\tfrac{1}{2}\boldsymbol{\chi}^T \mathbf{A} \boldsymbol{\chi} 
+ \mathbf{u}^T \boldsymbol{\chi} 
+ \boldsymbol{\psi}^T \boldsymbol{\chi}} &= \pf(\mathbf{A})+\sum_{m=1}^{L/2}\left[
   \frac{(-1)^m}{(2m)!}
   \sum_{j_1,\dots,j_{2m}=1}^L
\epsilon\left(J_{[2m]}\right)\pf\left(\mathbf{A}_{\left[J^c_{[2m]}\right]}\right)
   \left(\prod_{\alpha=1}^{2m}\psi_{j_\alpha}\right)\right.\\
&\qquad-\frac{(-1)^m}{(2m-1)!}
   \sum_{j_1,\dots,j_{2m-1}=1}^L
   \epsilon\left(J_{[2m-1]}\right)\pf\left(\mathbf{X}_{\left[J^c_{[2m-1]}\right]}\right)
   \left(\prod_{\alpha=1}^{2m-1}\psi_{j_\alpha}\right)
\Bigg],
    \end{aligned}
\end{equation}

\paragraph{Case 2: \(L\) odd.} In this case, the first contribution vanishes identically. The second term simplifies to the Pfaffinho of the matrix $\mathbf{A}$. Regarding the summations, in the third series only odd values of $m$ contribute nontrivially, while in the fourth series the contributions are restricted to even values of $m$. Collecting these observations, we obtain:
\begin{equation}
\begin{aligned}
 \int \mathbf{D}\boldsymbol{\chi}\,
e^{\tfrac{1}{2}\boldsymbol{\chi}^T \mathbf{A} \boldsymbol{\chi} 
+ \mathbf{u}^T \boldsymbol{\chi} 
+ \boldsymbol{\psi}^T \boldsymbol{\chi}}  =&  -\sum_{j_0=1}^L(-1)^{j_0}u_{j_0}\pf\left(\mathbf{A}_{[\{j_0\}^c]}\right)\\&  + \sum_{m=1}^{\tfrac{L+1}{2}} \frac{(-1)^{m}}{(2m-1)!}
\sum_{j_1,\ldots,j_{2m-1}=1}^L
\left( \int \mathbf{D}\boldsymbol{\chi} \, \prod_{\alpha=1}^{2m-1} \chi_{j_\alpha} 
e^{\frac{1}{2}\boldsymbol{\chi}^T\mathbf{A}\boldsymbol{\chi}} \right)
\prod_{\alpha=1}^{2m-1} \psi_{j_\alpha}\\
&  + \sum_{m=1}^{\tfrac{L-1}{2}} \frac{(-1)^{m}}{(2m)!}
\sum_{j_0,j_1,\ldots,j_{2m}=1}^L
\left( \int \mathbf{D}\boldsymbol{\chi} \, \prod_{\alpha=0}^{2m} \chi_{j_\alpha} 
e^{\frac{1}{2}\boldsymbol{\chi}^T\mathbf{A}\boldsymbol{\chi}} \right) u_{j_0}
\prod_{\alpha=1}^{2m} \psi_{j_\alpha}.
\end{aligned}
\end{equation}
By performing the Grassmann integrations, we have:
\begin{equation}
    \begin{aligned}
       \int \mathbf{D}\boldsymbol{\chi}\,
e^{\tfrac{1}{2}\boldsymbol{\chi}^T \mathbf{A} \boldsymbol{\chi} 
+ \mathbf{u}^T \boldsymbol{\chi} 
+ \boldsymbol{\psi}^T \boldsymbol{\chi}} =& -\sum_{m=1}^{(L+1)/2}\frac{(-1)^m}{(2m-1)!}\sum_{j_1,\dots,j_{2m-1}=1}^L\epsilon\left(J_{[2m-1]}\right)\pf\!\left(\mathbf{A}_{\left[J_{[2m-1]}^c\right]}\right)\left(\prod_{\alpha=1}^{2m-1}\psi_{j_\alpha}\right)\\
&-\sum_{m=1}^{(L-1)/2}\frac{(-1)^m}{(2m)!}\sum_{j_0,j_1,\dots,j_{2m}=1}^L\epsilon\left(J_{[0,2m]}\right)\pf\left(\mathbf{A}_{\left[J^c_{[0,2m]}\right]}\right)\left(u_{j_0}\prod_{\alpha=1}^{2m}\psi_{j_\alpha}\right)\\
&-\sum_{j_0=1}^L(-1)^{j_0}u_{j_0}\pf\left(\mathbf{A}_{[\{j_0\}^c]}\right),
    \end{aligned}
\end{equation}
which using Corollary \ref{exponential with bosonic sources} can be simplified as:
\begin{equation}
    \begin{aligned}
       \int \mathbf{D}\boldsymbol{\chi}\,
e^{\tfrac{1}{2}\boldsymbol{\chi}^T \mathbf{A} \boldsymbol{\chi} 
+ \mathbf{u}^T \boldsymbol{\chi} 
+ \boldsymbol{\psi}^T \boldsymbol{\chi}} &= \operatorname{pf}(\mathbf{X}) \\&\quad-\sum_{m=1}^{(L+1)/2}\frac{(-1)^m}{(2m-1)!}\sum_{j_1,\dots,j_{2m-1}=1}^L\epsilon\left(J_{[2m-1]}\right)\pf\!\left(\mathbf{A}_{\left[J_{[2m-1]}^c\right]}\right)\left(\prod_{\alpha=1}^{2m-1}\psi_{j_\alpha}\right)\\
&\quad+\sum_{m=1}^{(L-1)/2}\frac{(-1)^m}{(2m)!}\sum_{j_1,\dots,j_{2m}=1}^L\epsilon\left(J_{[2m]}\right)\pf\left(\mathbf{X}_{\left[J^c_{[2m]}\right]}\right)\left(\prod_{\alpha=1}^{2m}\psi_{j_\alpha}\right).
    \end{aligned}
\end{equation}

The proof of part (b) begins by decomposing the expression into separate exponentials. We then expand the exponential containing the linear complex terms, which yields:  
\begin{equation}
\begin{aligned}
\int \mathbf{D}\boldsymbol{\chi}\, \left(\prod_{\alpha=1}^r\chi_{i_\alpha}\right) 
e^{\tfrac{1}{2}\boldsymbol{\chi}^T \mathbf{A} \boldsymbol{\chi} 
+ \mathbf{u}^T \boldsymbol{\chi} 
+ \boldsymbol{\psi}^T \boldsymbol{\chi}}
&= \int \mathbf{D}\boldsymbol{\chi}\, \left(\prod_{\alpha=1}^r\chi_{i_\alpha}\right)
\left(1 + \mathbf{u}^T \boldsymbol{\chi}\right)\, 
e^{\tfrac{1}{2}\boldsymbol{\chi}^T \mathbf{A} \boldsymbol{\chi} + \boldsymbol{\psi}^T \boldsymbol{\chi}
} \\[1ex]
&= \int \mathbf{D}\boldsymbol{\chi}\, \left(\prod_{\alpha=1}^r\chi_{i_\alpha}\right)  
e^{\tfrac{1}{2}\boldsymbol{\chi}^T \mathbf{A} \boldsymbol{\chi} + \boldsymbol{\psi}^T \boldsymbol{\chi}
}\\[1ex] 
&
+ \sum_{k=1}^L u_k \int \mathbf{D}\boldsymbol{\chi}\, \left(\prod_{\alpha=1}^r\chi_{i_\alpha}\right) 
\chi_k \, e^{\tfrac{1}{2}\boldsymbol{\chi}^T \mathbf{A} \boldsymbol{\chi} + \boldsymbol{\psi}^T \boldsymbol{\chi}
}. 
\end{aligned}
\end{equation}
We now distinguish between the four possible scenarios for \(L\) and \(r\):

\paragraph{Case 1: \(L\) even, \(r\) even.} Now, by systematically expanding the remaining linear contribution, we obtain:
\begin{equation}
\begin{aligned}
\int \mathbf{D}\boldsymbol{\chi} \prod_{\alpha=1}^r \chi_{i_\alpha} 
   e^{\tfrac{1}{2}\boldsymbol{\chi}^T \mathbf{A} \boldsymbol{\chi} 
   + \mathbf{u}^T \boldsymbol{\chi} + \boldsymbol{\psi}^T \boldsymbol{\chi}} &= \int \mathbf{D}\boldsymbol{\chi} \prod_{\alpha=1}^r \chi_{i_\alpha} 
   \Biggl[ 1 + \sum_{m=1}^L \frac{1}{m!} (\boldsymbol{\psi}^T \boldsymbol{\chi})^m \Biggr] 
   e^{\tfrac{1}{2}\boldsymbol{\chi}^T \mathbf{A} \boldsymbol{\chi}} \\[1ex]
&\quad + \sum_{k=1}^L u_k \int \mathbf{D}\boldsymbol{\chi}\, \left(\prod_{\alpha=1}^r\chi_{i_\alpha}\right) 
\chi_k 
   \Biggl[ 1 + \sum_{m=1}^L \frac{1}{m!} (\boldsymbol{\psi}^T \boldsymbol{\chi})^m \Biggr] 
   e^{\tfrac{1}{2}\boldsymbol{\chi}^T \mathbf{A} \boldsymbol{\chi}}. \\[1ex]
&= \int \mathbf{D}\boldsymbol{\chi} \prod_{\alpha=1}^r \chi_{i_\alpha} 
   e^{\tfrac{1}{2}\boldsymbol{\chi}^T \mathbf{A} \boldsymbol{\chi}} 
   + \sum_{m=1}^L \frac{1}{m!} \int \mathbf{D}\boldsymbol{\chi} 
   \prod_{\alpha=1}^r \chi_{i_\alpha} (\boldsymbol{\psi}^T \boldsymbol{\chi})^m 
   e^{\tfrac{1}{2}\boldsymbol{\chi}^T \mathbf{A} \boldsymbol{\chi}} \\[1ex]
&\quad + \sum_{k=1}^L u_k \int \mathbf{D}\boldsymbol{\chi}\, \left(\prod_{\alpha=1}^r\chi_{i_\alpha}\right) 
\chi_k  
   e^{\tfrac{1}{2}\boldsymbol{\chi}^T \mathbf{A} \boldsymbol{\chi}}\\[1ex]
   &\quad + \sum_{k=1}^L u_k \sum_{m=1}^L \frac{1}{m!} \int \mathbf{D}\boldsymbol{\chi}\, \left(\prod_{\alpha=1}^r\chi_{i_\alpha}\right) 
\chi_k  (\boldsymbol{\psi}^T \boldsymbol{\chi})^m
   e^{\tfrac{1}{2}\boldsymbol{\chi}^T \mathbf{A} \boldsymbol{\chi}}.
\end{aligned}
\end{equation}
With the same procedure as employed in part (a), the first contribution reduces to the Pfaffinho of the matrix $\mathbf{A}$, while the third term vanishes identically. Concerning the summations, in the second series only the even values of $m$ yield nontrivial contributions, whereas in the fourth series the surviving terms are likewise restricted to even $m$. Collecting these observations and doing the integrals, we obtain:

\begin{equation}
\begin{aligned}
\hspace{-2cm}\int \mathbf{D}\boldsymbol{\chi}\, \left(\prod_{\alpha=1}^r\chi_{i_\alpha}\right) 
e^{\tfrac{1}{2}\boldsymbol{\chi}^T \mathbf{A} \boldsymbol{\chi} 
+ \mathbf{u}^T \boldsymbol{\chi} 
+ \boldsymbol{\psi}^T \boldsymbol{\chi}}&= \epsilon(I_{[r]})\,\operatorname{Pf}\!\big(\mathbf{A}_{[I_{[r]}^c]}\big) 
   + \sum_{m=1}^{\tfrac{L-r}{2}} \frac{(-1)^m}{(2m)!} 
   \sum_{j_1,\dots,j_{2m}=1}^L \epsilon(I_{[r]} \cup J_{[2m]}) 
   \operatorname{Pf}\!\big(\mathbf{A}_{[I_{[r]}^c \cap J_{[2m]}^c]}\big) 
   \prod_{\alpha=1}^{2m} \psi_{j_\alpha} \\[1ex]
&\quad + \sum_{m=1}^{\tfrac{L-r}{2}} \frac{(-1)^m}{(2m-1)!} 
   \sum_{j_0,j_1,\dots,j_{2m-1}=1}^L \epsilon(I_{[r]} \cup J_{[0,2m-1]}) 
   \operatorname{Pf}\!\big(\mathbf{A}_{[I_{[r]}^c \cap J_{[0,2m-1]}^c]}\big) 
   u_{j_0}\prod_{\alpha=1}^{2m-1} \psi_{j_\alpha}.
\end{aligned}
\end{equation}
Using Corollary \ref{exponential with bosonic sources}, the result can be simplified as:
\begin{equation}
\begin{aligned}
\hspace{-1cm}\int \mathbf{D}\boldsymbol{\chi}\, \left(\prod_{\alpha=1}^r\chi_{i_\alpha}\right) 
e^{\tfrac{1}{2}\boldsymbol{\chi}^T \mathbf{A} \boldsymbol{\chi} 
+ \mathbf{u}^T \boldsymbol{\chi} 
+ \boldsymbol{\psi}^T \boldsymbol{\chi}}&= \epsilon(I_{[r]})\pf\left(\mathbf{A}_{[I_{[r]}^c]}\right)\\&\qquad+\sum_{m=1}^{(L-r)/2}\left[
   \frac{(-1)^m}{(2m)!}
   \sum_{j_1,\dots,j_{2m}=1}^L
   \epsilon\left(I_{[r]}\cup J_{[2m]}\right)\pf\left(\mathbf{A}_{\left[I_{[r]}^c\cap J^c_{[2m]}\right]}\right)
   \left(\prod_{\alpha=1}^{2m}\psi_{j_\alpha}\right)\right.\\
&\qquad-\left.\frac{(-1)^m}{(2m-1)!}
   \sum_{j_1,\dots,j_{2m-1}=1}^L
   \epsilon\left(I_{[r]}\cup J_{[2m-1]}\right)\pf\left(\mathbf{X}_{\left[I_{[r]}^c\cap J^c_{[2m-1]}\right]}\right)
   \left(\prod_{\alpha=1}^{2m-1}\psi_{j_\alpha}\right)
\right].
\end{aligned}
\end{equation}

Applying the same procedure and invoking part (a), we derive the remaining three scenarios as follows:

\paragraph{Case 2: \(L\) even, \(r\) odd.}

\begin{equation}
\begin{aligned}
\hspace{-1cm}&\int \mathbf{D}\boldsymbol{\chi}\, \left(\prod_{\alpha=1}^r\chi_{i_\alpha}\right) 
e^{\tfrac{1}{2}\boldsymbol{\chi}^T \mathbf{A} \boldsymbol{\chi} 
+ \mathbf{u}^T \boldsymbol{\chi} 
+ \boldsymbol{\psi}^T \boldsymbol{\chi}}\\
&= \sum_{m=1}^{(L-r+1)/2}\frac{(-1)^m}{(2m-1)!}\sum_{j_1,\dots,j_{2m-1}=1}^L\epsilon\left(I_{[r]}\cup J_{[2m-1]}\right)\pf\left(\mathbf{A}_{\left[I_{[r]}^c\cap J^c_{[2m-1]}\right]}\right)\left(\prod_{\alpha=1}^{2m-1}\psi_{j_\alpha}\right)\\
&\qquad-\sum_{m=1}^{(L-r-1)/2}\frac{(-1)^m}{(2m)!}\sum_{j_1,\dots,j_{2m}=1}^L\epsilon\left(I_{[r]}\cup J_{[2m]}\right)\pf\left(\mathbf{X}_{\left[I_{[r]}^c\cap J^c_{[2m]}\right]}\right)\left(\prod_{\alpha=1}^{2m}\psi_{j_\alpha}\right)-\epsilon(I_{[r]})\pf\left(\mathbf{X}_{[I_{[r]}^c]}\right). 
\end{aligned}
\end{equation}

\paragraph{Case 3: \(L\) odd, \(r\) even.}

\begin{equation}
\begin{aligned}
\hspace{-1cm}&\int \mathbf{D}\boldsymbol{\chi}\, \left(\prod_{\alpha=1}^r\chi_{i_\alpha}\right) 
e^{\tfrac{1}{2}\boldsymbol{\chi}^T \mathbf{A} \boldsymbol{\chi} 
+ \mathbf{u}^T \boldsymbol{\chi} 
+ \boldsymbol{\psi}^T \boldsymbol{\chi}}\\
&= -\sum_{m=1}^{(L-r+1)/2}\frac{(-1)^m}{(2m-1)!}\sum_{j_1,\dots,j_{2m-1}=1}^L\epsilon\left(I_{[r]}\cup J_{[2m-1]}\right)\pf\left(\mathbf{A}_{\left[I_{[r]}^c\cap J^c_{[2m-1]}\right]}\right)\left(\prod_{\alpha=1}^{2m-1}\psi_{j_\alpha}\right)\\
&\qquad+\sum_{m=1}^{(L-r-1)/2}\frac{(-1)^m}{(2m)!}\sum_{j_1,\dots,j_{2m}=1}^L\epsilon\left(I_{[r]}\cup J_{[2m]}\right)\pf\left(\mathbf{X}_{\left[I_{[r]}^c\cap J^c_{[2m]}\right]}\right)\left(\prod_{\alpha=1}^{2m}\psi_{j_\alpha}\right)+\epsilon(I_{[r]})\pf\left(\mathbf{X}_{[I_{[r]}^c]}\right). 
\end{aligned}
\end{equation}

\paragraph{Case 4: \(L\) odd, \(r\) odd.}

\begin{equation}
\begin{aligned}
\hspace{-1cm}\int \mathbf{D}\boldsymbol{\chi}\, \left(\prod_{\alpha=1}^r\chi_{i_\alpha}\right) 
e^{\tfrac{1}{2}\boldsymbol{\chi}^T \mathbf{A} \boldsymbol{\chi} 
+ \mathbf{u}^T \boldsymbol{\chi} 
+ \boldsymbol{\psi}^T \boldsymbol{\chi}}&=- \epsilon(I_{[r]})\pf\left(\mathbf{A}_{[I_{[r]}^c]}\right)\\&\qquad-\sum_{m=1}^{(L-r)/2}\left[
   \frac{(-1)^m}{(2m)!}
   \sum_{j_1,\dots,j_{2m}=1}^L
   \epsilon\left(I_{[r]}\cup J_{[2m]}\right)\pf\left(\mathbf{A}_{\left[I_{[r]}^c\cap J^c_{[2m]}\right]}\right)
   \left(\prod_{\alpha=1}^{2m}\psi_{j_\alpha}\right)\right.\\
&\qquad+\left.\frac{(-1)^m}{(2m-1)!}
   \sum_{j_1,\dots,j_{2m-1}=1}^L
   \epsilon\left(I_{[r]}\cup J_{[2m-1]}\right)\pf\left(\mathbf{X}_{\left[I_{[r]}^c\cap J^c_{[2m-1]}\right]}\right)
   \left(\prod_{\alpha=1}^{2m-1}\psi_{j_\alpha}\right)
\right].
\end{aligned}
\end{equation}
These four possible scenarios can be written in a more compact form as eq \eqref{part (b)}.\\

To prove \eqref{part (c)}, at first, we rewrite the integral as follows:
 \begin{equation}
\begin{split}
\int \mathbf{D}\boldsymbol{\chi} (C{\chi})_{1} (C{\chi})_{2} \cdots (C{\chi})_{r} e^{\tfrac{1}{2}\boldsymbol{\chi}^T\mathbf{A}\boldsymbol{\chi}+\mathbf{u}^T\boldsymbol{\chi}+\boldsymbol{\psi}^T\boldsymbol{\chi}}=\sum\limits_{i_1,...,i_r} C_{1i_1}\cdots C_{ri_r} \int \mathbf{D}\boldsymbol{\chi} {\chi}_{i_1} {\chi}_{i_2} \cdots {\chi}_{i_r} e^{\tfrac{1}{2}\boldsymbol{\chi}^T\mathbf{A}\boldsymbol{\chi}+\mathbf{u}^T\boldsymbol{\chi}+\boldsymbol{\psi}^T\boldsymbol{\chi}}.
\end{split}
\end{equation}
Note that the only nonvanishing term in the sum comes when $i_1,\cdots,i_r$ are all distinct, then we will have:
 \begin{equation}
\begin{split}
\hspace{-2.5cm}\sum\limits_{i_1,...,i_r} C_{1i_1}\cdots C_{ri_r}&\int \mathbf{D}\boldsymbol{\chi} {\chi}_{i_1} {\chi}_{i_2} \cdots {\chi}_{i_r} e^{\tfrac{1}{2}\boldsymbol{\chi}^T\mathbf{A}\boldsymbol{\chi}+\mathbf{u}^T\boldsymbol{\chi}+\boldsymbol{\psi}^T\boldsymbol{\chi}}\\
&\qquad=\sum\limits_{|{I}|=r}\sum\limits_{\sigma \in S_r} C_{1i_{\sigma(1)}}\cdots C_{ri_{\sigma(r)}} \text{sgn}(\sigma) \int \mathbf{D}\boldsymbol{\chi} {\chi}_{i_1} {\chi}_{i_2} \cdots {\chi}_{i_r} e^{\tfrac{1}{2}\boldsymbol{\chi}^T\mathbf{A}\boldsymbol{\chi}+\mathbf{u}^T\boldsymbol{\chi}+\boldsymbol{\psi}^T\boldsymbol{\chi}},
\end{split}
\end{equation}
where the second sum run over all permutation $\sigma$, and $\text{sgn}(\sigma)$ is the sign of the permutation $\sigma$. Finally, using \eqref{part (b)}, the expression simplifies to equation~\eqref{part (d)}.
\end{proof}

\subsection{Proof of Corollary \ref{exponential with linear fermionic and bosonic sources (invertible)}}\label{Proof exponential with linear fermionic and bosonic sources (invertible)}
The proof of part (a) even case begins with the condition that $\pf(\mathbf{A}) \neq 0$ for even $L$, which guarantees that the inverse $\mathbf{A}^{-1}$ exists.  

We start by rewriting the expression from part (a):  
\begin{equation*}
\begin{split}
\int \mathbf{D}\boldsymbol{\chi}\,
e^{\tfrac{1}{2}\boldsymbol{\chi}^T \mathbf{A} \boldsymbol{\chi} + \mathbf{u}^T \boldsymbol{\chi} + \boldsymbol{\psi}^T \boldsymbol{\chi}}
&= \pf(\mathbf{A}) 
+ \sum_{m=1}^{L/2} \Bigg[
   \frac{(-1)^m}{(2m)!}
   \sum_{j_1,\dots,j_{2m}=1}^L
   \epsilon\!\left(J_{[2m]}\right)\pf\!\left(\mathbf{A}_{[J^c_{[2m]}]}\right)
   \Biggl(\prod_{\alpha=1}^{2m}\psi_{j_\alpha}\Biggr) \\
&\quad + \frac{(-1)^m}{(2m-1)!}
   \sum_{j_0,j_1,\dots,j_{2m-1}=1}^L
   \epsilon\!\left(J_{[0,2m-1]}\right)\pf\!\left(\mathbf{A}_{[J^c_{[0,2m-1]}]}\right)
   \Biggl(u_{j_0}\prod_{\alpha=1}^{2m-1}\psi_{j_\alpha}\Biggr)
\Bigg] \\
&= \pf(\mathbf{A})\Bigg\{1 + \sum_{m=1}^{L/2}\Bigg[
   \frac{(-1)^m}{(2m)!}
   \sum_{j_1,\dots,j_{2m}=1}^L
   \frac{\epsilon(J_{[2m]})\pf(\mathbf{A}_{[J^c_{[2m]}]})}{\pf(\mathbf{A})}
   \Biggl(\prod_{\alpha=1}^{2m}\psi_{j_\alpha}\Biggr) \\
&\quad + \frac{(-1)^m}{(2m-1)!}
   \sum_{j_0,j_1,\dots,j_{2m-1}=1}^L
   \frac{\epsilon(J_{[0,2m-1]})\pf(\mathbf{A}_{[J^c_{[0,2m-1]}]})}{\pf(\mathbf{A})}
   \Biggl(u_{j_0}\prod_{\alpha=1}^{2m-1}\psi_{j_\alpha}\Biggr)
\Bigg]\Bigg\}.
\end{split}
\end{equation*}

Using Jacobi’s identity to express the inverse of $\mathbf{A}$, we obtain:  
\begin{equation*}
    \begin{split}
    (-1)^m \frac{\epsilon(J_{[2m]})\pf(\mathbf{A}_{[J^c_{[2m]}]})}{\pf(\mathbf{A})}
    &= (2m-1)!! \prod_{\alpha=1}^m (\mathbf{A}^{-1})_{j_{2\alpha-1},j_{2\alpha}}, \\[0.5em]
    (-1)^m \frac{\epsilon(J_{[0,2m-1]})\pf(\mathbf{A}_{[J^c_{[0,2m-1]}]})}{\pf(\mathbf{A})}
    &= (2m-1)!! \prod_{\alpha=0}^{m-1} (\mathbf{A}^{-1})_{j_{2\alpha},j_{2\alpha+1}}.
    \end{split}
\end{equation*}
This can be rewritten in the compact form  
\[
\pf\!\left((\mathbf{A}^{-1})^T_{[J]}\right)
= \frac{\epsilon(J)\,\pf(\mathbf{A}_{[J^c]})}{\pf(\mathbf{A})}.
\]

---

We first focus on the purely fermionic contribution:
\begin{equation*}
\begin{split}
  \hspace{-1cm} 1 + \sum_{m=1}^{L/2}
   \frac{(-1)^m}{(2m)!}
   \sum_{j_1,\dots,j_{2m}=1}^L
   \frac{\epsilon(J_{[2m]})\pf(\mathbf{A}_{[J^c_{[2m]}]})}{\pf(\mathbf{A})}
   \prod_{\alpha=1}^{2m}\psi_{j_\alpha}
   &= 1 + \sum_{m=1}^{L/2}
   \frac{(2m-1)!!}{(2m)!}
   \sum_{j_1,\dots,j_{2m}=1}^L
   \prod_{\alpha=1}^m (\mathbf{A}^{-1})_{j_{2\alpha-1},j_{2\alpha}}\psi_{j_{2\alpha-1}}\psi_{j_{2\alpha}} \\
   &= 1 + \sum_{m=1}^{L/2}
   \frac{1}{m!\,2^m}
   \Biggl(\sum_{i,j=1}^L (\mathbf{A}^{-1})_{ij}\psi_i\psi_j\Biggr)^m \\
   &= 1 + \sum_{m=1}^{L/2}\frac{1}{m!}
   \Bigl(\tfrac{1}{2}\boldsymbol{\psi}^T\mathbf{A}^{-1}\boldsymbol{\psi}\Bigr)^m \\
   &= e^{\tfrac{1}{2}\boldsymbol{\psi}^T\mathbf{A}^{-1}\boldsymbol{\psi}}.
\end{split}
\end{equation*}

---

Next, consider the mixed bosonic–fermionic contribution:
\begin{equation*}
\begin{split}
   &\sum_{m=1}^{L/2}\frac{(-1)^m}{(2m-1)!}
   \sum_{j_0,j_1,\dots,j_{2m-1}=1}^L
   \frac{\epsilon(J_{[0,2m-1]})\pf(\mathbf{A}_{[J^c_{[0,2m-1]}]})}{\pf(\mathbf{A})}
   \Biggl(u_{j_0}\prod_{\alpha=1}^{2m-1}\psi_{j_\alpha}\Biggr) \\
   &= \sum_{m=1}^{L/2}\frac{(2m-1)!!}{(2m-1)!}
   \sum_{j_0,\dots,j_{2m-1}=1}^L
   \prod_{\alpha=0}^{m-1}(\mathbf{A}^{-1})_{j_{2\alpha},j_{2\alpha+1}}
   \Biggl(u_{j_0}\prod_{\alpha=1}^{2m-1}\psi_{j_\alpha}\Biggr) \\
   &= \mathbf{u}^T\mathbf{A}^{-1}\boldsymbol{\psi} 
   + \mathbf{u}^T\mathbf{A}^{-1}\boldsymbol{\psi}
     \sum_{m=2}^{L/2}\frac{1}{(m-1)!\,2^{m-1}}
     \sum_{j_2,\dots,j_{2m-1}=1}^L
     \prod_{\alpha=1}^{m-1}(\mathbf{A}^{-1})_{j_{2\alpha},j_{2\alpha+1}}\psi_{j_{2\alpha}}\psi_{j_{2\alpha+1}}.
\end{split}
\end{equation*}

By shifting the summation index $m-1=n$ and completing the missing ($n=L/2$) term, we obtain:
\begin{equation*}
\mathbf{u}^T\mathbf{A}^{-1}\boldsymbol{\psi}\sum_{n=0}^{L/2}\frac{1}{n!\,2^n}
\Biggl(\sum_{i,j=1}^L (\mathbf{A}^{-1})_{ij}\psi_i\psi_j\Biggr)^n
= \mathbf{u}^T\mathbf{A}^{-1}\boldsymbol{\psi}\,
e^{\tfrac{1}{2}\boldsymbol{\psi}^T\mathbf{A}^{-1}\boldsymbol{\psi}}.
\end{equation*}

---

Combining both contributions, we conclude:
\[
\int \mathbf{D}\boldsymbol{\chi}\,
e^{\tfrac{1}{2}\boldsymbol{\chi}^T\mathbf{A}\boldsymbol{\chi}
+ \mathbf{u}^T\boldsymbol{\chi} + \boldsymbol{\psi}^T\boldsymbol{\chi}}
= \pf(\mathbf{A})\,
e^{\tfrac{1}{2}\boldsymbol{\psi}^T\mathbf{A}^{-1}\boldsymbol{\psi}
+ \mathbf{u}^T\mathbf{A}^{-1}\boldsymbol{\psi}}.
\]

\subsection{Proof of Corollary \ref{pure fermionic sources}}\label{Proof pure fermionic sources}
\begin{proof}
Expanding the exponential, all terms vanish under Grassmann integration except the highest-order term:
\[
\int \mathbf{D}\boldsymbol{\chi}\, e^{\boldsymbol{\psi}\mathbf{A}\boldsymbol{\chi}}
= \frac{1}{L!} \int \mathbf{D}\boldsymbol{\chi}\, \bigl(\boldsymbol{\psi}\mathbf{A}\boldsymbol{\chi}\bigr)^L.
\]
Explicitly,
\begin{equation*}
\begin{split}
(\boldsymbol{\psi}\mathbf{A}\boldsymbol{\chi})^L
&=\Biggl[\sum_{i=1}^N \sum_{j=1}^L \psi_i A_{ij}\chi_j \Biggr]^L \\
&= \sum_{i_1,\dots,i_L=1}^N \sum_{j_1,\dots,j_L=1}^L 
   \prod_{\alpha=1}^L \bigl(\psi_{i_\alpha} A_{i_\alpha j_\alpha}\chi_{j_\alpha}\bigr) \\
&= (-1)^{\tfrac{L(L+1)}{2}}
   \sum_{i_1,\dots,i_L=1}^N \sum_{j_1,\dots,j_L=1}^L 
   \Bigl(\prod_{\beta=1}^L \chi_{j_\beta}\Bigr) 
   \Bigl(\prod_{\alpha=1}^L \psi_{i_\alpha} A_{i_\alpha j_\alpha}\Bigr),
\end{split}
\end{equation*}
where the sign factor accounts for the reordering of Grassmann variables. Performing the Grassmann integration gives
\begin{equation*}
\begin{split}
\int \mathbf{D}\boldsymbol{\chi}\, e^{\boldsymbol{\psi}\mathbf{A}\boldsymbol{\chi}}
&=\frac{(-1)^{\tfrac{L(L+1)}{2}}}{L!}
  \sum_{i_1,\dots,i_L=1}^N \sum_{j_1,\dots,j_L=1}^L
  \varepsilon^{j_1 \cdots j_L} 
  A_{i_1 j_1} \cdots A_{i_L j_L} 
  \prod_{\alpha=1}^L \psi_{i_\alpha},
\end{split}
\end{equation*}
where $\varepsilon^{j_1 \cdots j_L}$ is the Levi-Civita symbol. Recognizing the determinant formula, we obtain
\[
\int \mathbf{D}\boldsymbol{\chi}\, e^{\boldsymbol{\psi}\mathbf{A}\boldsymbol{\chi}}
= (-1)^{\tfrac{L(L+1)}{2}}
  \sum_{i_1<\dots<i_L=1}^N 
  \det\bigl(\mathbf{A}_{[I_{[L]},\star]}\bigr)\,
  \prod_{\alpha=1}^L \psi_{i_\alpha}.
\]
When $N=L$, the sum yields $L!$ identical terms, recovering the special case.
\end{proof}

\subsection{Proof of Lemma \ref{generalized pure fermionic sources}}\label{proof generalized pure fermionic sources}

\begin{proof}To prove this, by defining the new variables, we have:
\[
\int \mathbf{D}\boldsymbol{\chi} e^{\boldsymbol{\psi}^{(1)}\mathbf{A}^{(1)}\boldsymbol{\chi}+\boldsymbol{\psi}^{(2)}\mathbf{A}^{(2)}\boldsymbol{\chi}+\cdots+\boldsymbol{\psi}^{(N)}\mathbf{A}^{(N)}\boldsymbol{\chi}}
= \int \mathbf{D}\boldsymbol{\chi} \, e^{\boldsymbol{\Psi} \boldsymbol{\mathbb{A}} \boldsymbol{\chi}},
\]
where
\[
\boldsymbol{\Psi} = \left( \boldsymbol{\psi}^{(1)}, \boldsymbol{\psi}^{(2)}, \dots, \boldsymbol{\psi}^{(N)} \right), 
\quad 
\mathbb{A} = \begin{pmatrix}
\mathbf{A}^{(1)} \\ \mathbf{A}^{(2)} \\ \vdots \\ \mathbf{A}^{(N)}
\end{pmatrix}.
\]
Then, by using Corollary \ref{pure fermionic sources} we have:

\[
\hspace{-2cm}\int \mathbf{D}\boldsymbol{\chi} \, e^{\boldsymbol{\Psi} \boldsymbol{\mathbb{A}} \boldsymbol{\chi}}
= (-1)^{\tfrac{m(m+1)}{2}}\sum_{i_1^{(1)}<\dots<i_m^{(1)}=1}^{n^{(1)}}\cdots\sum_{i_1^{(N)}<\dots<i_m^{(N)}=1}^{n^{(N)}}\det\left(\mathbb{A}_{\left[\left\{i_1^{(1)},\dots,i_m^{(1)},\dots,i_1^{(N)},\dots,i_m^{(N)}\right\},\star\right]}\right)\psi_{i_1^{(1)}}^{(1)}\cdots\psi_{i_m^{(1)}}^{(1)}\cdots\psi_{i_1^{(N)}}^{(N)}\cdots\psi_{i_m^{(N)}}^{(N)}.
\]
\end{proof}
\subsection{Proof of Theorem \ref{theorem 2.2}}\label{prooftheorem 2.2}
\begin{proof}
(a) To prove the result, we first decompose the exponential into four distinct exponentials. 
The quadratic term \(\bar{\bm{\chi}}^{T}A\bm{\chi}\) commutes with every other contribution, 
and the linear Grassmann bilinears \(\bar{\bm{\psi}}^{T}\bm{\chi}\) and \(\bar{\bm{\chi}}^{T}\bm{\psi}\) 
also commute with one another, since they involve only Grassmann variables and therefore generate 
mutually commuting nilpotent elements. 
In contrast, the source terms \(\bar{\bm{u}}^{T}\bm{\chi}\) and \(\bar{\bm{\chi}}^{T}\bm{u}\) 
cannot be disentangled: they mix ordinary complex with Grassmann variables, so their exponentials fail to commute and must be retained as a single combined factor. After carrying out this separation, we expand each linear exponential, then we obtain:
\begin{align}
\int \mathbf{D}(\boldsymbol{\chi},\bar{\boldsymbol{\chi}}) e^{\bar{\boldsymbol{\chi}}^T \mathbf{A}\boldsymbol{\chi}+\bar{\boldsymbol{\psi}}^T\boldsymbol{\chi}+\bar{\boldsymbol{\chi}}^T \boldsymbol{\psi}+\bar{\mathbf{u}}^T\boldsymbol{\chi}+\bar{\boldsymbol{\chi}}^T \mathbf{u}}
&= (-1)^{\frac{L(L-1)}{2}}
   \int \mathbf{D}\bm{\chi}\,\mathbf{D}\bar{\bm{\chi}}
   \big(1 + \bar{\bm{u}}^{T}\bm{\chi} + \bar{\bm{\chi}}^{T}\bm{u}\big)\notag\\
&\qquad\qquad\times
   \Big(1 + \sum_{m=1}^{L}\frac{1}{m!}(\bar{\bm{\psi}}^{T}\bm{\chi})^{m}\Big)
   \Big(1 + \sum_{n=1}^{L}\frac{1}{n!}(\bar{\bm{\chi}}^{T}\bm{\psi})^{n}\Big)
   \; e^{\bar{\bm{\chi}}^{T}A\bm{\chi}}.
\end{align}
Multiplying out the three series and doing some simplification, we have:
\begin{align}
\int \mathbf{D}(\boldsymbol{\chi},\bar{\boldsymbol{\chi}})&e^{\bar{\boldsymbol{\chi}}^T \mathbf{A}\boldsymbol{\chi}+\bar{\boldsymbol{\psi}}^T\boldsymbol{\chi}+\bar{\boldsymbol{\chi}}^T \boldsymbol{\psi}+\bar{\mathbf{u}}^T\boldsymbol{\chi}+\bar{\boldsymbol{\chi}}^T \mathbf{u}}\\
&=(-1)^{\frac{L(L-1)}{2}}
   \int \mathbf{D}\bm{\chi}\,\mathbf{D}\bar{\bm{\chi}} \Big\{ 1
       + \sum_{m=1}^{L}\frac{1}{m!}(\bar{\bm{\psi}}^{T}\bm{\chi})^{m}
       + \sum_{n=1}^{L}\frac{1}{n!}(\bar{\bm{\chi}}^{T}\bm{\psi})^{n}\notag\\
       &
       + \sum_{m=1}^{L}\sum_{n=1}^{L}\frac{1}{m!\,n!}
         (\bar{\bm{\psi}}^{T}\bm{\chi})^{m}(\bar{\bm{\chi}}^{T}\bm{\psi})^{n}
       + \sum_{j_0=1}^L \bar{u}_{j_0}{\chi}_{j_0}
       + \sum_{j_0=1}^L \bar{u}_{j_0}{\chi}_{j_0}\sum_{m=1}^{L}\frac{1}{m!}(\bar{\bm{\psi}}^{T}\bm{\chi})^{m} \notag\\
&
       + \sum_{j_0=1}^L \bar{u}_{j_0}{\chi}_{j_0}\sum_{n=1}^{L}\frac{1}{n!}(\bar{\bm{\chi}}^{T}\bm{\psi})^{n}
       + \sum_{j_0=1}^L \bar{u}_{j_0}{\chi}_{j_0}\sum_{m=1}^{L}\sum_{n=1}^{L}\frac{1}{m!\,n!}
         (\bar{\bm{\psi}}^{T}\bm{\chi})^{m}(\bar{\bm{\chi}}^{T}\bm{\psi})^{n}
\notag\\
&
       + \sum_{i_0=1}^L {u}_{i_0}\bar{\chi}_{i_0}
       + \sum_{i_0=1}^L {u}_{i_0}\bar{\chi}_{i_0}\sum_{m=1}^{L}\frac{1}{m!}(\bar{\bm{\psi}}^{T}\bm{\chi})^{m}
       + \sum_{i_0=1}^L {u}_{i_0}\bar{\chi}_{i_0}\sum_{n=1}^{L}\frac{1}{n!}(\bar{\bm{\chi}}^{T}\bm{\psi})^{n}
\notag\\
&
       + \sum_{i_0=1}^L {u}_{i_0}\bar{\chi}_{i_0}\sum_{m=1}^{L}\sum_{n=1}^{L}\frac{1}{m!\,n!}
         (\bar{\bm{\psi}}^{T}\bm{\chi})^{m}(\bar{\bm{\chi}}^{T}\bm{\psi})^{n}
\Big\} \; e^{\bar{\bm{\chi}}^{T}A\bm{\chi}}.
\end{align}
The Grassmann integral is nonvanishing only when the integrand contains the pair Grassmann degree, i.e., a product involving \(\chi_k\) and \(\bar{\chi}_k\) for \(k=1,\dots,L\).  
Consequently, in the expansion above, this implies that the second, third, fifth, sixth, ninth, and eleventh contributions vanish identically. Expanding the retained contributions and simplifying the resulting expressions yields:
\begin{align}
&\int \mathbf{D}(\boldsymbol{\chi},\bar{\boldsymbol{\chi}}) e^{\bar{\boldsymbol{\chi}}^T \mathbf{A}\boldsymbol{\chi}+\bar{\boldsymbol{\psi}}^T\boldsymbol{\chi}+\bar{\boldsymbol{\chi}}^T \boldsymbol{\psi}+\bar{\mathbf{u}}^T\boldsymbol{\chi}+\bar{\boldsymbol{\chi}}^T \mathbf{u}}\\&=(-1)^{\tfrac{L(L-1)}{2}}
   \int \bold{D}\bm{\chi}\, \bold{D}\bar{\bm{\chi}} \Bigg[1+ \;
   \sum_{m,n= 1}^L
   \frac{(-1)^{m}}{m!\,n!}
   \sum_{j_{1},\dots,j_{m}=1}^{L}
   \sum_{i_{1},\dots,i_{n}=1}^{L}
     \prod_{\beta=1}^{n} \bar{\chi}_{i_\beta} \prod_{\alpha=1}^{m}\chi_{j_\alpha}
     \prod_{\alpha=1}^{m} \bar{\psi}_{j_\alpha}
     \prod_{\beta=1}^{n}
     \psi_{i_\beta}
\notag\\
&\quad-\sum_{j_0=1}^L \bar{u}_{j_0} \sum_{i_1=1}^L \bar{\chi}_{i_1}\chi_{j_0} \psi_{i_1} + \sum_{m,n= 1}^L
   \frac{(-1)^{m}}{m!\,n!}
   \sum_{j_0,j_{1},\dots,j_{m-1}=1}^{L}
   \sum_{i_{1},\dots,i_{n}=1}^{L}
   \bar{u}_{j_0}
     \prod_{\beta=1}^{n} \bar{\chi}_{i_\beta} \prod_{\alpha=0}^{m-1}\chi_{j_\alpha}
     \prod_{\alpha=1}^{m-1} \bar{\psi}_{j_\alpha}
     \prod_{\beta=1}^{n}
     \psi_{i_\beta}\notag\\
     &\quad-\sum_{i_0=1}^L \bar{u}_{i_0} \sum_{j_1=1}^L \bar{\chi}_{i_0}\chi_{j_1} \bar{\psi}_{j_1} + \sum_{m,n= 1}^L
   \frac{(-1)^{m}}{m!\,n!}
   \sum_{j_{1},\dots,j_{m}=1}^{L}
   \sum_{i_0,i_{1},\dots,i_{n-1}=1}^{L}
   {u}_{i_0}
     \prod_{\beta=0}^{n-1} \bar{\chi}_{i_\beta} \prod_{\alpha=1}^{m}\chi_{j_\alpha}
     \prod_{\alpha=1}^{m} \bar{\psi}_{j_\alpha}
     \prod_{\beta=1}^{n-1}
     \psi_{i_\beta}
   \,\Bigg] e^{\bar{\bm{\chi}}^{T}A\bm{\chi}}.
\end{align}
By explicitly performing the Grassmann integrations, the expression simplifies to the following closed
form:
\begin{align}
&\int \mathbf{D}(\boldsymbol{\chi},\bar{\boldsymbol{\chi}}) e^{\bar{\boldsymbol{\chi}}^T \mathbf{A}\boldsymbol{\chi}+\bar{\boldsymbol{\psi}}^T\boldsymbol{\chi}+\bar{\boldsymbol{\chi}}^T \boldsymbol{\psi}+\bar{\mathbf{u}}^T\boldsymbol{\chi}+\bar{\boldsymbol{\chi}}^T \mathbf{u}}\\&=\det[\mathbf{A}]+
   \sum_{m= 1}^L
   \frac{(-1)^{m}}{(m!)^2}
   \sum_{j_{1},\dots,j_{m}=1}^{L}
   \sum_{i_{1},\dots,i_{m}=1}^{L}
   \epsilon(I_{[m]},J_{[m]}) \det[\mathbf{A}_{[I_{[m]}^c,J_{[m]}^c]}]
     \prod_{\alpha=1}^{m} \bar{\psi}_{j_\alpha}
     \prod_{\beta=1}^{m}
     \psi_{i_\beta}\notag\\
&
 -\sum_{j_0=1}^L \sum_{i_1=1}^L \epsilon(I_{[1]},J_{[0]}) \det[\mathbf{A}_{[I_{[1]}^c,J_{[0]}^c]}] \bar{u}_{j_0} \psi_{i_1} -\sum_{i_0=1}^L \sum_{j_1=1}^L \epsilon(I_{[0]},J_{[1]}) \det[\mathbf{A}_{[I_{[0]}^c,J_{[1]}^c]}] {u}_{i_0} \bar{\psi}_{j_1}\notag\\
& + \sum_{m= 1}^L
   \frac{(-1)^{m}}{m!(m-1)!}
   \sum_{j_0,j_{1},\dots,j_{m-1}=1}^{L}
   \sum_{i_{1},\dots,i_{m}=1}^{L}
   \epsilon(I_{[m]},J_{[0,m-1]}) \det[\mathbf{A}_{[I_{[m]}^c,J_{[0,m-1]}^c]}] \bar{u}_{j_0}
     \prod_{\alpha=1}^{m-1} \bar{\psi}_{j_\alpha}
     \prod_{\beta=1}^{m}
     \psi_{i_\beta}\notag\\
     &+ \sum_{m= 1}^L
   \frac{(-1)^{m}}{m!(m-1)!}
   \sum_{j_{1},\dots,j_{m}=1}^{L}
   \sum_{i_0,i_{1},\dots,i_{m-1}=1}^{L}
   \epsilon(I_{[0,m-1]},J_{[m]}) \det[\mathbf{A}_{[I_{[0,m-1]}^c,J_{[m]}^c]}] u_{i_0}
     \prod_{\alpha=1}^{m} \bar{\psi}_{j_\alpha}
     \prod_{\beta=1}^{m-1}
     \psi_{i_\beta}.
\end{align}
To simplify this expression, we consider the third and fourth terms and rewrite them as follows:
\begin{equation}
    \begin{split}
    \sum_{j_0=1}^L \sum_{i_1=1}^L (-1)^{j_0+i_1} \det\left(\mathbf{A}_{[\{i_1,j_0\}^c]}\right) \bar{u}_{j_0}\psi_{i_1}= \sum_{k=2}^{L+1} \sum_{i_1=1}^L (-1)^{k+i_1-1} \det\left(\mathbf{A}_{[\{i_1,k-1\}^c]}\right) \bar{u}_{k-1}\psi_{i_1},\\
    \sum_{i_0=1}^L \sum_{j_1=1}^L (-1)^{i_0+j_1} \det\left(\mathbf{A}_{[\{i_0,j_1\}^c]}\right) {u}_{i_0}\bar{\psi}_{j_1}= \sum_{l=2}^{L+1} \sum_{j_1=1}^L (-1)^{l+j_1-1} \det\left(\mathbf{A}_{[\{k-1,j_1\}^c]}\right) {u}_{k-1}\bar{\psi}_{i_1},
    \end{split}
\end{equation}
where $k\equiv j_0+1$ and $l\equiv i_0+1$. We define the following matrices:
\begin{equation}
   \bar{\mathbf{Y}}=\begin{pmatrix} 0 & \bar{\mathbf{u}}^T \\ \boldsymbol{\psi} & \mathbf{A} \end{pmatrix}, \quad  \mathbf{Y}=\begin{pmatrix} 0 & \boldsymbol{\bar{\psi}}^T \\ \mathbf{u} & \mathbf{A} \end{pmatrix}.
\end{equation}
From the structure of the matrices $\bar{\mathbf{Y}}$ and $\mathbf{Y}$ it follows that:
\begin{equation}
    \begin{cases}
        \bar{u}_{k-1}=\bar{Y}_{1k},\\
        {u}_{l-1}={Y}_{l1},\\
        \sum\limits_{i_1=1}^L \mathbf{A}_{\{i_1,k-1\}}\psi_{i_1}=\bar{\mathbf{M}}_{\{1,k\}},\\
        \sum\limits_{j_1=1}^L \mathbf{A}_{\{l-1,j_1\}}\bar{\psi}_{j_1}=\mathbf{M}_{\{l,1\}},
    \end{cases} 
\end{equation}
where $\bar{\mathbf{M}}_{\{1,k\}}$ and $\mathbf{M}_{\{l,1\}}$ are the minors of the matrices $\bar{\mathbf{Y}}$ and $\mathbf{Y}$. Then we obtain:
\begin{equation}
\begin{cases}
    \sum\limits_{k=2}^L(-1)^k \det\left(\mathbf{M}_{[\{1,k\}^c]}\right) \bar{Y}_{1k}=\det{\bar{\mathbf{Y}}},\\
    \sum\limits_{l=2}^L(-1)^l \det\left(\mathbf{M}_{[\{l,1\}^c]}\right) {Y}_{l1}=\det{{\mathbf{Y}}}.
\end{cases}
\end{equation} 
Finally, one can simplify the expression as:
\begin{align}
&\int \mathbf{D}(\boldsymbol{\chi},\bar{\boldsymbol{\chi}}) e^{\bar{\boldsymbol{\chi}}^T \mathbf{A}\boldsymbol{\chi}+\bar{\boldsymbol{\psi}}^T\boldsymbol{\chi}+\bar{\boldsymbol{\chi}}^T \boldsymbol{\psi}+\bar{\mathbf{u}}^T\boldsymbol{\chi}+\bar{\boldsymbol{\chi}}^T \mathbf{u}}\\&=\det[\mathbf{A}]+
   \sum_{m= 1}^L
   \frac{(-1)^{m}}{(m!)^2}
   \sum_{j_{1},\dots,j_{m}=1}^{L}
   \sum_{i_{1},\dots,i_{m}=1}^{L}
   \epsilon(I_{[m]},J_{[m]}) \det[\mathbf{A}_{[I_{[m]}^c,J_{[m]}^c]}]
     \prod_{\alpha=1}^{m} \bar{\psi}_{j_\alpha}
     \prod_{\beta=1}^{m}
     \psi_{i_\beta}\notag\\
&
 -\det[\bar{\mathbf{Y}}] - \sum_{m= 1}^L
   \frac{1}{m!(m-1)!}
   \sum_{j_{1},\dots,j_{m-1}=1}^{L}
   \sum_{i_{2},\dots,i_{m}=1}^{L}
   \epsilon(I_{[m]}\setminus \{i_1\},J_{[m-1]}) \det[\bar{\mathbf{Y}}_{[(I_{[m]}\setminus \{i_1\})^c,J_{[m-1]}^c]}]
     \prod_{\alpha=1}^{m-1} \bar{\psi}_{j_\alpha}
     \prod_{\beta=2}^{m}
     \psi_{i_\beta} \notag\\
     &-\det[\mathbf{Y}]+ \sum_{m= 1}^L
   \frac{(-1)^{m}}{m!(m-1)!}
   \sum_{j_{2},\dots,j_{m}=1}^{L}
   \sum_{i_{1},\dots,i_{m-1}=1}^{L}
   \epsilon(I_{[m-1]},J_{[m]}\setminus \{j_1\}) \det[\mathbf{Y}_{[I_{[m-1]}^c,(J_{[m]}\setminus \{j_1\})^c]}]
     \prod_{\alpha=2}^{m} \bar{\psi}_{j_\alpha}
     \prod_{\beta=1}^{m-1}
     \psi_{i_\beta}.
\end{align}
\end{proof}

\end{document}